\documentclass[11pt]{article}
\usepackage[a4paper, left=1in, right=1in, top=1in, bottom=1in]{geometry}
\usepackage{authblk}
\usepackage{graphicx}
\usepackage{amsfonts}
\usepackage{amssymb}
\usepackage{physics}
\usepackage{cancel}
\usepackage{hyperref}
\usepackage{float}
\usepackage{tikz}
\usepackage{quantikz}
\usepackage{enumitem}
\usepackage{mathtools}
\usepackage{amsthm}
\usepackage{booktabs}
\usepackage[dvipsnames]{xcolor}

\newtheorem{theorem}{Theorem}

\newtheorem{definition}[theorem]{Definition}

\newtheorem{lemma}[theorem]{Lemma}

\newtheorem{proposition}[theorem]{Proposition}

\allowdisplaybreaks

\usepackage{cleveref} 
\crefname{section}{Section}{Sections} 
\crefname{theorem}{Theorem}{Theorems} 
\crefname{definition}{Definition}{Definitions}
\crefname{proposition}{Proposition}{Propositions} 
\crefname{corrollary}{Corrollary}{Corrollaries} 
\crefname{remark}{Remark}{Remarks} 
\crefname{conjecture}{Conjecture}{Conjectures} 
\crefname{example}{Example}{Examples} 
\crefname{equation}{Eq.}{Eqs.} 

\newcommand{\rledit}[1]{{#1}}
\newcommand{\rleditb}[1]{{#1}}
\newcommand{\rleditc}[1]{{#1}}

\newcommand{\hide}[1]{{}}
\newcommand{\edited}[1]{{#1}}

\newcommand{\argline}{\underline{\hspace{0.3cm}}}

\newcommand{\hilb}[1]{\mathcal{#1}} 
\newcommand{\operators}[1]{L(\hilb{#1})}
\newcommand{\Herm}[1]{Herm(\hilb{{#1}})}
\newcommand{\states}[1]{S(\hilb{#1})}
\newcommand{\idop}{id} 
\newcommand{\idchan}{\mathcal{I}} 
\newcommand{\linmaps}[2]{L({#1},{#2})}
\newcommand{\channels}[2]{C({#1}, {#2})}
\newcommand{\cpmaps}[2]{CP({#1}, {#2})}

\newcommand{\visible}{observed}  
\newcommand{\hidden}{latent} 

\newcommand{\sampleindex}[2]{{#1}^{(#2)}}  
\newcommand{\sampleindexthree}[3]{{#1}_{#2}^{(#3)}}  
\newcommand{\decompindex}[2]{{#1}_{#2}}  

\newcommand{\sot}{SOT} 
\newcommand{\sotlong}{weak state-over-time}

\newcommand{\sotmap}{\odot}
\newcommand{\starprod}{\star}

\newcommand{\ASset}[3]{AS^{#1}_{#2,#3}} 
\newcommand{\ARCset}[4]{AC^{#1; #4}_{#2,#3}} 
\newcommand{\ARCsetpos}[4]{\widetilde{AC}^{#1; #4}_{#2,#3}}

\newcommand{\infmap}{extended inference map}
\newcommand{\infchan}{inference channel}
\newcommand{\infchansymbol}{\mathcal{E}}
\newcommand{\genmap}{extended generation map}
\newcommand{\genchan}{generation channel}
\newcommand{\genchansymbol}{\mathcal{D}}  

\newcommand{\dephchannel}{\mathcal{K}}

\newcommand{\LSodot}{\odot^{LS}}
\newcommand{\LSshort}{LS}
\newcommand{\JPodot}{\odot^{JP}}
\newcommand{\JPshort}{JP}  
		
\newcommand{\MPcapital}{Operator-product}	
\newcommand{\MPodot}{\odot^{OP}}
\newcommand{\MPshort}{OP}
\newcommand{\rMPodot}{\odot^{PO}}
\newcommand{\rMPshort}{PO}

\newcommand{\LSCPmap}[1]{\mathcal{E}_{#1}^{\LSshort}}
\newcommand{\JPCPmap}[1]{\mathcal{E}_{#1}^{\JPshort}}
\newcommand{\MPCPmap}[1]{\mathcal{E}_{#1}^{\MPshort}}

\newcommand{\CXPPT}{CX-PPT}

\newcommand{\DXPPT}{DX-PPT}

\newcommand{\scale}{1cm}
\newcommand{\spacecolor}{darkgray}

\newcommand{\gendataext}{generalised extension map}
\newcommand{\corr}{corr}

\tikzstyle{whitedott}=[circle, draw=black, fill=white, inner sep=.4ex]

\newcommand{\tinycomult}[1][whitedott]{
\smash{\raisebox{-2pt}{\hspace{-5pt}\ensuremath{\begin{aligned}\begin{tikzpicture}[scale=2.0, font=\tiny,scale=0.25,yscale=0.9]
    \node (0) at (0,0) {};
    \node[#1, inner sep=1.5pt] (1) at (0,0.55) {};
    \node (2) at (-0.5,1) {};
    \node (3) at (0.5,1) {};
    \draw (0.center) to (1.center);
    \draw (1.center) to [out=left, in=down, out looseness=1.5] (2.center);
    \draw (1.center) to [out=right, in=down, out looseness=1.5] (3.center);
    \node[#1, inner sep=1.5pt] (1) at (0,0.55) {};
\end{tikzpicture}\end{aligned}
}\hspace{-3pt}}}}

\newcommand{\tinycopy}{\tinycomult[whitedott]}

\title{Towards unsupervised representation learning for quantum data:\\quantum models with inference and generation}

\author{Robin Lorenz\thanks{robin.lorenz@quantinuum.com}} 
\author{Eric Brunner}
\author{Marcello Benedetti\thanks{marcello.benedetti@quantinuum.com}}

\affil{Quantinuum, London, United Kingdom}

\date{September 1, 2026}

\begin{document}

\maketitle	

\begin{abstract}
	With quantum sensors, simulators and networks emerging, a future of quantum technology may produce quantum states as data---that is, coherently rather than as classical measurement records---thus motivating the study of suitable quantum generalisations of modern machine learning, including the automated, unsupervised extraction of useful representations. Two ingredients are central to the latter: inference, mapping observations to latent representations, and generation, mapping latent states back to synthetic data. Both are related to each other and to joint distributions for training models by the chain-rule of classical probability theory. The fact that quantum states however lack such universal, standard factorisation property thus poses a challenge. 
	
	Here we develop a conceptual and mathematical framework for unsupervised representation learning from quantum data. Models are joint quantum states over visible and latent systems; state-over-time maps provide a notion of factorisation into a marginal state and inference (generation) channel; models with inference (generation) are ambiguous states---states for which such factorisation obtains---subject to a further consistency condition on extended inference maps as data extension. These stipulations are restrictive: we show that non-trivial models must feature non-linear such maps to the extended space. For three representative state-over-time maps, we completely characterise the ambiguous states, uncovering a hierarchy tied to the positive-partial-transpose (PPT) criterion from entanglement theory. Notably, the Leifer-Spekkens construction supports inference and generation exactly for model classes of PPT states, thus allowing genuinely quantum visible-latent correlations. We also formulate quantum counterparts of exact and approximate inference training, explore weaker notions of data extension and sketch a future research programme.
\end{abstract}

\section{Introduction} \label{sec:introduction}

Real-world data such as text, image, audio and video data tend to be high-dimensional, multi-modal and heterogeneous. 
Machine learning (ML) models take as input numerical representations of such raw data and are trained to solve tasks such as classification, clustering, anomaly detection and next-token prediction, to name a few. 
The knowledge stored in the model is expected to transfer across multiple tasks and be relevant even when the distribution of the input changes.      
The key problem of extracting useful features from raw data is what the field of \emph{representation learning} studies~\cite{Bengio2013}. 
\emph{Deep learning}, which has been the predominant paradigm for representation learning since the early 2010s, has unlocked tremendous progress in many areas of application~\cite{LeCun2015}. 
Deep learning emphasises the importance of extracting a hierarchy of increasingly abstract representations from raw data. Ideally, the costly process of sifting through the noise of raw data to extract this hierarchy of features should be automated, i.e. the learning done in an \emph{unsupervised} way. For example, a modern take on this consists of constructing pretext tasks on unlabelled data for the model to solve~\cite{Jing2021, Jaiswal2021, Uelwer2025}.   

All this progress concerns a classical world where the real-world data is of a classical nature as are the methods, the models, the latent spaces for representation, and the hardware which models are run on. Usually the underlying problem or phenomenon is fundamentally classical, but even for data that arises from quantum systems and to which machine learning has been applied, such as particle detection in high-energy physics~\cite{Astrand2026,Aad2026} or protein folding in biochemistry~\cite{Noe2020,Abramson2024}, any data fed into models is, roughly speaking, classical records of measurement outcomes.

However, it is fascinating to entertain---and increasingly reasonable to anticipate---a \emph{future of quantum technology}, in which quantum states are acquired by quantum sensors and quantum networks, prepared by programmable quantum simulators, and retained or transmitted through quantum memories or a quantum internet 
\edited{\cite{degen2017quantum, zaiser2016enhancing, khabiboulline2019optical, stas2026entanglement, altman2021quantum, ebadi2021quantum, wehner2018quantum, knaut2024entanglement}. Such platforms have already motivated learning tasks on coherently available quantum data, including quantum-state and phase classification, quantum principal-component analysis, and learning properties or dynamics of physical systems \cite{herrmann2022realizing, huang2022quantum}.} 
Though it is natural to ask more generally: what machine learning models could be deployed in such an environment? 
This requires a notion of \emph{quantum models} that generalise classical models in that they take raw \emph{quantum states as input} and process them in a fully quantum way---that is, \emph{coherently on a quantum computer}. 
Moreover, suppose that the general high-level goal would still be to learn correlations in the data so as to perform a wide range of tasks such as classification, clustering and so on, though now on quantum data. Then, in light of the lessons from classical representation learning, one would expect that suitably generalised methods for representation learning may well again lead to significant improvements across many domains. 
We thus propose the study of \emph{unsupervised representation learning for quantum data} and initiate its study in a principled manner.  

Tracing the trajectory of classical breakthroughs in representation learning, two topics stand out as central: 
\emph{inference} and \emph{generation}. 
The first refers to a model's ability to predict and reason in a manner that is consistent with data; 
depending on the context, the `inference map' is also called the encoder, used to represent data internally. 
The second refers to a model's ability to produce meaningful synthetic data; 
depending on the context, the `generation map' is also called a decoder, used to decode internal representations. 
Both are stochastic maps---the first from some visible space $X$ to some latent space $Y$, the second in the converse direction---subject to the condition of playing well with each other, i.e. being compatible with some suitable joint probability distribution on $X$ and $Y$, which is what typically enters some objective function during training.  

All these aspects---inference and generation, as well as their interplay---will need \rledit{to be generalised}. 
\rledit{As is well known, and as we will discuss in more depth later, in quantum theory there is no obvious, standard analogue to the chain rule of probability theory, yet the latter is what underlies the classical constructions alluded to above.} 
So how might one then go about a quantum generalisation? 
In particular, how does one factorise quantum states into something like an inference or generation map and a marginal state, or conversely, `put them back' together? 
These will turn out central questions, to be discussed in detail below. 
\rleditb{In addition, note that when using models for inference or generation one would want---and we will focus on---models that enable doing inference or generation on a \emph{one-shot} basis, i.e. outputting a corresponding inferred or generated state deterministically.} 

Now, one might wonder why such arguably natural questions would not have been adressed before in the field of quantum machine learning (QML), and more generally `quantum AI', given its increasingly large and rich literature.  
\edited{Part of the reason is that QML encompasses several distinct regimes, depending on whether the data, the processing and the intended output are classical or quantum. Much of the literature begins with classical data, which are embedded into quantum states or supplied through a quantum access model, and uses quantum or hybrid processing for a classical task such as classification, probabilistic inference or sample generation 
\cite{Rebentrost2014, Schuld2019, PerezSalinas2020, Liu2021, Lorenz2023, Sunderhauf2024, Zhang2024}. Quantum algorithms for Bayesian or variational inference and quantum-assisted generative models similarly tend to represent or generate classical probability distributions \cite{Low2014, Miyahara2018, Benedetti_2021, Benedetti_2018, Benedetti2019, Coyle2020, Huang2021, Rudolph2022}. In another line of work, the object of study is a quantum system, but the learner acts on measurement records or other classical descriptions \cite{de2026discovering}. Recent work retaining coherent quantum inputs has established advantages for particular inference tasks, including tasks with quantum outputs \cite{li2026exponential}, but does not aim at a generic unsupervised latent-variable model.

Closer in spirit, quantum autoencoders and fully quantum variational autoencoders map quantum states to a quantum latent system and reconstruct them \cite{romero2017quantum, wang2025quantum, cha2026toward}, while recent adversarial extensions learn a generator over the compressed latent states \cite{raj2025quantum}. Their encoder and decoder, however, are directly parametrised operations rather than being constrained to arise as compatible factorisations of a common joint model. Quantum Boltzmann machines and quantum GANs can also learn or generate quantum states, sometimes using quantum visible and hidden units \cite{kieferova2017tomography, lloyd2018quantum, Wilde2025}, 
but generally do not endow the same model with an inference map. A particularly close recent development from the joint-state side \rleditc{is the \emph{DO-EM}\footnote{\rleditc{standing for `Density Operator Expectation Maximization'}} framework, which takes a visible-latent density operator as the model and formulates an information-projection step based on the data processing inequality, and} with Petz recovery under suitable conditions \cite{vishnu2025density}. 
To our knowledge, these approaches do not impose all the ingredients together: a single joint visible-latent quantum model admitting both inference and generation factorisations, \rledit{with corresponding maps that yield valid states in the joint space for every input state} and a construction that reduces, in the classical limit, to the two chain-rule factorisations of a joint distribution. 
}

In this work we develop a first conceptual and theoretical framework for unsupervised representation learning for quantum data. 
In order to overcome the challenge posed by the lack of an obvious and generic analogue of the chain rule---a challenge one may in fact see as a manifestation of the famous no-cloning theorem---we take inspiration from works on `states over time'~\cite{Leifer_2013,Horsman_2017,Fullwood_2022} and define a suitable sense of a state factorising into a channel and a marginal state.  
On its basis we present definitions of quantum models with inference and generation, show that such models exist, sketch how one might in principle train them, and derive characterisation and no-go theorems for classes of models with various desirable properties. 

In more detail, after recalling some key elements of classical representation learning (Sec.~\ref{subsec:intro-classical-case}) and giving our conventions and notation for the quantum formalism (Sec.~\ref{subsec:notation}), we expand the analysis of the above mentioned key challenge in Sec.~\ref{subsec:intro-the-challenge} by carefully contrasting with the classical case.  
Sec.~\ref{sec:the-framework} then lays out our \emph{general framework}---the novel definitions lifting key ideas from classical representation learning to a quantum context---and argues that our desiderata are indeed appropriate. 
The framework involves in particular definitions of \emph{quantum data and model}, \emph{data extension maps}, 
a suitable notion of \emph{state-over-time} (\sot) maps, which induces notions of \emph{ambiguous quantum states}, which generalise how bipartite probability distributions factorise according to the chain rule, and on the basis of which we then define what a \emph{quantum model with inference and generation} is; finally, we explain how data extension and SOT-map-based inference maps relate to one another, and what quantum generalisations of basic training setups look like. 

Sec.~\ref{sec:sot-maps} presents three examples of \sot\ maps from the literature which also fit our definition and which we then focus on for the purpose of this exposition: the Leifer-Spekkens (LS) notion~\cite{Leifer_2013}, the Jordan-product notion~\cite{Fitzsimons2015,Horsman_2017} and one based on the ordinary operator product. 
Sec.~\ref{sec:general-results-sot-ambiguity} presents a general no-go result, establishing that well-behaved yet non-trivial models necessarily come with non-linear extended inference (and data extension) maps; we explain why this is not an in-principle issue.  
Sec.~\ref{sec:ambiguous-states} studies the respective ambiguous states for each of the \sot\ maps from Sec.~\ref{sec:sot-maps} and presents complete characterisation theorems for them, in one case rediscovering a result from~\cite{song2025bipartite}. All three are related to each other in interesting ways and to PPT states, i.e. ones with a positive partial transpose---a condition of great significance in entanglement theory since it implies limited (yet in general genuinely) quantum correlations known as \emph{bound entanglement}~\cite{Horodecki1998}. 

Sec.~\ref{sec:main-section-on-LS-models} puts all the pieces together and establishes the existence of classes of quantum models with inference and generation using the LS construction; the other \sot\ maps considered in this work, on the other hand, are affected by the no-go result. 
Finally, Sec.~\ref{sec:data-extension-not-from-inf} touches on a much weaker goal than the rest of the work, exploring what it would mean to use a data extension map not induced by \sot-map based quantum inference and discusses a concrete construction to this end.  
We conclude in Sec.~\ref{sec:discussion} with remarks on the wider relevance of the results and further discussion of future directions.

\section{Preliminaries} \label{sec:layign-the-ground}

\subsection{Cornerstones of classical representation learning} \label{subsec:intro-classical-case}

Representation learning is a vast topic with a multitude of deep learning approaches. 
It includes supervised learning of features, often performed via convolutional neural networks, in this paper though we are interested in the unsupervised approach because of its generality and historical importance---early breakthroughs happened within that paradigm~\cite{Dayan1995,Hinton2006,kingma2014auto}. 
Although the key idea of \emph{deep} learning is to train multiple layers of latent variables to extract a hierarchy of features, for the purpose of this paper,  
\rledit{we do not make any structural assumptions about the latent space; results will thus be general in this respect.\footnote{\rledit{Of course the structure of the latent space is important in practice, e.g., if it can be exploited by the training algorithm. However, such practical questions are not the focus of this work.}}} This section will recall standard, core ideas of unsupervised representation learning and present them in a way that will be useful later. 

Suppose data is given, which we assume to be described by a random variable $X$ taking values in (some subset of) $\mathbb{R}^{d_X}$ for some $d_X \in \mathbb{N}$, 
with the components $X_j$ for $j=1,...,d_X$ considered the \emph{visible variables}, 
and by a probability distribution $Q(X)$; 
a dataset is a collection $\{x_i \in \mathbb{R}^{d_X} \}_{i=1}^N$ of $N$ realisations, drawn i.i.d. according to $Q(X)$. 
With a view to representation learning then, models additionally use a \emph{latent} variable $Y$ for representation, which takes values in (some subset of) $\mathbb{R}^{d_Y}$ for some $d_Y \in \mathbb{N}$. 

A central fact of classical probability theory worth emphasising for later reference is the following 
\begin{align}
		P(Y | X) P(X) \ = \ P(X,Y) \ = \ P(X | Y) P(Y) \ , \label{eq:classcial-bipartite-factorisation}
\end{align}
which one may read in two ways: given a distribution $P(X,Y)$ it can always be factorised according to the chain rule as on the left and the right-hand side for some stochastic maps $P(Y | X)$ and $P(X | Y)$, respectively\footnote{These are uniquely given by the Bayesian conditional provided  $P(X,Y)$ has full support.}; conversely, given a stochastic map  $P(Y | X)$ and a distribution $P(X)$ they define a joint distribution via the left equality, analogously for the right equality. 

Returning to the notion of models, a model class is a parametrised family $\{P_{\theta}(X,Y)\}_{\theta \in I}$ for some $I \subseteq \mathbb{R}^p$ and $p \in \mathbb{N}$. 
With the above in mind, for each $\theta \in I$ the joint distribution induces two stochastic maps: $P_{\theta}(Y | X)$, commonly referred to as the model's \emph{encoder} or \emph{inference} map, and $P_{\theta}(X | Y)$, referred to as the \emph{decoder} or \emph{generation} map. 
Equivalently, a model class may be given by families of generation maps $P_{\theta}(X | Y)$ and prior distributions $P_{\theta}(Y)$---then often referred to as a generative model---which then \emph{define} the joint $P_{\theta}(X,Y)$, and in turn also the corresponding factorisation in the encoding direction. 
The terminology reflects that given a data point $x_i$, the model encodes a representation of it via $P_{\theta}(Y | X=x_i)$, which is obtained by feeding the point distribution on $x_i$ into $P_{\theta}(Y | X)$; depending on the context one may also think of sampling from $P_{\theta}(Y | X=x_i)$ as doing inference such as inferring a label. Conversely, given a vector $y$ from the latent space the model can generate synthetic data---it defines a distribution $P_{\theta}(X | Y=y)$ from which one may sample vectors $x$. 

Naturally, one says $P_{\theta}(X,Y)$ is a model of given data $Q(X)$ if
\begin{align}
	Q(X) \ = \ P_{\theta}(X) = \int dY P_{\theta}(X,Y) \label{eq:classical-model-of-data}
\end{align}
and the training phase consists of minimising, with respect to the model parameters, some notion of divergence $D$ between the respective distributions, i.e. $\theta^* = \arg\min_\theta \, D(Q(X) \| P_\theta(X))$.  
For concreteness we now use the Kullback-Leibler (KL) divergence 
\begin{align}
\label{obj_kl}
	D(Q(X) \| P_\theta(X)) := \int dX Q(X) \log \frac{Q(X)}{P_\theta(X)} ,
\end{align}
but emphasise that the ideas carry over to more general Csisz\'{a}r $f$-divergences, R\'{e}nyi $\alpha$-divergences, etc. 
The above Eq.~\eqref{eq:classical-model-of-data} is agnostic as to how the model is parametrised, but it is worth noting that the success of representation learning partly comes from using a generative model, i.e. where the integrand in Eq.~\eqref{eq:classical-model-of-data} is given by $P_{\theta}(X | Y)P_{\theta}(Y)$---intuitively speaking, with the model given as decoder and prior on the latent space, training to reproduce the probability distribution of samples in the dataset forces the latent variables to capture the most relevant features of the data. 
For the rest of this exposition assume the model is given as a generative model. 

The above kind of minimisation problems are usually approached by stochastic gradient descent (SGD)~\cite{Bottou2018}.
In its basic form, at iteration $t$ SGD updates the parameters as $\theta^{t+1} = \theta^{t} - \alpha^t \, g(\theta^{t})$ where $\alpha^t > 0$ is the learning rate and, $g$ is a finite-sample approximation to the gradient.
Using a batch of $M \leq N$ datapoints, and using that they are i.i.d., the gradient of the KL divergence is approximated as
\begin{align}
\label{grad_kl}
	\grad D = - \int dX Q(X) \frac{\grad P_\theta(X)}{P_\theta(X)}  \approx g(\theta) = -\frac{1}{M} \sum_{i=1}^M\frac{\grad P_\theta(X=x_i)}{P_\theta(X=x_i)}. 
\end{align} 
While \rleditc{one is here only explicitly training the} prior and decoder, an encoder can of course be recovered from the model posterior, $P_\theta(Y | X) = P_\theta(X , Y) / P_\theta(X)$. 

Importantly, in practice the latent space usually is high-dimensional and its marginalisation to yield $P_\theta(X)$ needed in Eq.~\eqref{grad_kl} is unfortunately intractable. 
To overcome this problem the standard approach is to minimise an upper bound instead. 
By the data processing inequality (DPI), which applies to any $f$-divergence, it holds that
\begin{align}
\label{classical_dpi}
	D(Q(X) \| P_\theta(X)) \ \leq \ D(Q(X, Y) \| P_\theta(X,Y) ) \ ,
\end{align}
where the probability distribution $Q(X,Y)$ is some version of the data $Q(X)$ \emph{extended} to the joint space of $X$ and $Y$\rleditc{, that is such that $Q(X) = \int dY Q(X,Y)$.}

First of all, note that if setting $Q(X, Y) := Q(X) P_\theta(Y | X)$, the above bound is tight, showing that the encoder recovered from the model posterior is indeed optimal. 
This observation suggests a training method where at iteration $t$ one \rleditc{holds $\theta^t$ fixed and performs SGD on 
\begin{align}
\label{objective_em}
    D(Q(X) P_{\theta^t} (Y|X) \| P_\theta(X,Y) ) ,
\end{align}
with respect to $\theta$ to yield $\theta^{t+1}$; this} is an instance of the celebrated Expectation-Maximisation (EM) algorithm~\cite{Dempster1977}.
Using the KL divergence, the approximate gradient now reads 
\begin{align}
\label{grad_em}
	\grad D = - \int dX dY Q(X) P_{\theta^t}(Y|X) \frac{\grad P_\theta(X, Y)}{P_\theta(X,Y)} \approx g(\theta) = -\frac{1}{M} \sum_{i=1}^M\frac{\grad P_\theta(X=x_i, Y=y_i)}{P_\theta(X=x_i, Y=y_i)} ,
\end{align}
where the dataset has been \emph{extended} using the encoder as $(x_i, y_i) \sim Q(X) P_{\theta^t}(Y|X)$. 
Since the encoder is optimal, this data extension is referred to as \emph{exact inference}. 
Its advantage over Eq.~\eqref{grad_kl} is that no high-dimensional integral \rleditc{over the latent space is ever explicitly computed (or rather, ever explicitly approximated)}. 

\begin{figure}
\centering
\begin{tikzpicture}[>=Latex]
	\node (A) at (-6.5*\scale,-1.0*\scale) {\begin{minipage}{2.0cm} \centering {\footnotesize \color{gray} $\ Q_{\phi}(X,Y):=$ \\ $Q_{\phi}(Y | X) Q(X)$} \end{minipage}}; 
	\node (A) at (7.0*\scale,-1.0*\scale) {\begin{minipage}{2.0cm} \centering {\footnotesize \color{gray} $\ P_{\theta}(X,Y):=$ \\ $P_{\theta}(X | Y) P_{\theta}(Y)$} \end{minipage}};
    \node (A) at (0,0) {$Q_{\phi}(Y)$};
    \node (B) at (2.5*\scale,0) {$P_{\theta}(Y)$};
    \node (C) at (0,-2.2*\scale) {$Q(X)$};
    \node (D) at (2.5*\scale,-2.2*\scale) {$P_{\theta}(X)$};
    \node at (0, 0.8*\scale) {\footnotesize Inference};
    \node at (2.5*\scale, 0.8*\scale) {\footnotesize Generation}; 
    \draw[->] (C) -- node[left] {$Q_{\phi}(Y | X)$} (A);
    \draw[->] (B) -- node[right] {$P_{\theta}(X | Y)$} (D);
    \node at (-2.3*\scale,0.05*\scale) {\footnotesize \color{\spacecolor} Latent space};
    \node at (-2.3*\scale,-2.15*\scale)  {\footnotesize \color{\spacecolor} Visible space};  
\end{tikzpicture}
\caption{The generic setup of what we here regard as basic classical representation learning. \label{fig:schema-classical}}
\end{figure}
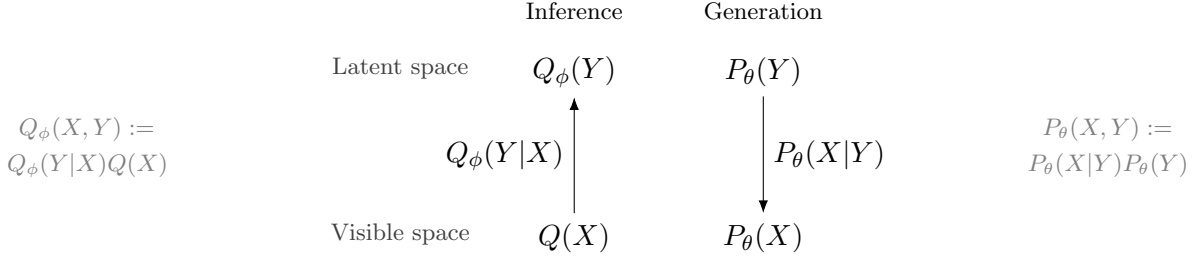

Exact inference is practical when it can be implemented by an efficient algorithm, such as a quickly converging Markov chain Monte Carlo. 
But in most generative models this is not the case, and we have therefore not yet arrived at a generically practical method.  
The concept of \emph{approximate inference} then refers to a large class of methods that basically aim at trading accuracy for efficiency by introducing bias and variance. 
Again consider Eq.~\eqref{classical_dpi}, but set $Q(X, Y) :=  Q_\phi(Y | X) Q(X)$, where $Q_\phi(Y | X)$ is an independently given and parametrised `encoder' via which the data $Q(X)$ can still be \emph{extended} to the joint space of $X$ and $Y$. 
Training the model now also involves tuning the new parameters $\phi \in \mathbb{R}^q$ so as to approximate the model posterior. 
This idea goes all the way back to the Wake-Sleep (WS) algorithm~\cite{Hinton1995}, and was later popularised by the Variational Autoencoder (VAE)~\cite{kingma2014auto}. 
This leads to a training method where at iteration $t$ one performs SGD on both
\begin{align}
\label{objective_var}
    D_1(Q(X) Q_\phi(Y|X) \| P_{\theta^{t}}(X,Y)) \qquad \text{and} \qquad D_2(Q(X) Q_{\phi^t}(Y|X) \| P_{\theta}(X,Y)) .
\end{align}
For example, when $D_1$ is the reversed KL divergence and $D_2$ is the KL divergence, gradients are approximated as
\begin{align}
	&g_1(\phi) = -\frac{1}{M_1} \sum_{i=1}^{M_1} \frac{\grad Q_\phi(Y=y_i | X=x_i)}{Q_\phi(Y=y_i | X=x_i)}, \qquad (x_i, y_i) \sim P_{\theta^t}(X | Y) P_{\theta^t}(Y), \label{grad_var} \\
	&g_2(\theta) = -\frac{1}{M_2} \sum_{j=1}^{M_2} \frac{\grad P_\theta(X=x_j , Y=y_j) }{P_\theta(X=x_j , Y=y_j)}, \qquad (x_j, y_j) \sim Q_{\phi^t}(Y | X) Q(X), \label{grad_var_2}
\end{align}
yielding a generic method for unsupervised representation learning and, in fact, generative modelling.

Although we have only scratched the surface of this broad topic, the take-away is that the key to all the methods presented here is \emph{inference and generation}. 
Not just---as their names suggest---when \emph{using} a model to infer missing values, make predictions, generate new data etc., but also for practical training setups. 
Inference maps allow one to extend the visible data to include latent variables; 
generation provides new samples that reflect the model's current understanding. 
Together they enable the estimation of correlations, divergences, gradients etc.

\subsection{Quantum formalism -- notation and conventions} \label{subsec:notation}

This work will assume finite-dimensional quantum systems throughout. 
For $X$ a system (label), let $\hilb{X}$ denote the associated Hilbert space with dimension $d_X$, 
$\operators{X}$ the algebra of linear operators on $\hilb{X}$ and $\Herm{X} = \{A \in \operators{X} \mid A^{\dagger} = A \}$ the Hermitian ones.  
For $A \in \operators{X}$ write $A \succeq 0$ iff $A$ is positive semi-definite (PSD), i.e. $\forall \ket{x}\in \hilb{X}: \bra{x}A\ket{x} \geq 0$, or equivalently, $A \in \Herm{X}$ and $A$ has non-negative spectrum.  
The states on $X$ are denoted $\states{X}$, i.e. $\states{X}= \{\rho \in \operators{X} \mid \rho \succeq 0, \tr[\rho]=1\}$. 
We denote the support of $A \in \operators{H}$ by $supp(A) \subseteq \hilb{H}$ and `ONB' stands for orthonormal basis. 

Moving on from single systems, $\linmaps{X}{Y}$ denotes the space of linear maps of type $\operators{X} \rightarrow \operators{Y}$, 
with $\cpmaps{X}{Y}$ denoting the subset of completely positive (CP) maps and $\channels{X}{Y}$ the CPTP maps, i.e. CP and additionally trace-preserving (TP) ones. 
We use `channel' synonymously with CPTP map. 
Given a tensor product space $\hilb{X \otimes Y}$ we denote the partial trace over $Y$ by  $\tr_Y \in \channels{X \otimes Y}{X}$. 
Given a joint state $\rho \in \states{X \otimes Y}$, write $\rho_X := \tr_Y(\rho)$, i.e. a system label in subscript indicates being a marginal of a given joint state, otherwise we omit system labels and write, e.g., $\sigma \in \states{X}$. 

The identity operator is denoted $\idop_X \in \operators{X}$ and the identity channel $\idchan_X \in \channels{X}{X}$.
We use $\pi \in \operators{X}$ exclusively to denote an orthogonal projector, i.e. $\pi^2=\pi$ and $\pi^{\dagger} =\pi$, and then write $\pi^{\perp}:= \idop_X - \pi$ and $\Pi(\argline) := \pi (\argline)\pi$ for the corresponding map $\Pi \in \cpmaps{X}{X}$, as well as $\Pi^{\perp} := \idchan_X - \Pi$. 
For any $A\in \operators{X}, A \succeq 0$ we let $A^{-1}$ denote the pseudo-inverse, i.e. $A^{-1}A = \pi = AA^{-1}$, where $\pi$ projects onto $supp(A)$; similarly, writing $A^{-1/2}$ implicitly restricts to the support of the argument, i.e. it denotes the operator $\sum_{i:\lambda_i \neq 0} \lambda_i^{-1/2} \pi_i$ for $\sum_i \lambda_i \pi_i$ an eigendecomposition of $A$.	 
Recall a state $\rho \in \states{X \otimes Y}$ is called \emph{non-separable} (or entangled) iff it cannot be written as $\rho = \sum_i p_i \alpha_i \otimes \beta_i$ for some $\alpha_i \in \states{X}, \beta_i\in \states{Y}$ and $p_i \geq 0$ with $\sum_i p_i =1$. 

Finally, a common convention we will use throughout to avoid clutter drops identity operators in tensor product where this is innocent, i.e. for $A \in \operators{X \otimes Y}, B \in \operators{X}$ we write $AB$ as shorthand for $A(B \otimes \idop_Y)$---it is understood one always pads appropriately with identity operators to obtain well-formed products of operators. 
For emphasis we occasionally drop the convention and make the underlying structure explicit again.

\subsection{The challenges of a quantum generalisation} \label{subsec:intro-the-challenge}

With some cornerstones of classical representation learning thus laid out let us expand on why it is not at all obvious what a viable overall quantum generalisation would be. 

\paragraph{Straightforward part I: model components.} 
Some parts of lifting the setup in Fig.~\ref{fig:schema-classical} to a quantum version are straightforward: 
of course, quantum states and quantum channels (i.e. CPTP maps) formally generalise probability distributions and stochastic maps, respectively.  
A genuinely quantum version of a data source, which a user can access to sample data points from, is also straightforward---basically, an \emph{ensemble of} states and a distribution over them---formalised later in Def.~\ref{def:quantum-data}. 
In short, one arrives at the sort of picture as in Fig.~\ref{fig:schema-quantum-done-naively}: let $X$ and $Y$ be quantum systems, given is a dataset of $N$ quantum states $\{\sampleindexthree{\sigma}{}{i}\}_{i=1}^N$, which are i.i.d. realisations of some unknown density $\sigma \in \states{X}$ and the model components are lifted to quantum versions, where the classical prior $P_\theta(Y)$ becomes a quantum state $\eta_{\theta} \in \states{Y}$, the encoder $Q_\phi(Y|X)$ a channel $\mathcal{E}_\phi \in \channels{X}{Y}$, and the decoder $P_\theta(X|Y)$ a channel ${\rleditb{\genchansymbol}}_\theta \in \channels{Y}{X}$; just like before $\phi$ and $\theta$ parametrise respective families of states and channels.  

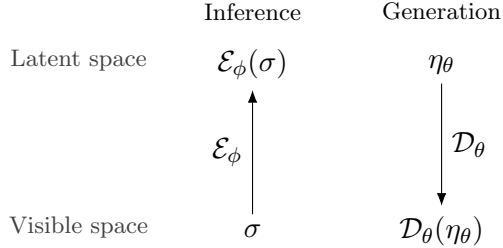
\begin{figure}[h]
\centering
\begin{tikzpicture}[>=Latex]
	\node (A) at (0,0) {$\mathcal{E}_{\phi}(\sigma)$};
    \node (B) at (2.5*\scale,0) {$\eta_{\theta}$};
    \node (C) at (0,-2.2*\scale) {$\sigma$};
	\node (D) at (2.5*\scale,-2.2*\scale) {${\rleditb{\genchansymbol}}_{\theta}(\eta_{\theta})$};
    \node at (0, 0.7*\scale) {\footnotesize Inference};
    \node at (2.5*\scale, 0.7*\scale) {\footnotesize Generation}; 
    \draw[->] (C) -- node[left] {$\mathcal{E}_{\phi}$} (A);
    \draw[->] (B) -- node[right] {${\rleditb{\genchansymbol}}_{\theta}$} (D);
    \node at (-2.3*\scale,0.05*\scale) {\footnotesize \color{\spacecolor} Latent space};
    \node at (-2.3*\scale,-2.15*\scale)  {\footnotesize \color{\spacecolor} Visible space};   
\end{tikzpicture}
\caption{A partial generalisation to quantum theory of the setup in Fig.~\ref{fig:schema-classical}.\label{fig:schema-quantum-done-naively}}
\end{figure}

\paragraph{Straightforward part II: a quantum divergence.} 
The KL divergence can be generalised, e.g., to the quantum relative entropy $D(\alpha \| \beta) := \tr(\alpha (\log \alpha - \log\beta))$ for $\alpha, \beta \in \operators{H}$, $\alpha, \beta \succeq 0$ and provided $supp(\alpha) \subseteq supp(\beta)$ so that $D(\alpha \| \beta)$ is well-defined given $0 \log 0 := 0$.  
The generalisation of Eq.~\eqref{obj_kl}, writing $\tilde{\eta}_{\theta} := {\rleditb{\genchansymbol}}_{\theta}(\eta_{\theta})$, then reads 
\begin{align}
\label{obj_qre}
	D(\sigma \| \tilde{\eta}_{\theta} ) := \tr\big[\sigma \big(\log \sigma - \log \tilde{\eta}_{\theta} \big)\big] \ .
\end{align}

\paragraph{The challenges.}
Similarly to the classical case one could in principle propose to train prior and decoder by SGD.\footnote{Although we do not give the expressions specifically for this approach, they can be derived by analogous steps as in Eqs.~\eqref{eq:log-der}-\eqref{eq:app-grad-expression} in App.~\ref{app:gradient_qre}, which refer to joint systems and the DPI.} 
One would then try to recover the encoder by some form of channel inversion. 
Albeit not a unique choice,  e.g., the Petz recovery map~\cite{PETZ1988} 
\begin{align}
	\mathcal{E}_\theta (\argline) := \eta_{\theta}^{\frac{1}{2}} \; {\rleditb{\genchansymbol}}_\theta^{\dagger}\left( {\rleditb{\genchansymbol}}_\theta(\eta_{\theta})^{-\frac{1}{2}} \;(\argline)\; {\rleditb{\genchansymbol}}_\theta(\eta_{\theta})^{-\frac{1}{2}} \right) \; \eta_{\theta}^{\frac{1}{2}} \label{eq:Petz-recovery-map}
\end{align} 
would be a valid choice in our context, seeing as indeed $\mathcal{E}_\theta (\tilde{\eta}_{\theta} )  = \eta_{\theta}$. 
However both, the gradient of Eq.~\eqref{obj_qre} and the Petz recovery map are daunting mathematical expressions to deal with in all but the very simplest models. 

The next natural step, mirroring the classical move in Eq.~\eqref{classical_dpi}, considers the DPI, which also holds quantumly  
$\forall \alpha,\beta \in \states{X \otimes Y}$, we have $D(\alpha_X \| \beta_X ) \leq D(\alpha \| \beta )$.\footnote{The DPI is valid for any channel $\mathcal{N}$, and thus one may be tempted to invoke $D(\mathcal{N}(\sigma)\|\mathcal{N}(\tilde{\eta}_{\theta}))$ which however leads to a \emph{lower} bound. The partial trace over a larger space is arguably the most sensible way to get an \emph{upper} bound for minimisation. Also, in the quantification over states $\alpha,\beta \in \states{X \otimes Y}$ we suppressed that one of course also requires that $supp(\alpha_X) \subseteq supp(\beta_X)$ and $supp(\alpha) \subseteq supp(\beta)$.} 
So a quantum generalisation of Eq.~\eqref{classical_dpi} structurally would look like this
\begin{align}
	D(\sigma \| \tilde{\eta}_{\theta}) \ \leq \ D(\bar{\sigma}_{\phi} \| \bar{\eta}_{\theta}), \label{eq:quantum-DPI}
\end{align} 
where $ \bar{\sigma}_{\phi}, \bar{\eta}_{\theta} \in \states{X \otimes Y}$ are joint states with the right marginals, i.e. $\bar{\sigma}_{\phi;X} = \sigma$ and $\bar{\eta}_{\theta;X} = \tilde{\eta}_{\theta}$.  
Above we have already allowed for generally distinct parameters $(\phi,\theta)$ so that the RHS above may potentially be thought of as a quantum version of Eq.~\eqref{objective_em} as well as Eq.~\eqref{objective_var}, depending on the relation between $\phi$ and $\theta$. 
The approximate gradients that would be needed---the analogues of Eqs.~\eqref{grad_em}, \eqref{grad_var} and \eqref{grad_var_2}---are discussed in App.~\ref{app:gradient_qre}.

Crucially however, now we face this \emph{challenge}: 
what are these joint states $\bar{\sigma}_{\phi}, \bar{\eta}_{\theta}$ supposed to be? 
We have defined the model in terms of \emph{its components} as in Fig.~\ref{fig:schema-quantum-done-naively}, namely data, encoder, prior, and decoder; how then should joint states $\bar{\sigma}_{\phi}$ and $\bar{\eta}_{\theta}$ be constructed in terms of these and what would such states even mean? 
In essence, one ends up asking what, if anything at all, the quantum generalisation of \emph{data extension} is that, analogously to the classical case, encodes a data state $\sigma$ into some joint state $\bar{\sigma}_{\phi}$, but ideally also plays well with `having inference and generation': $\bar{\eta}_{\theta}$ must in some sense be compatible with the existence of inference and generation channels $\mathcal{E}_{\theta}$ and ${\rleditb{\genchansymbol}}_{\theta}$, respectively.\footnote{Assuming for the moment that valid extended data $\bar{\sigma}_{\phi}$ and model $\bar{\eta}_{\theta}$ can be defined in our context, can the DPI be saturated? For the partial trace, which is the channel used in our DPI inequality, recoverability theorems~\cite{PETZ1988,Hayden2004, Wilde2015} imply there exist \emph{some} pairs of states that saturate the DPI, but these may not be achievable within the model class $\{\bar{\eta}_{\theta}\}_\theta$, or may be trivial. In contrast, the classical DPI in Eq.~\eqref{classical_dpi} can \emph{always} be saturated by extending the data using the model posterior. Thus in the quantum setting there exists no posterior-like object that provides optimal inference in general.} 

It is worth dwelling on this for a moment for, classically, the corresponding step easily goes unnoticed. 
Recall Eq.~\eqref{eq:classcial-bipartite-factorisation}, the starting point of our exposition of classical models in Sec.~\ref{subsec:intro-classical-case}---data extension, inference and generation and their mutual compatibility are essentially facilitated by the existence of a universal copy map. 
The role of the latter is implicit in how the factorisation of probability distributions is conventionally written, just like in Eq.~\eqref{eq:classcial-bipartite-factorisation}; a representation that makes it explicit is the following \emph{string diagrammatic} notation  
\begin{align}
	\begin{minipage}{0.9\textwidth}
		\centering
		\includegraphics[scale=0.2]{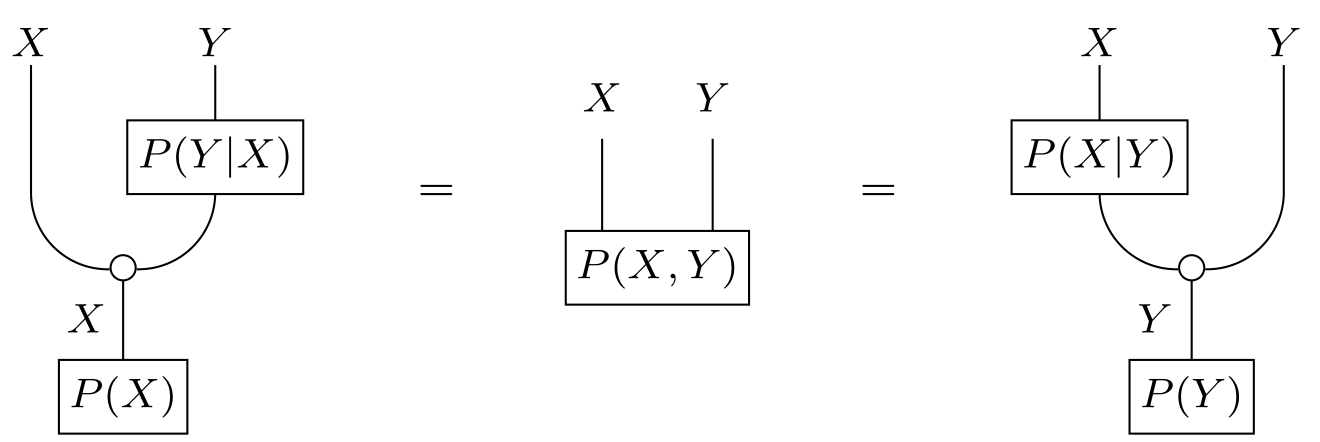}
	\end{minipage}
\end{align}
where \ $\tinycopy$ \ denotes the copy map. Here we only note that such string diagrams allow an intuitive, yet fully formal representation of probability theory based on (symmetric monoidal) categories, while for further details we  refer to, e.g. \cite{cho2019disintegration,fritz2020synthetic,lorenz2023causal} and references therein. 
A copy map can in particular be thought of as inducing a map of the following form that we will come back to later
\begin{align}
	\begin{minipage}{0.9\textwidth}
		\centering
		\includegraphics[scale=0.2]{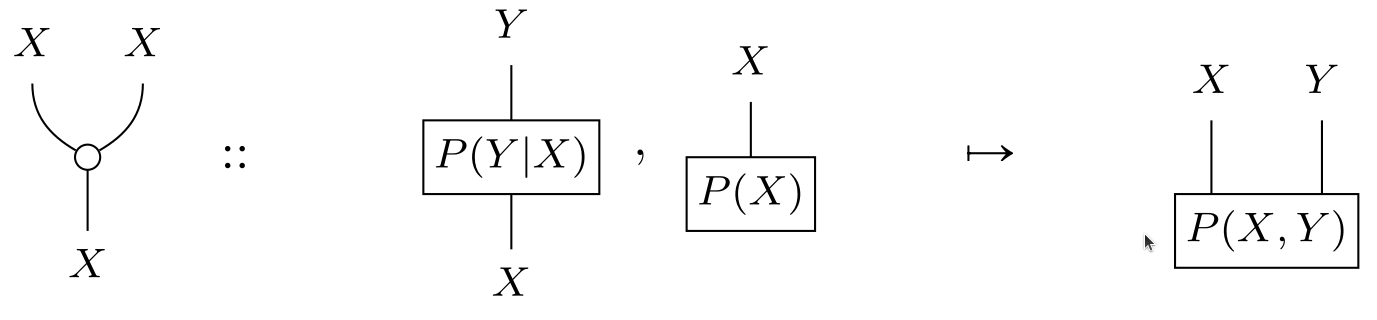}
	\end{minipage}
	\label{eq:classical-sot-like-map}
\end{align}

One way then to relate the challenges posed by a quantum generalisation of classical representation learning to fundamental facts about quantum theory thus is the no-cloning theorem \cite{wootters1982single, dieks1982communication}, that is, the fact that there does \emph{not} exist a (physical, i.e. CPTP) map that copies all pure states of a system.  
\rledit{Indeed, there is no unique quantum analogue of Eq.~\eqref{eq:classcial-bipartite-factorisation}. 
This fact has motivated and been studied in various lines of work, such as on 
quantum conditional amplitudes \cite{cerf1999quantum}, quantum conditional states \cite{leifer2006quantum, Leifer_2013}, 
Bayesian state updating and density-operator probability calculi \cite{schack2001quantum, warmuth2010bayesian}, 
quantum graphical models and noncommutative Bayesian inversion \cite{leifer2008quantum, coecke2012picturing, parzygnat2022non},  
and states over time \cite{Horsman_2017, Fullwood_2022, Parzygnat_2023}. 
These strands overlap---in particular, many proposed quantum Bayes rules can be formulated relative to a choice of \sot\ construction \cite{Parzygnat_2023}---but differ in their aims and axioms. 
The \sot\ perspective is the one most directly aligned with our problem, and the framework developed below adopts only the structural ingredients needed for representation learning rather than a complete quantum Bayesian probability calculus.}

\section{The framework} \label{sec:the-framework}

This section presents our framework: notions that generalise the key ideas from the classical setup in Sec.~\ref{subsec:intro-classical-case} to fully quantum versions. 
The desiderata in the definitions aim for strong, ideal notions in some respects, while keeping them broad and general in other respects, yielding a general framework of `foil notions'---we do not claim that every instance satisfying the definitions constitutes an interesting or efficiently implementable model; 
rather, it is a framework within which one can study at all quantum models in a principled manner and also make progress towards optimal classes of models from a practical perspective.

\subsection{Quantum data and models} \label{subsec:data-and-models}

Let us start with the central notion of quantum data. 

\begin{definition} \label{def:quantum-data}
	A \emph{quantum data source} on system $X$ 
	is given by an ensemble $\{ (p_i, \decompindex{\sigma}{i}) \}_{i=1}^k$ for some $k\in \mathbb{N}$, where for each $i= 1,..,k$ $\decompindex{\sigma}{i}  \in \states{X}$, $p_i \geq 0$ and $\sum_i p_i =1$. 
	Often we refer to just $\sigma=\sum_i p_i \decompindex{\sigma}{i} $ as the \emph{quantum data}, where it is understood that when querying the data source one obtains state $\decompindex{\sigma}{i} $ with probability $p_i$, and the process can be repeated to obtain i.i.d. samples. 
	While the `data' is defined with respect to a fixed ensemble, the user with access to the data source does not generally know this decomposition.
\end{definition}

First of all recall that this work considers only finite-dimensional quantum systems (see Sec.~\ref{subsec:notation}); this is for simplicity so as to focus on the conceptual and technical issues that are present already in finite dimensions, but also natural anyway for digital quantum computing paradigms. 
The classical case that the above definition then generalises is a discrete variable $V$ taking values in a set of finite cardinality $\{v_i\}_{i=1}^{d_X}$ together with a probability distribution $Q(V)$: for $\ket{v_i}_{i=1}^{d_X}$ an ONB of $\hilb{X}$, the tuple $\{ \big(Q(V=v_i),  \dyad{v_i} \big)_i  \}_{i=1}^{d_X}$ is a special case of a quantum data source.
Henceforth, given quantum system $X$ with Hilbert space $\hilb{X}$ we will at times allow ourselves to refer to a probability distribution $P(X)$, where we leave the choice of ONB to encode that distribution as a state in $\states{X}$ implicit---this ambiguity of $X$ will be harmless given the context. 

Given data source $\{ (p_i, \decompindex{\sigma}{i}) \}_{i=1}^k$, we will denote a specific realisation of some $j$-th query as $\sampleindex{\sigma}{j}$---for instance, the fifth query may return the third component of the ensemble, i.e. $\sampleindex{\sigma}{5} = \decompindex{\sigma}{3}$. 
By a \emph{dataset} then we mean the collection of $N$ i.i.d. queries (for some $N \in \mathbb{N}$) and note that this may be seen as a sequence $(\sampleindex{\sigma}{j})_{j=1}^N$, or alternatively as a tensor product $\otimes_{j=1}^N \sigma^{(j)}$---albeit relevant in practice, it is irrelevant to the remainder here which of the two one chooses to think of.\footnote{NB in the sequential case after a datapoint $\sigma^{(j)}$ is provided (and maybe processed), the data source must `reset' before providing another datapoint $\sigma^{(j+1)}$ to guarantee i.i.d. samples.}  
Importantly, the user does not have access to any state preparation protocol and cannot rely on the availability of copies of any particular datapoint; querying the source twice returns $\sampleindex{\sigma}{1} = \decompindex{\sigma}{i}$ and $\sampleindex{\sigma}{2} = \decompindex{\sigma}{j}$, for some $i$ and $j$, but no guarantee that $i=j$. 

With a view to \emph{quantum representation learning} we are interested in a notion of quantum model over \emph{\visible\ and \hidden} systems; the obvious basic notion---without any frills yet---is thus as follows.  

\begin{definition} \label{def:quantum-model-of-data}
	Given systems $X$ and $Y$, considered `\visible' and `\hidden', respectively, a \emph{model class} is a parametrised class of states $\{\rho_{\theta}\}_{\theta \in I}  \subseteq \states{X \otimes Y}$ for $I \subseteq \mathbb{R}^n$ for some $n \in \mathbb{N}$; each element $\rho_{\theta}$ is referred to as a model. 
	Given quantum data $\sigma \in \states{X}$ and some $\Theta \in I$ we say $\rho_{\Theta}$ is a model of the data if 	(approximately):\footnote{In practice choosing a suitable notion of quantifying the approximation is of course crucial and this would then go through all subsequent definitions. It is however irrelevant to the challenges and concerns this work studies and hence we consider equalities on the nose everywhere.}
	\begin{align}
		\tr_Y(\rho_{\Theta}) = \sigma \ . \label{eq:model-of-sigma-X}
	\end{align}
\end{definition}

Obviously, any parametrised bipartite distribution $P_{\theta}(X,Y)$ defined by some classical model as in Sec.~\ref{subsec:intro-classical-case} can be seen (to be encoded) as a special case of the above notion by choosing some product ONB of $\hilb{X \otimes Y}$, and with Eq.~\eqref{eq:model-of-sigma-X} being the generalisation of Eq.~\eqref{eq:classical-model-of-data}.  

It is worth emphasising though that given some state $\sigma\in \states{X}$ ensemble decompositions of the form $\sigma= \sum_i p_i \sigma_i$ are of course highly non-unique---an aspect, which has no analogue classically.  
According to quantum theory no (even in principle) physically implementable operation, i.e. measurement or more generally quantum instrument, can distinguish different ensemble decompositions of some fixed $\sigma$. 
Indeed, we defined a model of data in terms of $\sigma$, that is, with a condition that does \emph{not} depend on the ensemble decomposition behind a specific data source. 
So why then stipulate a data source come with a fixed ensemble decomposition $\{ (p_i, \decompindex{\sigma}{i}) \}_{i=1}^k$ at all? 
First, this is the immediate and natural generalisation of a classical data source while the independence from any specific ensemble decomposition in quantum theory is an afterthought, i.e. a consequence, rather than a desideratum for modelling. 
Second, in light of the fact that our constructions will consider \emph{non-linear} maps, it is a priori actually not entirely obvious why the choice of ensemble decomposition does not matter \rledit{(as will turn out to be the case)}, hence it will be useful for pedagogical purposes to make explicit that the individual data samples $\sampleindex{\sigma}{j}$ may not be equal to \rledit{the ensemble average---the data state} $\sigma$.

\subsection{Quantum data extension} \label{subsec:def-data-extension}

Classical information and learning theory do not distinguish between the concepts of, and methods for, data extension and (extended) inference. 
The reason is that they \emph{always} can be taken to be given by the same construction---essentially because of the existence of a universal copy map, i.e. the facts reflected in Eq.~\eqref{eq:classcial-bipartite-factorisation}. 

For the development of a general quantum framework however it is instructive to separate data extension and (extended) inference. 
They are conceptually distinct and come with logically distinct requirements---as we will see, those of the latter being strictly stronger than those of the former. 
A priori, one may simply want to study constructions for data extension only (see Sec.~\ref{sec:data-extension-not-from-inf}); and even if not, just for pedagogical purposes, making the relation between the two explicit relies on separating them in the first place.

\begin{definition} \label{def:extension-map}
	Given \visible\ and \hidden\ systems, $X$ and $Y$, a \emph{data extension map} is a map 
	$\Phi \ : \ \operators{X} \ \rightarrow \ \operators{X \otimes Y}$  
	such that $\forall \eta \in \states{X}$ it holds that $\Phi(\eta) \in \states{X \otimes Y}$ and $\tr_Y( \Phi(\eta)) = \eta$. 
\end{definition}

A data extension map takes any input state on $X$---in particular any data state $\sigma$ or sampled state $\sampleindex{\sigma}{j}$---to a joint state on $X\otimes Y$ in a way that preserves information about the input. So it really just gives a name to the requirement that a quantum generalisation of Eq.~\eqref{classical_dpi} will need to specify a \emph{joint state}, obtained from data states on $X$, so as to compare against a model state on $X\otimes Y$ via, say, quantum relative entropy. 

It thus generalises \emph{one aspect} of the classical case, where any stochastic map $P(Y|X)$ induces the map that takes any distribution $P(X)$ to $P(Y|X) P(X)$; 
this map may be written diagrammatically---by fixing the first argument in Eq.~\eqref{eq:classical-sot-like-map}---as: 
\begin{align}
	\begin{minipage}{0.9\textwidth}
		\centering
		\includegraphics[scale=0.2]{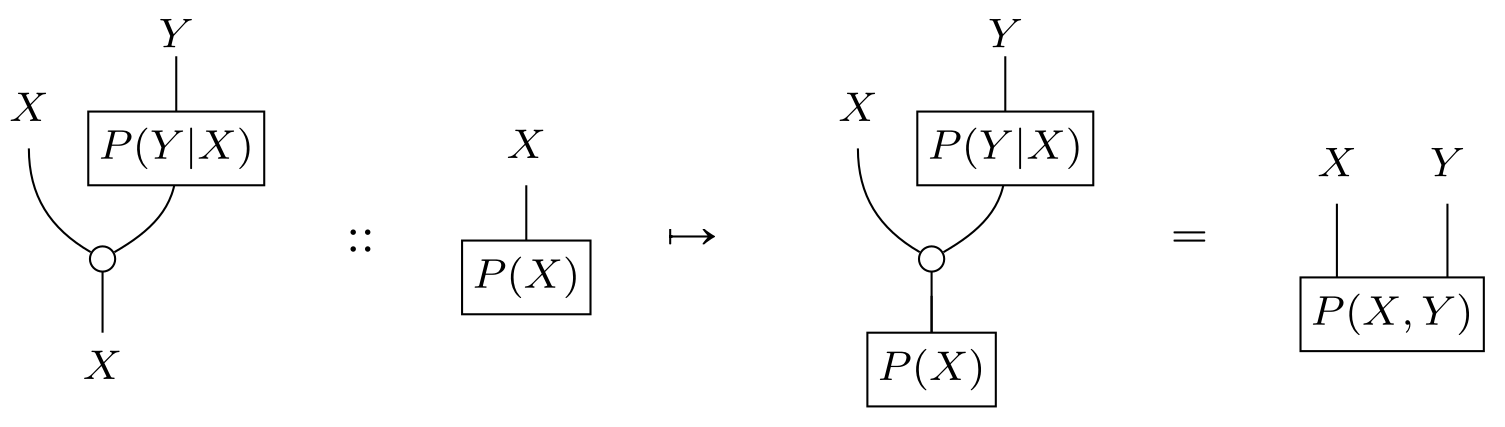}
	\end{minipage}
	\label{eq:classical-data-extension-map}
\end{align}

Note however that we have not said how any classical stochastic map could actually become a special case of the above definition, i.e. yield a map defined on all quantum states $\states{X}$, not just those diagonal in a fixed ONB of $\hilb{X}$. This aspect will be covered by our treatment of extended inference below. 

Notably, we have only required what is the minimum condition for `data extension', but have not asked for the map itself to be physical, or in any way efficiently implementable. 
In fact, demanding a data extension map be CPTP forces it to be trivial in the following sense---it would simply append a fixed state on $Y$, independent of the input state.  

\begin{proposition} \label{prop:no-braodcasting-kind-of-fact}
	Let $\mathcal{E} \in \channels{X}{X \otimes Y}$ such that $\tr_Y \circ \mathcal{E}  = \idchan_X$ then $\mathcal{E} = \idchan_X \otimes \eta$ for some $\eta \in \states{Y}$. 
\end{proposition}

\begin{proof} 
See \hyperref[proof:prop:no-braodcasting-kind-of-fact]{Proof in App.~B}. 
\end{proof}

\edited{It is not obvious that maps trivial in this sense are completely useless: by optimising over a parametrised family, a training procedure could make the appended state on $Y$ effectively depend on the training data, analogously to encoder-free latent optimisation in classical representation learning~\cite{BojanowskiEtAL_2018} and local variational inference in statistics~\cite{Blei2017}.}
However, a general framework surely ought to allow for some non-trivial maps, hence above definition does not require CPTPness.  
We return to questions of implementation in Sec.~\ref{subsec:general-training-setups}. 

It is worth noting here the connection to a fundamental fact about quantum theory, namely the \emph{no-broadcasting theorem} \cite{barnum1996noncommuting, barnum2007generalized}, which generalises the no-cloning theorem. As is immediate, Prop.~\ref{prop:no-braodcasting-kind-of-fact} indeed implies the impossibility of a universal broadcasting map, i.e. a channel $\mathcal{E}$ with $Y=X$ that satisfies, both, $\tr_Y \circ \mathcal{E}  = \idchan_X$ and $\tr_X \circ \mathcal{E}  = \idchan_Y$.\footnote{There also is a connection to the literature that studies the limits of CP evolution in quantum theory; there the `assignment problem' considers a consistency condition, which is precisely the marginal condition in Def.~\ref{def:extension-map} \cite{pechukas1994reduced, schmid2019initial}.}

\subsection{Quantum models with inference and generation} \label{subsec:def-inf-and-gen}

Let us finally turn to the quantum generalisation of models with inference and generation.  
Just like classically, inference and generation will be considered the same mathematical construction, the difference between them being just which roles the factors $X$ and $Y$ are considered to play. 

Now, as already pointed out in Sec.~\ref{subsec:intro-the-challenge}\hspace*{0.05cm}: while it is straightforward to say what the natural generalisation of an inference (generation) map is, namely a \emph{quantum channel} from the \visible\ to the \hidden\ system (the other way round, respectively), it is not obvious how such data is to relate to a quantum model as a \emph{joint state} over both systems as in Def.~\ref{def:quantum-model-of-data}.  
Basically, one has to ask: what is the right notion of `factorisation' such that---at least for some---quantum states in $\states{X\otimes Y}$ a generalisation of Eq.~\eqref{eq:classcial-bipartite-factorisation} does go through; and conversely, when and how can a channel $\mathcal{E} \in \channels{X}{Y}$ and a state $\eta \in \states{X}$ be combined into a state on $X \otimes Y$ that still encodes both, $\mathcal{E}$ and $\sigma$ in an analogous way to Eq.~\eqref{eq:classical-sot-like-map}?   

It turns out that the literature on \emph{states over time} (see, e.g., \cite{Leifer_2013, Horsman_2017,Fullwood_2022,Parzygnat_2023, lie2025multipartite}) has studied very similar questions. 
Although its motivation and context are somewhat different from ours, there is much overlap at a technical level. 
We will in fact adopt the terminology of a state-over-time map, though we will refer to our notion, defined below, as a \emph{weak} state-over-time map seeing as it requires only two out of the usual set of conditions.\footnote{The other conditions are not as obviously necessary for our context, in particular, there is no reason to insist on a sense of associativity since we are only studying a bipartite setup.}
Such a map $\odot$ essentially generalises the kind of map as in Eq.~\eqref{eq:classical-sot-like-map} in an appropriate way so as to complete our general quantum scheme as depicted in Fig.~\ref{fig:schema-quantum-done-properly}.

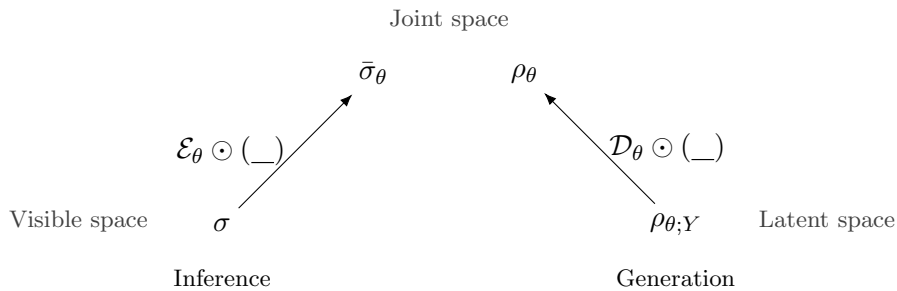
\begin{figure}[h]
\centering
\begin{tikzpicture}[>=Latex]
    \node (A) at (0.5*\scale,0) {$\bar{\sigma}_{\theta}$};
    \node (B) at (2.5*\scale,0) {$\rho_{\theta}$};
    \node (C) at (-1.5*\scale,-2*\scale) {$\sigma$};
    \node (D) at (4.5*\scale,-2*\scale) {$\rho_{\theta; Y}$};
	\node at (1.5*\scale, 0.7*\scale) {\footnotesize \color{\spacecolor} Joint space};
	\node at (6.5*\scale,-1.95*\scale) {\footnotesize \color{\spacecolor} Latent space};
    \node at (-3.4*\scale,-1.95*\scale)  {\footnotesize \color{\spacecolor} Visible space};
    \node[align=center] at (-1.5*\scale, -2.7*\scale) {\footnotesize Inference};
    \node[align=center] at (4.5*\scale, -2.7*\scale) {\footnotesize Generation};
    \draw[->] (C) -- node[left] {$\mathcal{E}_{\theta} \odot (\argline)$} (A);
    \draw[->] (D) -- node[right] {${\rleditb{\genchansymbol}}_{\theta} \odot (\argline)$} (B);
\end{tikzpicture}
\caption{Sketch of the scheme that, using an \sot{} map $\odot$, complements Fig.~\ref{fig:schema-quantum-done-naively}, generalising the classical scheme from Fig.~\ref{fig:schema-classical} appropriately to quantum theory. 
In the direction of inference it shows \emph{data extension} $\bar{\sigma}_{\theta} := \mathcal{E}_{\theta} \odot \sigma$, but \emph{extended inference} can be applied to any state, including the model's own marginal, i.e. $\rho_{\theta} = \mathcal{E}_{\theta} \odot \rho_{\theta; X}$. 
See main text for the details, including a version for approximate inference where data extension is not given by $\mathcal{E}_{\theta}$, but by a data extension map (Def.~\ref{def:extension-map}), i.e. $\bar{\sigma}_{\phi} := \Phi_{\phi}(\sigma)$. 
\label{fig:schema-quantum-done-properly}}
\end{figure}

\begin{definition} \label{def:sot-map}
 	Given finite-dimensional systems $X,Y$,  
 	a map 	
	$\odot : \cpmaps{X}{Y} \times \operators{X} \rightarrow \operators{X \otimes Y} \nonumber$, 
	where we usually write $\mathcal{E} \odot A := \odot(\mathcal{E},A)$,
 	is called a \emph{\sotlong} (\sot{}) map iff it satisfies both of the following conditions.
 	\begin{enumerate}[label=(\arabic*), leftmargin=2.3cm]
 		\item[(M)] \label{itm:condition-M} \emph{Preservation of marginals}. 
		$\forall \mathcal{E} \in \channels{X}{Y}$, $A \in \operators{X}$: 
		\ $ \tr_Y (\mathcal{E} \odot A) = A$, \\ and $\tr_X (\mathcal{E} \odot A) = \mathcal{E}(A)$. 
 		\item[(C)] \label{itm:condition-C} \emph{Preservation of classical limit}.  
		For any $\mathcal{E} \in \channels{X}{Y}$, $\sigma \in \states{X}$ such that with respect to a product ONB $\mathcal{E}$ encodes a stochastic map $P(Y|X)$ and $\eta$ a distribution $P(X)$, $\mathcal{E} \odot \eta$ is diagonal in that same product ONB and encodes $P(Y|X) P(X)$.\footnote{Equivalently, though using the Jamio{\l}kowski representation $D^J(\argline)$ of linear maps in $\linmaps{X}{Y}$, which is only introduced properly in Sec.~\ref{sec:sot-maps}, $(C)$ may be expressed as\hspace*{0.07cm}: $\forall \mathcal{E} \in \channels{X}{Y}$, $\eta \in \states{X}$ such that there exists a product ONB, with respect to which $D^J(\mathcal{E})$ and $\eta \otimes \idop_Y$ are simultaneously diagonal, it holds that $\mathcal{E} \odot \eta \ = \ D^J(\mathcal{E}) ( \eta \otimes \idop_Y)$.}
 	\end{enumerate}
\end{definition}

Condition $(C)$ ensures that we only consider genuine generalisations of the classical case, that is, it stipulates that the map in Eq.~\eqref{eq:classical-sot-like-map} is a special case of the above quantum notion---as always, once encoded via choosing a product ONB.

Condition $(M)$ may be considered the most quintessential one:\footnote{and the motivation for referring to  $\mathcal{E} \odot \eta$ as a `state over time' in the literature---thinking of a channel $\mathcal{E}$ as a time evolution map and of state $\eta$ as an initial state \label{footnote:meaning-sot}} 
just like classically, where the respective marginals of $P(Y|X)P(X)$ are $P(X)$ and $P(Y)$, the condition demands that $\mathcal{E} \odot \eta$---the supposed generalisation of a model's joint state---had better not just somehow encode correlations between states $\eta$ and their images $\mathcal{E}(\eta)$ under inference or generation, but also have both states be directly accessible and recoverable from $\mathcal{E} \odot \eta$. 

We will see concrete choices of \sot\ maps in Sec.~\ref{sec:sot-maps}, but for now continue developing the general framework. 
Crucially, given a state $\eta \in \states{X}$ and channel $\mathcal{E} \in \channels{X}{Y}$, the operator $\mathcal{E} \odot \eta$ may \emph{not} be a state---Def.~\ref{def:sot-map} does not require it and the later examples will indeed illustrate this feature.  
Yet, the whole point of having introduced \sot\ maps was to induce a notion of \emph{state} factorisation. 
Hence, one is led to the central notion of ambiguity.

\begin{definition} \label{def:sot-map-based-factorisation-of-state}
	Given \rleditc{a} state $\rho \in \states{X \otimes Y}$ and \sot\ map $\odot : \cpmaps{X}{Y} \times \operators{X} \rightarrow \operators{X \otimes Y}$, we will say that $\rho$ has a \emph{$(X,Y)$-factorisation wrt $\odot$} iff \ $\exists \mathcal{E} \in \channels{X}{Y}$ such that 
	$\rho = \mathcal{E} \odot \rho_X$. 
	We will refer to such states also as \emph{ambiguous states} and denote their collection as\footnote{Note the distinguished roles of factors is encoded in their order in subscripts, i.e. we distinguish $\ASset{\sotmap}{X}{Y}$ and $\ARCset{\sotmap}{X}{Y}{\rho}$ from $\ASset{\sotmap}{Y}{X}$ and $\ARCset{\sotmap}{Y}{X}{\rho}$, respectively, while we do not distinguish $\hilb{X} \otimes \hilb{Y}$ and $\hilb{Y} \otimes \hilb{X}$, i.e. we suppress the isomorphism between the latter.}
\begin{align}
	\ASset{\sotmap}{X}{Y} \ := \ \left\{\rho \in \states{X \otimes Y} \mid  \rho \text{ has } (X,Y)\text{-factorisation wrt} \ \odot   \right\}
\end{align}
and for any $\rho \in \states{X \otimes Y}$ the (possibly empty) set of \emph{ambiguity-realising channels} as 
\begin{align}
	\ARCset{\sotmap}{X}{Y}{\rho} := \left\{ \mathcal{E} \in \channels{X}{Y} \mid \rho = \mathcal{E} \odot \rho_X \right\}
\end{align}
\end{definition}
\rledit{The term `ambiguity' is thus owed to the fact that ambiguous states denote operators that are compatible with two distinct interpretations: they can be seen as ordinary states \emph{and} as `states over time' (see footnote \ref{footnote:meaning-sot}). It thus is an ambiguity in the same spirit as it applies to bipartite classical probability distributions, which are compatible with different spatio-temporal and causal relations between the two random variables. 

Suppose now} $\rho \in \ASset{\sotmap}{X}{Y}$, i.e. $\rho \ = \ \mathcal{E} \odot \rho_X$ for some channel $\mathcal{E}$. 
This is starting to look a lot like the generalisation of a classical model with inference, or generation, depending on which of $X$ and $Y$ are considered \visible\ and \hidden, respectively. 
However, for the pair $(\rho, \mathcal{E})$ to qualify as either it ought to meet at least one further condition, namely that the induced map 
\begin{align}
	\mathcal{E} \odot (\argline) \ & : \ \operators{X} \rightarrow \operators{X \otimes Y} \\
					& :: \ \alpha \mapsto \mathcal{E} \odot \alpha
\end{align}
map states to states again. 
We thus end up with the following basic notion of our framework. 

\begin{definition} \label{def:model-with-inf-or-gen}
	Given systems $X$ and $Y$, considered \visible\ and \hidden, respectively, a model class $\{\rho_{\theta}\}_{\theta \in I}  \subseteq \states{X \otimes Y}$ with $I \subseteq \mathbb{R}^n$ and an \sot\ map $\odot : \cpmaps{X}{Y} \times \operators{X} \rightarrow \operators{X \otimes Y}$, we say the model class has 
	\begin{itemize}
		\item \emph{inference} if $\forall \theta \in I$ the state is ambiguous as in 
		$\rho_{\theta} \in \ASset{\sotmap}{X}{Y}$ and 
		$\exists \infchansymbol_{\theta} \in \ARCset{\sotmap}{X}{Y}{\rho_{\theta}}$ such that 
		$\infchansymbol_{\theta} \odot (\argline)$ takes states to states, i.e. $\states{X} \hookrightarrow \states{X \otimes Y}$, in which case we refer to $\infchansymbol_{\theta}$ as an \emph{\infchan} and to $\infchansymbol_{\theta} \odot (\argline)$ as the \emph{\infmap};

		\item \emph{generation} if $\forall \theta \in I$ the state is ambiguous as in 
		$\rho_{\theta} \in \ASset{\sotmap}{Y}{X}$, and 
		$\exists \genchansymbol_{\theta} \in \ARCset{\sotmap}{Y}{X}{\rho_{\theta}}$ such that $\genchansymbol_{\theta} \odot (\argline)$ takes states to states, i.e. $\states{Y} \hookrightarrow \states{X \otimes Y} $, in which case we refer to $\genchansymbol_{\theta}$ as a \emph{\genchan} and to $\genchansymbol_{\theta} \odot (\argline)$ as the \emph{\genmap}.		
	\end{itemize}
\end{definition}

The requirement that $\rho_{\theta}$ is not just ambiguous, but also such that the \infchan\ $\infchansymbol_{\theta}$---which as a channel of course is positive---has a corresponding extended map $\infchansymbol_{\theta} \odot (\argline)$ into the joint system that maps states to states again\footnote{In case the extended map is linear this of course is equivalent to positivity and trace-preservation.}  will be very consequential; analogously for the generative direction with $\genchansymbol_{\theta}$ and $\genchansymbol_{\theta} \odot (\argline)$. That a good notion of model with inference (generation) ought to satisfy this condition of `positivity on all input states' stems from the fact that, just like classically, it should make sense---not least for training---to consider inference (generation) on any input state, while also the joint object that is constructed via the \sot\ map and encodes the `correlations', is well-defined. The next section will expand on this point. 

Now, one might naturally wonder why follow our seemingly involved approach, which starts from quantum generalisations of the obvious model components, i.e. inference channels and joint states, and which then faces the challenges of when and how this data is consistent, leading to the business of ambiguous states and `positive' extended inference and generation maps; why not simply start from a channel $\mathcal{C} \in \channels{X}{X \otimes Y}$, which by definition is positive and then simply declare its marginal $\infchansymbol := \tr_X \circ \mathcal{C}$ to be the inference channel, thereby ensuring that, by construction, the desired marginal condition on $Y$ holds, i.e. that for all states $\eta$ we have $\infchansymbol(\eta) = \tr_X[ \mathcal{C}(\eta)]$? 
First of all, such an approach breaks conceptually with aiming for a structurally as analogous as possible quantum generalisation of classical representation learning. 
Second, Prop.~\ref{prop:no-braodcasting-kind-of-fact} amounts to a no-go-result for such an approach: if $\mathcal{C}$ also satisfies the marginal condition $\tr_Y \circ \mathcal{C} = \idchan_X$, \rleditb{then this \infmap{} $\mathcal{C}$ is rendered trivial and the model's inference channel $\infchansymbol$ necessarily is a discard-and-prepare-a-fixed-state channel, i.e. it would be entirely useless for inference.}  
An analogous question and argument can of course be given for the generative direction.  

Note that this analysis in turn also means that the \infmap{s} and \genmap{s} in our framework cannot be expected to be CPTP maps. This is a central observation, which we will strengthen and expand on in Sec.~\ref{sec:general-results-sot-ambiguity}.

\subsection{Towards representation learning} \label{subsec:general-training-setups}

Before turning to concrete \sot\ map constructions and corresponding model classes in Sections~\ref{sec:sot-maps}-\ref{sec:main-section-on-LS-models}, let us sketch the overall picture we have arrived at and see how the notions coined thus far indeed yield quantum generalisations of the standard approaches of (unsupervised) representation learning from Sec.~\ref{subsec:intro-classical-case}.

\paragraph{Training: exact inference.}
Let $\sigma\in\states{X}$ be given quantum data and $\{\rho_{\theta}\}_{\theta \in I}$ for some $I \subseteq \mathbb{R}^n$ and some $n \in \mathbb{N}$ a model class with inference with respect to $\odot$, denoting the corresponding \infchan\ as $\infchansymbol_{\theta}$. 
Then each \infmap\ $\infchansymbol_{\theta} \odot (\argline)$ is in particular a data extension map---this is ensured by condition (M) of Def.~\ref{def:sot-map} together with the fact that Def.~\ref{def:model-with-inf-or-gen} requires $\infchansymbol_{\theta} \odot (\argline)$ to map states to states. 
Hence, one may use \infmap{s} to study what becomes the quantum analogue of \emph{exact inference}, that is, training the model through the minimisation of 
\begin{align}
	 D\big( \infchansymbol_{\theta^t} \odot \sigma \ \| \ \rho_{\theta} \big) ,
\end{align}
with respect to $\theta$ and at time $t$ and for some information-theoretical function $D$ quantifying divergence, such as the quantum relative entropy. This is the quantum generalisation of Eq.~\eqref{objective_em}.

\paragraph{Training: approximate inference.} 
Just like classically, data extension maps do not necessarily have to be induced by the (current) model itself and the above definitions allow us to also say what a quantum version of \emph{approximate inference} would be.  
Let again $\sigma$ be given quantum data and $\{\rho_{\theta}\}_{\theta \in I}$ a model class with inference with respect to $\odot$; in addition let $\{\Phi_{\phi}\}_{\phi \in I'}$ for some $I' \subseteq \mathbb{R}^m$ and $m \in \mathbb{N}$ be a family of data extension maps. 
The quantum analogue of an approximate inference training setup then is the minimisation with respect to both $\phi$ \emph{and} $\theta$ of~\footnote{NB here $\Phi_{\phi}$ could of course be given by the extended inference maps of a further model with inference, but in principle also by a completely different approach.}
\begin{align}
	 D_1\big( \Phi_{\phi}( \sigma) \ \| \ \rho_{\theta^t} \big) \qquad \text{and} \qquad D_2\big( \Phi_{\phi^t}( \sigma) \ \| \ \rho_{\theta} \big)
\end{align}
at time $t$ and for two (possibly equal) divergences $D_1$ and $D_2$. This is the quantum generalisation of Eq.~\eqref{objective_var}.

\paragraph{Trainability and non-CPTP data extension maps.}
So very generally in the framework, training a model class $\{\rho_{\theta}\}_{\theta}$, given quantum data $\sigma$, is taken to refer to the minimisation of an objective function of the form
\begin{align}	
	D\big( \Phi_{\phi}( \sigma) \ \| \ \rho_{\theta} \big) \label{eq:generic-minimisation-quantity}
\end{align}
with respect to parameters $\theta$ and $\phi$, for some choice of divergence $D$ such as the quantum relative entropy and some family of data extension maps $\Phi_{\phi}$.\footnote{Recalling Sec.~\ref{subsec:intro-the-challenge}, we always assume that $ supp\big(\Phi_{\phi}( \sigma) \big) \subseteq supp(\rho_{\theta})$.} 
Such minimisation will involve estimating gradients of Eq.~\eqref{eq:generic-minimisation-quantity}. 
Given this work's scope, not focussing on the practical side, we will not develop all the technical details needed if one were to actually run any experiments; we are laying the basic ground and framework for that. 
However, see App.~\ref{app:gradient_qre} for the generic form of some of the needed gradients. 

Recalling Prop.~\ref{prop:no-braodcasting-kind-of-fact} data extension maps cannot be expected to be CPTP; in fact, as we will establish later, not even to be linear. 
Now, by definition, the extended operator, $\Phi_{\phi}( \sigma)$, is required to be a state, but how to implement the training procedure if $\Phi_{\phi}$ is not an implementable physical map?
The reason this is not an in-principle obstruction to trainability is that one does not actually ever need to get hold of $\Phi_{\phi}( \sigma)$---estimating gradients means estimating functionals of states; and often they can be written as estimating expectation values for some observable on $\hilb{X \otimes Y}$ with respect to $\Phi_{\phi}( \sigma)$. 
If given enough data samples, drawn i.i.d. from a data source, then for certain non-CPTP or non-linear maps $\Phi_{\phi}$ it is at least in principle possible to approximate the desired gradients using standard techniques.   
See App.~\ref{app:gradient_qre} for more details on this.

\paragraph{Using a model: inference and generation.} 
As just remarked, that an \infmap\ $\infchansymbol_{\theta} \odot (\argline)$ is non-CPTP or even non-linear, is not a priori an issue, because training does not require that one can deterministically prepare the extended joint state $\infchansymbol_{\theta} \odot \sigma$, only enough statistics to estimate gradients. 
However, after training has converged, or in any case given some fixed model $\rho_{\theta}$ with inference and generation, one certainly would demand to be able to \emph{do} inference and  \emph{do} generation, that is, output \rleditb{deterministically---in a one-shot manner---the} representation on the latent space given any state $\alpha\in\states{X}$, or conversely generate a `synthetic data' state on $X$ given any state $\beta\in\states{Y}$ on the latent space.  
Indeed, by definition the marginal of an \infmap\ is the \infchan\ $\infchansymbol_{\theta}$, i.e. a CPTP map and hence, in principle  $\infchansymbol_{\theta}(\alpha)$ can always be computed. Analogously, the marginal of the \genmap\ is the \genchan\ $\genchansymbol_{\theta}$, also CPTP, and thus $\genchansymbol_{\theta}(\beta)$ is computable.

\paragraph{Generative decomposition for more efficient training.}
Finally we note that, just like it is the case classically, using a decomposition in the generative direction, i.e. via an extended generation map may in some cases simplify the parametrisation of the model and its training. 
For instance, if $\{\rho_{\theta}\}_{\theta \in I}$ is a model with inference and generation (wrt $\odot$), training via `exact inference' may take the form of minimising 
$D\big( \infchansymbol_{\theta^t} \odot \sigma \ \| \ \genchansymbol_{\theta^t} \odot  \rho_{Y;\theta} \big)$,  where $\genchansymbol_{\theta}$ is the corresponding \genchan. 
Analogously for `approximate inference'. 
In fact the model may be thought of as defined in the form $\genchansymbol_{\theta} \odot \rho_{Y;\theta}$ to start with.

\section{Examples of state-over-time maps} \label{sec:sot-maps}

The previous section stipulated in Def.~\ref{def:sot-map} the notion of \sot{} maps and argued why it is important for quantum models with inference or generation.  
In the following we will give, and start discussing, a selection of specific examples. 
To this end a little more formalism is needed.  

First, the two common choices of an isomorphism for `channel-state duality' will be useful. 
Recall that the Jamio{\l}kowski isomorphism~\cite{Jamiolkowski1972} 
    \begin{align}
    	D^J \ : \ \linmaps{X}{Y} \rightarrow \operators{X \otimes Y} \quad :: \quad 
        \mathcal{E} \quad \mapsto \quad  & D^J(\mathcal{E}) :=  \sum_{i,j} \dyad{j}{i} \otimes \mathcal{E}(\dyad{i}{j})  \\ 
        (D^J)^{-1} \ : \ \operators{X \otimes Y}  \rightarrow \linmaps{X}{Y} \quad :: \quad
        A \quad \mapsto \quad  &  \tr_{X} \left[ A \ ( \cdot \otimes \idop_Y) \right]
    \end{align}
    yields a basis-independent representation $D^J(\mathcal{E})$ for every linear map $\mathcal{E} : \operators{X} \rightarrow \operators{Y}$, though for a CP map $\mathcal{E} \in \cpmaps{X}{Y}$ the operator $D^J(\mathcal{E})$ is in general not positive semi-definite, $D^J(\mathcal{E}) \nsucceq 0$. 
In turn, given some choice of ONB of $\hilb{X}$, the Choi isomorphism~\cite{Choi1975} 
    \begin{align}
    	D^C \ : \ \linmaps{X}{Y} \rightarrow \operators{X \otimes Y} \quad :: \quad
        \mathcal{E} \quad \mapsto \quad  &  D^C(\mathcal{E}) :=  \sum_{i,j} \dyad{i}{j} \otimes \mathcal{E}(\dyad{i}{j}) \label{eq:def-Choi-iso} \\ 
        (D^C)^{-1} \ : \ \operators{X \otimes Y}  \rightarrow \linmaps{X}{Y} \quad :: \quad
        A \quad \mapsto \quad  &  \tr_{X} \left[ A \ \left( \ ( \cdot )^T \otimes \idop_Y \right) \right]
    \end{align}
 	    where $(\cdot)^T$ represents transposition with respect to the given ONB, yields a basis-dependent representation  $D^C(\mathcal{E})$, which is positive semi-definite $D^C(\mathcal{E}) \succeq 0$ iff the map is CP, $\mathcal{E} \in \cpmaps{X}{Y}$.  
 	Note that for all $\mathcal{E} \in \cpmaps{X}{Y}$ it holds that $\tr_Y\big[D^J(\mathcal{E})\big] = \idop_X = \tr_Y\big[D^C(\mathcal{E})\big]$ iff $\mathcal{E}$ is also trace-preserving, i.e. $\mathcal{E} \in \channels{X}{Y}$. 
    Also note that for every linear map $\mathcal{E}$ the two representations are related by $D^J(\mathcal{E})  = \mathcal{T}_X \big( D^C(\mathcal{E}) \big)$, and equivalently by $\mathcal{T}_X \big( D^J(\mathcal{E}) \big) = D^C(\mathcal{E})$, where $\mathcal{T}_X$ is partial transposition with respect to the basis chosen for $D^C$.  
See, e.g., \cite{frembs2024variations}  for an overview of such isomorphisms and further details. 

Second, and as is common in the context of states over time, we will refer to a map of the form
\begin{align}
	\star : \operators{H} \times \operators{H} \rightarrow \operators{H} \label{eq:def-satr-product}
\end{align}
as a \emph{star product}. 
Relevant examples are the following with $\alpha,\beta \in \operators{H}$:\footnote{Among other constructions considered in the literature, though not studied in this work, is e.g. $\alpha \star^{(n)} \beta :=  \left(\beta^{\frac{1}{2n}} \alpha^{\frac{1}{n}} \beta^{\frac{1}{2n}}\right)^n$ and in particular, $\alpha \star^{CA} \beta := \lim_{n \rightarrow \infty} \alpha \star^{(n)} \beta $ \cite{Cerf_1997, Warmuth_2005, Leifer_2013}. 
NB this generalises LS seeing as $\alpha \star^{\LSshort} \beta = \alpha  \star^{(1)} \beta$.}  
\begin{align}
	\alpha \star^{\MPshort} \beta \quad := \quad & \alpha \beta \\
   \alpha \star^{\rMPshort} \beta \quad := \quad & \beta  \alpha \\
   \alpha \star^{\JPshort} \beta \quad := \quad & \frac{1}{2} \left\{\alpha, \beta \right\} = \frac{1}{2}  \big(\alpha \beta + \beta  \alpha \big) \\ 
    \alpha \star^{\LSshort} \beta \quad  := \quad & \beta^{\frac{1}{2}} \alpha \beta^{\frac{1}{2}} 
\end{align}
So \rleditc{$\star^{\MPshort}$} is just the ordinary operator product, i.e. the composition that defines the algebra $\operators{H}$ in the first place and \rleditc{$\star^{\rMPshort}$} the same just with the reverse order. 
Note that for reasons of simplicity and an analogous treatment we presented $\star^{\LSshort}$ as a map of the form as in Eq.~\eqref{eq:def-satr-product}, but suppressed the choices such definition strictly speaking involves (concerning where square roots exist but are non-unique or do not exist at all), since they will not matter in this work. 

We now have all the ingredients to state the examples of \sot{} maps this work focuses on:
\begin{align}
	\mathcal{E} \MPodot \eta \quad  := \quad & D^J(\mathcal{E}) \ \star^{\MPshort} \ (\eta \otimes \idop_Y)  
	\ = \ D^J(\mathcal{E}) (\eta \otimes \idop_Y)  \label{eq:def-MPodot} \\
    \mathcal{E} \rMPodot \eta \quad  := \quad & D^J(\mathcal{E}) \ \star^{\rMPshort} \ (\eta \otimes \idop_Y) 
    \ = \ (\eta \otimes \idop_Y) D^J(\mathcal{E})   \label{eq:def-rMPodot} \\
    \mathcal{E} \JPodot \eta \quad  := \quad & D^J(\mathcal{E}) \ \star^{\JPshort} \ (\eta \otimes \idop_Y) 
    \ = \ \frac{1}{2} \left\{D^J(\mathcal{E}), (\eta \otimes \idop_Y) \right\}  
     \label{eq:def-JPodot}  \\
     \mathcal{E} \LSodot \eta \quad  := \quad &  D^J(\mathcal{E}) \ \star^{\LSshort} \ (\eta \otimes \idop_Y) 
     \ = \  (\eta \otimes \idop_Y)^{\frac{1}{2}} D^J(\mathcal{E}) (\eta \otimes \idop_Y)^{\frac{1}{2}} \ ,  \label{eq:def-LSodot}     
\end{align} 
Going forward we will use our convention again and for the respective right-hand sides above write $D^J(\mathcal{E}) \eta$, $D^J(\mathcal{E}) \star^{JP} \eta$ and $D^J(\mathcal{E}) \star^{LS} \eta$ etc. 
That $\LSodot$, $\JPodot$, $\MPodot$ and $\rMPodot$ indeed define \sot\ maps is easy to verify (also see  \cite{Horsman_2017}).

This choice of examples will turn out to serve us well to expose the intersection between \sot\ maps and our framework of quantum models, while a comprehensive study of the landscape of all \sot\ maps with regards to quantum models is left for future work.  
It is worth noting already some of the basic facts that distinguish the chosen ones, such as concerning linearity, hermiticity-preservation and positivity.\footnote{A map $f : \operators{H} \rightarrow \operators{H'}$ is  hermiticity-preserving iff for all $A \in \Herm{H}$ it holds that $f(A) \in \Herm{H'}$; it is positive (also say `positivity-preserving') iff $\forall A \in \operators{H}$ such that $A \succeq 0$ it holds that $f(A) \succeq 0$. 
NB for (non-)\allowbreak{linear} maps positivity does (not) imply hermiticity-preservation. 
For star-products $\star : \operators{H} \times \operators{H} \rightarrow \operators{H}$, naturally, by hermiticity-preservation we mean $A \star B \in \Herm{H}$ for all $A,B \in \Herm{H}$ and analogously, by positivity that $A \star B \succeq 0$ for all $A,B \succeq 0$.}

\begin{proposition} \label{prop:basic-properties-star-products} 
The above star-products and respective \sot\ maps have the following properties (with $\mathcal{E} \in \channels{X}{Y}$). 
	\begin{enumerate}[label=(\arabic*), leftmargin=1.3cm]
		\item Ordinary operator multiplication: 
		$\star^{\MPshort}$ is bilinear but not hermiticity-preserving. 
		In particular a map $\mathcal{E} \MPodot (\argline) : \operators{X} \rightarrow \operators{X \otimes Y}$ is linear and trace-preserving, but generally not hermiticity-preserving. Similarly for $\star^{\rMPshort}$.		

		\item The Jordan-product: $\star^{JP}$ is bilinear and hermiticity-preserving, but not positive.  
		In particular a map $\mathcal{E} \JPodot (\argline) : \operators{X} \rightarrow \operators{X \otimes Y}$ is linear, trace-preserving and hermiticity-preserving, but generally not positive.
		
		\item Leifer-Spekkens: $\star^{LS}$ is linear in the first, but non-linear in the second argument; 
		it is hermiticity-preserving and positive. 
		In particular the map $\mathcal{E} \LSodot (\argline) : \operators{X} \rightarrow \operators{X \otimes Y}$ is non-linear, hermiticity-preserving and trace-preserving; it is positive iff $D^J(\mathcal{E})$ is. 
	\end{enumerate}
\end{proposition}
\begin{proof} 
	Mostly obvious, but for completeness see \hyperref[proof:prop:basic-properties-star-products]{Proof in App.~B}. 
\end{proof}

Finally, the chosen \sot\ maps induce corresponding notions of ambiguous states (see Def.~\ref{def:sot-map-based-factorisation-of-state}), and Sec.~\ref{sec:ambiguous-states} will study precisely the sets $\ASset{\odot}{X}{Y}$ for $\odot \in \{\LSodot, \JPodot, \MPodot, \rMPodot\}$. 
To avoid clutter we will usually denote these sets as $\ASset{sot}{X}{Y}$ for $sot \in \{ \LSshort, \JPshort, \MPshort, \rMPshort \}$ and at times, when the direction of the factorisation $(X,Y)$ vs $(Y,X)$ is clear from context, also refer to a state simply as \LSshort{}-, \JPshort{}-, \MPshort{}- and \rMPshort{}-ambiguous.

\section{Non-trivial models need non-linear maps} \label{sec:general-results-sot-ambiguity}

This section will expose a few general facts and consequences of the framework from Sec.~\ref{sec:the-framework}, in particular that non-linearity is an inevitable feature of it. 
To this end note the following generalisation of Prop.~\ref{prop:no-braodcasting-kind-of-fact} to positive, that is not necessarily completely positive, maps.

\begin{lemma}[See {\cite[Appendix]{jordan2004dynamics}}] \label{lem:no-braodcasting-kind-of-fact-for-any-positive-map}
	Let $\mathcal{E} \in \linmaps{X}{X \otimes Y}$ be a positive and trace-preserving linear map such that $\tr_Y \circ \mathcal{E}  = \idchan_X$ then $\mathcal{E} = \idchan_X \otimes \eta$ for some $\eta \in \operators{Y}$. 
\end{lemma}
 
Next note a complementary fact concerning the following class of `trivially ambiguous' states, which always exists for any \sot\ map. 

\begin{lemma} \label{lem:general-fact-product-states}
	Let $\odot : \cpmaps{X}{Y} \times \operators{X} \rightarrow \operators{X \otimes Y}$ be an \sot\ map. 
	Then for $\rho \in \ASset{\sotmap}{X}{Y}$ the following are equivalent:  
	\begin{enumerate}[label=(\arabic*), leftmargin=2cm]
		\item it is a product state, i.e. $\rho = \rho_X \otimes \rho_Y$;  
		\item $\exists \mathcal{E} \in \ARCset{\sotmap}{X}{Y}{\rho}$ such that $\mathcal{E} = \rho_Y \circ \tr_X$, i.e. $\mathcal{E}$ is discard-and-prepare; 
		\item $\exists  \mathcal{E} \in \ARCset{\sotmap}{X}{Y}{\rho}$ such that $\mathcal{E} \odot (\argline) = \idchan_X \otimes \rho_Y$. 
	\end{enumerate}
	Moreover, for any \sot\ map $\odot$ all product states are indeed ambiguous, that is, $\forall \alpha \in \states{X}, \beta \in \states{Y}$ it holds that $\alpha \otimes \beta \in \ASset{\sotmap}{X}{Y}$. 
\end{lemma}
\begin{proof}
	See \hyperref[proof:lem:general-fact-product-states]{Proof in App.~B}. 
\end{proof}
 
Seeing as they are all in correspondence with each other we will use the attribute \emph{trivial} for all of the following objects: a state $\rho$ if it is a product state, an inference (generation) channel $\mathcal{E}$ if it is discard-and-prepare and an extended inference (generation) map $\mathcal{E} \odot (\argline)$ if it factorises into $\idchan_X$ and a fixed state. We refer to the respective objects as \emph{non-trivial} otherwise. 
NB trivialness of $\mathcal{E}$ or $\mathcal{E} \odot (\argline)$ implies that of a corresponding ambiguous state $\rho$, though not the converse---in general, not every ambiguity realising channel of a product state will be trivial. 
Though in combination with the first observation above in Lem.~\ref{lem:no-braodcasting-kind-of-fact-for-any-positive-map} such non-trivial channels cannot yield positive, yet linear extended maps. Hence, we conclude the following general fact. 

\begin{theorem} \label{thm:no-go-linear} 
	For any \sot\ map $\odot : \cpmaps{X}{Y} \times \operators{X} \rightarrow \operators{X \otimes Y}$ 
	and $\rho \in \ASset{\sotmap}{X}{Y}$ if $\mathcal{E} \odot (\argline)$ is linear and positive for some  $\mathcal{E} \in \ARCset{\sotmap}{X}{Y}{\rho}$ then $\mathcal{E}$ and $\rho$ are trivial in the above sense. 
	In other words, there are no quantum models with inference (generation) such that the inference (generation) channel is non-trivial, but the \emph{extended}  inference (generation) map is linear. 
\end{theorem}
\begin{proof}
	Suppose $\rho \in \ASset{\sotmap}{X}{Y}$ and $\mathcal{E} \in \ARCset{\sotmap}{X}{Y}{\rho}$ such that 
	$\mathcal{E} \odot (\argline)$ is linear and positive. 
	This together with (M) of Def.~\ref{def:sot-map} then implies, by Lem.~\ref{lem:no-braodcasting-kind-of-fact-for-any-positive-map}, that 
	$\mathcal{E} \odot (\argline) = \idchan_X \otimes \beta$ for some $\beta \in \states{Y}$. 
	By Lem.~\ref{lem:general-fact-product-states} it therefore must be the case that indeed $\rho = \rho_X \otimes \beta$. 	
\end{proof}

With the condition on an \infmap\ to map states to states being a requirement in our framework, non-trivialness of the model is traded for non-linear \infmap{s}---it is the linearity that has to be given up. 
See Sec.~\ref{subsec:general-training-setups} for why this is not an issue in principle. 
Any \sot\ map that is based on a star product $\starprod$, which is linear in the second argument, can therefore never yield non-trivial models. 
Recalling Prop.~\ref{prop:basic-properties-star-products} this thus rules out \JPshort, \MPshort\ and \rMPshort\ as interesting constructions, but leaves \LSshort\ as a candidate. 
While the above theorem identifies non-linearity of the extended inference (generation) maps as a necessary condition, it is left for future work to identify a sufficient condition for \sot\ maps that yield non-trivial models.  
That at least one such \sot\ map exists is demonstrated in Sec.~\ref{sec:main-section-on-LS-models}. 

\rledit{Finally note that Lem.~\ref{lem:no-braodcasting-kind-of-fact-for-any-positive-map} also strengthens the conclusions from Prop.~\ref{prop:no-braodcasting-kind-of-fact} for data extension maps in general, that is, regardless of whether they are given by an \infmap\ or not, any data extension map $\Phi \ : \ \operators{X} \ \rightarrow \ \operators{X \otimes Y}$ in the sense of Def.~\ref{def:extension-map} that is linear has to be \rleditb{trivial in that} $\Phi = \idchan_X \otimes \eta$ for some $\eta \in \states{Y}$.}

\section{Ambiguous states} \label{sec:ambiguous-states}

This section gives complete characterisations of the respective classes of ambiguous states for all \sot\ maps this work considers. 
Now, Sec.~\ref{sec:general-results-sot-ambiguity} already concluded that whatever the ambiguous states with respect to \MPshort, \rMPshort\ and \JPshort\ are, they can only lead to trivial models. 
Nonetheless we will present the characterisation of the respective ambiguous states for they are closely related to those of \LSshort\ and in a way that makes their exposition of pedagogical value. 
Also App.~\ref{app:algorithm} provides additional motivation, and of course one may also be interested in \sot\ maps and ambiguity independently from quantum representation learning.

\subsection{Some useful facts} \label{subsec:useful-facts}

We start with a few useful concepts and facts that will make frequent appearance throughout. 

First, recall that a state $\rho \in \states{X \otimes Y}$ satisfies the \emph{positive partial transpose} (PPT) condition with respect to the $X,Y$ factorisation iff $\mathcal{T}_X(\rho) \succeq 0$, where  $\mathcal{T}_X : \operators{X \otimes Y} \rightarrow \operators{X \otimes Y}$ denotes the \emph{partial transpose} with respect to some fixed ONB of $\hilb{X}$. 
Note though that the PPT condition is independent from the choice of ONB. 
Moreover it is symmetric in the underlying factorisation, i.e. $\mathcal{T}_X(\rho) \succeq 0$ iff $\mathcal{T}_Y(\rho) \succeq 0$. See Sec.~\ref{subsec:ambiguity-summary} for a few details on the significance of this important notion in quantum-information theory; also see, e.g., \cite{Horodecki2009, Horodecki1997}.
We will often refer to a state that satisfies the PPT condition as `being PPT' and we let `PPT' also denote the set of all PPT states (wrt a fixed tensor product space)---the context will always make it unambiguous what is being referred to.

\begin{lemma} \label{lem:psd-corrected}
For $A,B \in \operators{X}$ with $A \succeq 0$, it holds that $B \succeq 0$ implies $A^{1/2} B A^{1/2}  \succeq 0$ and conversely, $A^{1/2} B A^{1/2}  \succeq 0$ implies $\pi_A B \pi_A \succeq 0$, where $\pi_A$ is the projector onto the support of $A$.  
\end{lemma}
\begin{proof} See \hyperref[proof:lem:psd-corrected]{Proof in App.~B}. 
\end{proof}

\begin{lemma} \label{lem:transpose-commutativity} 
	For $A \in \operators{X}$ Hermitian and $\mathcal{T}_X : \operators{X \otimes Y} \rightarrow \operators{X \otimes Y}$ the partial transpose on $X$ with respect to $A$'s eigenbasis, $\mathcal{T}_X$ commutes with the map $\Phi_{f(A)}(\argline) = (f(A) \otimes \idop_Y) (\argline) (f(A)^\dag \otimes \idop_Y)$ for any function $f : D \rightarrow \mathbb{C}$ with $D \subseteq \mathbb{R}$ such that $f$ is defined on the eigenvalues of $A$, i.e. $\mathcal{T}_X \circ \Phi_{f(A)} = \Phi_{f(A)} \circ \mathcal{T}_X$.
\end{lemma}
\begin{proof} See \hyperref[proof:lem:transpose-commutativity]{Proof in App.~B}. 
\end{proof}

The following is easy to verify, e.g., by plugging $\mathcal{T}_X(A) = \sum_{i,j} (\dyad{i}{j} \otimes \idop_Y) A (\dyad{i}{j} \otimes \idop_Y)$ into the corresponding expressions.
\begin{proposition}
\label{prop:partial-transpose-trace-properties}
Partial transposition and partial trace with respect to the same factor satisfy the following property: $\forall A,B \in \operators{X \otimes Y}$ $\tr_X( \mathcal{T}_X(A) B ) = \tr_X( A \mathcal{T}_X(B))$; 
and with respect to different factors they commute, that is, $\tr_Y( \mathcal{T}_X(A)) = \mathcal{T}_X(\tr_Y(A))$.
\end{proposition}

Finally, the following fact will make frequent appearance since we do not assume anything about a state's support. 
\begin{lemma} \label{lem:support-inclusion-property}
	Let $\rho \in \operators{X \otimes Y}$ such that $\rho \succeq 0$. Then $supp(\rho) \subseteq supp(\rho_X \otimes \idop_Y)$.
\end{lemma}
\begin{proof} See \hyperref[proof:lem:support-inclusion-property]{Proof in App.~B}. 
\end{proof}

This implies in particular that for any $\rho \in \states{X \otimes Y}$ we have that not only $\pi_X \rho_X \pi_X = \rho_X$, but also $\pi_X \rho \pi_X = \rho$, which is, as usual, shorthand for $(\pi_X \otimes \idop_Y ) \rho (\pi_X \otimes \idop_Y )= \rho$. 

\begin{proposition} \label{prop:positive-DJ-equiv-to-PPT}
	CP maps $\mathcal{E} \in \cpmaps{X}{Y}$ with  $D^J(\mathcal{E}) \succeq 0$ and PPT states $\eta \in \states{X \otimes Y}$ are in correspondence with one another, modulo a choice of ONB of $\hilb{X}$ and constant $N > 0$, via $D^J(\mathcal{E}) \ = \ N \mathcal{T}_X \big( \eta \big)$. 
	More precisely, 
	\begin{enumerate}
		\item Given $\mathcal{E} \in \cpmaps{X}{Y}$ with $D^J(\mathcal{E}) \succeq 0$ (and a choice of ONB), 
		$\eta := \big(1 / \tr [D^C(\mathcal{E})]\big) D^C(\mathcal{E})$ yields a PPT state, for which indeed  $D^J(\mathcal{E}) \ = \ N \mathcal{T}_X \big( \eta \big) $ with $N:= \tr[D^J(\mathcal{E})]$. 
		\item Given PPT state $\eta \in \states{X \otimes Y}$ and $N > 0$ (and a choice of ONB), 
		$D^J(\mathcal{E}) \ := \ N  \mathcal{T}_X \big( \eta \big)$ defines a CP map $\mathcal{E}$ with $D^J(\mathcal{E}) \succeq 0$. 
		\item Channels with $D^J(\mathcal{E}) \succeq 0$ correspond, via $D^J(\mathcal{E}) = d_X \mathcal{T}_X \big( \eta \big)$, precisely to PPT states $\eta$ with maximally mixed marginals, $\eta_X = \frac{1}{d_X} \idop_X$  (and a choice of ONB). 
	\end{enumerate}	
\end{proposition}
\begin{proof} 
	Given all the details above, the proof of the statement is immediate. 
\end{proof}

\subsection{{\protect\MPcapital} ambiguity} \label{subsec:MP-ambiguity}

We start with arguably the simplest form of an \sot\ map, namely $\MPodot$, where the star product is given by ordinary operator composition. 
Let us say a state $\rho \in  \states{X \otimes Y}$ is \emph{\CXPPT} iff $\rho$ is PPT and $[\rho,\rho_X ]=0$. 
It turns out that this condition completely captures the ambiguous states with the appearance of this kind of commutation relation being perhaps unsurprising seeing as the product of two PSD operators is PSD iff they commute. 

\begin{theorem} \label{thm:OPA-equivalence}
	Let $\rho \in  \states{X \otimes Y}$, then the following are equivalent 
	\begin{enumerate}[label=(\arabic*), leftmargin=2cm]
		\item $\rho \in \ASset{\MPshort}{X}{Y}$.
		\item $\rho \in \ASset{\rMPshort}{X}{Y}$.
		\item $\rho$ is \CXPPT. 
	\end{enumerate}
\end{theorem}
\begin{proof}
	See \hyperref[proof:thm:OPA-equivalence]{Proof in App.~B}. 
\end{proof}
Unsurprisingly, \MPshort\ and \rMPshort\ are the same as far as ambiguous states are concerned. 
Henceforth we will often only refer to \MPshort, but it is understood that any statement about it applies analogously to \rMPshort. 

One observes the following properties, which we shall see hold for all notions of ambiguity in this work, in particular: ambiguity-realising channels are unique (on the support of $\rho_X$) and they have non-negative Jamio{\l}kowski representation. 

\begin{proposition}\label{prop:unique-CP-map-for-OP}
	Let $\rho \in \ASset{\MPshort}{X}{Y}$, then 
	\begin{enumerate}[label=(\arabic*), leftmargin=2cm]
		\item There is a unique $\MPCPmap{\rho} \in \cpmaps{X}{Y}$ such that $\forall \mathcal{E} \in \ARCset{\MPshort}{X}{Y}{\rho}$ it holds that $\mathcal{E} \circ \Pi_X = \MPCPmap{\rho}$, 
		or equivalently 
		$\pi_X D^J(\mathcal{E}) \pi_X \ = \ D^J(\MPCPmap{\rho})$; 
		and $\{ \MPCPmap{\rho} + \mathcal{E}' \circ \Pi_X^{\perp} \mid \mathcal{E}' \in \channels{X}{Y} \} \subseteq \ARCset{\MPshort}{X}{Y}{\rho}$.
		\item $D^J(\MPCPmap{\rho} ) = \rho \rho_X^{-1}$. 
		\item $D^J(\MPCPmap{\rho}) \succeq 0 $.
	\end{enumerate}
\end{proposition} 

\begin{proof}
	See \hyperref[proof:prop:unique-CP-map-for-OP]{Proof in App.~B}.
\end{proof}

\subsection{Leifer-Spekkens ambiguity} \label{subsec:LSA-ambiguity}

Next we turn to LS-based ambiguity, which interestingly is captured precisely by the PPT condition.  

\begin{theorem} \label{thm:LSA-PPT-theorem}
Let $\rho \in  \states{X \otimes Y}$ then  $\rho \in \ASset{\LSshort}{X}{Y}$ iff it is PPT. 
\end{theorem}
\begin{proof} 
	See \hyperref[proof:thm:LSA-PPT-theorem]{Proof in App.~B}. 
\end{proof}
Further discussion to contextualise this fact is postponed to the summary in Sec.~\ref{subsec:ambiguity-summary}.  
Again one finds that ambiguity-realising channels are basically unique (on the support of $\rho_X$) and that they have non-negative Jamio{\l}kowski states.  

\begin{proposition} \label{prop:basic-LSA-propeties}
	Let $\rho \in \ASset{\LSshort}{X}{Y}$.  Then 
	\begin{enumerate}[label=(\arabic*), leftmargin=2cm]
		\item There is a unique $\LSCPmap{\rho} \in \cpmaps{X}{Y}$ such that  
		$\forall \mathcal{E} \in \ARCset{\LSshort}{X}{Y}{\rho}$ it holds that 
		$\mathcal{E} \circ \Pi_X = \LSCPmap{\rho}$, or equivalently 
		$\pi_X D^J(\mathcal{E}) \pi_X \ = \ D^J(\LSCPmap{\rho})$; 
		and $\{ \LSCPmap{\rho} + \mathcal{E}' \circ \Pi_X^{\perp} \mid \mathcal{E}' \in \channels{X}{Y} \} \subseteq \ARCset{\LSshort}{X}{Y}{\rho}$.
		\item $D^J(\LSCPmap{\rho}) \ = \ \rho_X^{-1/2} \rho \ \rho_X^{-1/2}$.
		\item $D^J(\LSCPmap{\rho}) \succeq 0$.
	\end{enumerate}
\end{proposition} 
\begin{proof} 
	See \hyperref[proof:prop:basic-LSA-propeties]{Proof in App.~B}. 
\end{proof}

\subsection{Jordan-product-based ambiguity} \label{subsec:Jordan-product-based-ambiguity}

Finally, we discuss ambiguity with respect to \JPshort. Although the main characterisation result has appeared before in \cite{song2025bipartite}, we restate it here including a proof for reasons of self-containment and since our proof methods for all \sot\ maps naturally complement each other. 

First of all, note that the \sot\ map  $\JPodot$ is a linear combination of $\MPodot$ and $\rMPodot$, namely their average, $\mathcal{E} \JPodot \eta = \frac{1}{2} \big( \mathcal{E} \MPodot \eta + \mathcal{E} \rMPodot \eta \big)$. Even though the two summands individually fail to preserve hermiticity, their sum does preserve it (also see Prop.~\ref{prop:basic-properties-star-products}). Intuitively speaking, there can hence be more ambiguous states compared to $\MPodot$, since positivity can be obtained through the right kind of cancellation between two summands each of which may not be PSD. 

In order to state the condition on a state to be \JPshort-ambiguous, we will first need the following kind of dephasing CP map. 

\begin{proposition} \label{prop:channel-D}
	Let $\eta \in \states{X}$ and $\eta = \sum_i p_i \dyad{i}{i}$ be an eigendecomposition of it with $r=rank(\eta)$ and wlog assuming that the support of $\eta$ is spanned by the first $r$ basis vectors. 
	Then ${\rleditb{\dephchannel}}_{\eta} \ : \ \operators{X} \ \rightarrow \ \operators{X}$ defined as follows is a CP map.
	\begin{align}
					{\rleditb{\dephchannel}}_{\eta} (A) := \sum_{k,l=1}^{r} \ \frac{2 \sqrt{p_k p_l}}{p_k + p_l} \  \dyad{k}{k} A \dyad{l}{l} \label{eq:def-D-rho}
	\end{align}
\end{proposition}
\begin{proof} 
	See \cite{song2025bipartite}. For an alternative derivation and explicit Kraus operators see this  \hyperref[proof:prop:channel-D]{proof in App.~B}. 
\end{proof}

In keeping with our notation we will often suppress the identity channel in $ \big( {\rleditb{\dephchannel}}_{\rho_X} \otimes \idchan_Y \big) (\rho)$ and just write ${\rleditb{\dephchannel}}_{\rho_X}(\rho)$.
For $\rho \in \states{X \otimes Y}$ we say it is \textit{\DXPPT} iff $\mathcal{T}_X \big( {\rleditb{\dephchannel}}_{\rho_X} (\rho) \big) \ \succeq \ 0$ for $\mathcal{T}_X$ with respect to some ONB. 
Note that the two maps commute and hence $\mathcal{T}_X \big( {\rleditb{\dephchannel}}_{\rho_X} (\rho) \big) \ \succeq \ 0$ iff ${\rleditb{\dephchannel}}_{\rho_X} \big( \mathcal{T}_X  (\rho) \big) \ \succeq \ 0$.\footnote{The commutation relation is straightforward to check using $\mathcal{T}_X(A) = \sum_{i,j} (\dyad{i}{j} \otimes \idop_Y) A (\dyad{i}{j} \otimes \idop_Y)$ and noting the symmetry of the fraction in Eq.~\eqref{eq:def-D-rho} under swapping $k$ and $l$. Also see \cite{song2025bipartite}.} 
We can now state the characterisation theorem.

\begin{theorem}[\cite{song2025bipartite}] \label{thm:JPA-equivalence-with-DPPT} 
	Let $\rho \in \states{X \otimes Y}$. Then $\rho \in \ASset{\JPshort}{X}{Y}$ iff it is \DXPPT. 
\end{theorem}
\begin{proof} 
	See \hyperref[proof:thm:JPA-equivalence-with-DPPT]{Proof in App.~B}. 
\end{proof}

Evidently, PPT implies \DXPPT. That the converse is not true is not entirely obvious, but indeed the case---see App.~\ref{app:example-DPPT-but-not-PPT} for an example to this effect. 

Lastly, we again note the properties analogous to those in Props.~\ref{prop:unique-CP-map-for-OP} and \ref{prop:basic-LSA-propeties}. 

\begin{proposition} \label{prop:JPA-channel-uniqueness}
	Let $\rho \in \ASset{\JPshort}{X}{Y}$.  Then 
	\begin{enumerate}[label=(\arabic*), leftmargin=2cm]
		\item There is a unique $\JPCPmap{\rho} \in \cpmaps{X}{Y}$ such that 
	$\forall \mathcal{E} \in \ARCset{\JPshort}{X}{Y}{\rho}$ it holds that 
	$\mathcal{E} \circ \Pi_X = \JPCPmap{\rho}$, or equivalently 
	$\pi_X D^J(\mathcal{E}) \pi_X \ = \ D^J(\JPCPmap{\rho})$, where $\pi_X$ is wrt $\rho_X$;  
	and $\{ \JPCPmap{\rho} + \mathcal{E}' \circ \Pi_X^{\perp} \mid \mathcal{E}' \in \channels{X}{Y} \} \subseteq \ARCset{\JPshort}{X}{Y}{\rho}$.
		\item $D^J(\JPCPmap{\rho}) \ = \ \rho_X^{-1/2} \ {\rleditb{\dephchannel}}_{\rho_X} (\rho) \  \rho_X^{-1/2}$.
		\item $D^J(\JPCPmap{\rho}) \succeq 0$.
	\end{enumerate}
\end{proposition}
\begin{proof} 
	See \hyperref[proof:prop:JPA-channel-uniqueness]{Proof in App.~B}. 
\end{proof}

\subsection{Summary} \label{subsec:ambiguity-summary}

For the respective classes of ambiguous states for the four \sot\ maps this work considers we established the following 

\vspace*{0.3cm}
\begin{center}
\begin{tabular}{ccccccc}
	$\ASset{\rMPshort}{X}{Y}$ &  $=$ & \ \ $\ASset{\MPshort}{X}{Y}$ \quad & $\subsetneq$ & \quad $\ASset{\LSshort}{X}{Y}$ \quad & $\subsetneq$ & \quad $\ASset{\JPshort}{X}{Y}$ \quad \\[0.2cm]
	  & & $\shortparallel$ & & $\shortparallel$ & & $\shortparallel$ \\[0.2cm]
	 & & \begin{minipage}{1.8cm} \centering \CXPPT \\[-0.1cm] states \end{minipage} & 
	 & \begin{minipage}{1.8cm} \centering PPT \\[-0.1cm] states \end{minipage} & 
	 & \begin{minipage}{1.8cm} \centering \DXPPT \\[-0.1cm] states \end{minipage}
\end{tabular}
\end{center}
\vspace*{0.3cm}
where both inclusions are strict. 
Note these three sets \CXPPT, PPT and \DXPPT\ are in fact of general interest in their own right in the context of entanglement theory. 

First, PPT states have been studied intensively over decades and it is notoriously hard to characterise them in general. Pioneered originally by Peres~\cite{Peres1996} and the Horodeckis~\cite{Horodecki1996}, their significance stems from the fact that being non-PPT is a sufficient condition for non-separability, i.e. being entangled. It is not a necessary condition though; some PPT states are non-separable, in which case they exhibit \emph{bound entanglement}. This in particular also means that it is impossible to produce Bell pairs from just copies of PPT states and only using local operations and classical communication (LOCC)~\cite{Horodecki1998}. 

Second, the class of \DXPPT\ states is strictly larger, i.e. it contains non-PPT states as we demonstrate through the example in App.~\ref{app:example-DPPT-but-not-PPT}. 
In particular, we show the example to be non-separable beyond bound entanglement---it features 	 \emph{distillable entanglement}, i.e. one \emph{can} use LOCC operations to produce Bell pairs from copies of such states. 
Note that in general it is not known whether non-PPT states with bound entanglement exist~\cite{Horodecki2022}, and the nature of entanglement of states in the set \DXPPT$\setminus$PPT in general is open, too. 

Third, turning to \CXPPT\ states, to the best of our knowledge, this class has not yet been characterised either. The condition on a state to commute with a marginal has appeared in the literature: it is known that states with zero \emph{quantum discord} commute with one of their marginals~\cite{Ferraro2010}, but also that the converse is not true. For example, the $2 \times 2$ Werner state of the form $\rho = \frac{\lambda}{2} (\ket{01} - \ket{10})(\bra{01} - \bra{10}) + \frac{1-\lambda}{4} \idop$ for $\lambda \in [0,1]$  commutes with its marginal seeing as the latter is maximally mixed, but $\rho$ has non-zero discord for $\lambda \in (0,1)$~\cite{Luo2008}. 
\edited{In fact it is not too hard to see that \CXPPT\ states exist that are also non-separable: given some non-separable PPT state $\tau \in \states{X \otimes Y}$ with full support\footnote{Such states can be constructed through a suitable admixture of a non-separable PPT state and the maximally mixed state.} and setting $\rho := (1/d_X) (\tau_X^{-1/2} \otimes \idop_Y) \tau (\tau_X^{-1/2} \otimes \idop_Y)$, clearly $\rho$ still is entangled and PPT, while also $\rho_X = (1/d_X) \idop_X$ and thus $[\rho, \rho_X]=0$.  
} 

Finally, given the non-negativity of the Jamio{\l}kowski states of the unique CP maps induced by our ambiguous states (see Props.~\ref{prop:unique-CP-map-for-OP}, \ref{prop:basic-LSA-propeties}, \ref{prop:JPA-channel-uniqueness}), note that this property always extends to some ambiguity-realising \emph{channels} in the following sense. 
Defining
\begin{align}
	\ARCsetpos{sot}{X}{Y}{\rho} & \ := \ \left\{ \mathcal{C} \in \ARCset{sot}{X}{Y}{\rho} \mid	D^J(\mathcal{C}) \succeq 0 \right\}
\end{align}
for $sot \in \{ \LSshort, \JPshort, \MPshort, \rMPshort \}$ it holds that $\ARCsetpos{sot}{X}{Y}{\rho} \neq \emptyset$ iff $\ARCset{sot}{X}{Y}{\rho} \neq \emptyset$, that is, whenever the state $\rho$ is ambiguous wrt $\odot^{sot}$ then there exists a realising channel with positive Jamio{\l}kowski-state.\footnote{This is immediate from that in all four cases it holds that (1) $D^J(\mathcal{E}^{sot}_{\rho}) \succeq 0$, (2) $\mathcal{E}^{sot}_{\rho} = \mathcal{E}^{sot}_{\rho} \circ \Pi_X$ and (3) for any $\mathcal{C} \in \channels{X}{Y}$ such that $D^J(\mathcal{C}) \succeq 0$ (see Prop.~\ref{prop:positive-DJ-equiv-to-PPT}) setting $\mathcal{E} := \mathcal{E}^{sot}_{\rho} + \mathcal{C} \circ \Pi_X^{\perp}$ one has indeed $D^J(\mathcal{E}) \succeq 0$ and $\mathcal{E} \in \ARCset{sot}{X}{Y}{\rho}$. (See Props.~\ref{prop:unique-CP-map-for-OP}, \ref{prop:basic-LSA-propeties} and \ref{prop:JPA-channel-uniqueness}, respectively, for all three kinds of facts.)}

This is useful for later, but also makes the following observation straightforward.  
For any unitary channel $\mathcal{U}(\argline) = U (\argline) U^\dag$ with unitary $U : \hilb{X} \rightarrow \hilb{Y}$ one finds $\mathcal{U}\not\in \ARCset{\sotmap}{X}{Y}{\rho}$ for any $\sotmap$ studied in this paper and any $\rho$ with full support, i.e. unitary channels cannot realise \MPshort, \LSshort\ or \JPshort\ ambiguity.\footnote{With $\rho_X$ having full support the claim is immediate noting Prop.~\ref{prop:unique-CP-map-for-OP}, \ref{prop:basic-LSA-propeties} and \ref{prop:JPA-channel-uniqueness}, since $D^J(\mathcal{U}) \not \succeq 0$. 
In general, $\mathcal{U}\not\in \ARCsetpos{\sotmap}{X}{Y}{\rho}$ for any state $\rho$.} 
Closely related then is the obvious but important example of a state that lies outside of the three discussed classes, namely that the maximally-entangled state $\rho = \frac{1}{d_X} \sum_{i,j} \dyad{ii}{jj}$ is not \DXPPT\ seeing as it is non-PPT and $\rho_X$ is maximally mixed, hence ${\rleditb{\dephchannel}}_{\rho_X} = \idchan_X$. It is thus not ambiguous in any of the senses studied here.  

\begin{figure}[H]
\centering
\includegraphics[width=1.0\textwidth]{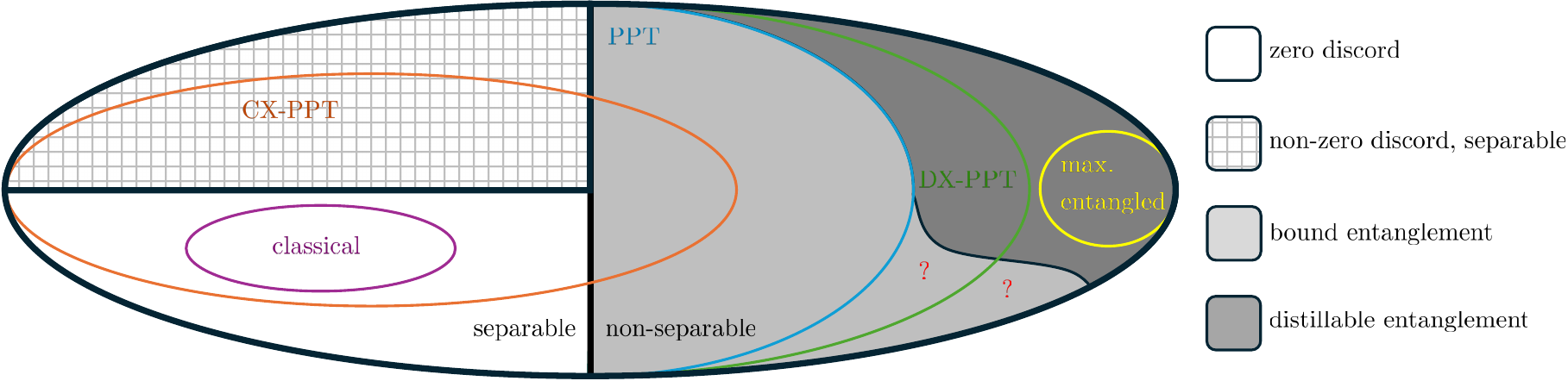}
\caption{Schematic Venn diagram of how the sets \CXPPT, PPT and \DXPPT\ sit within the set of all states $\states{X \otimes Y}$; the question marks indicate which intersections are not known to definitely be (non-)empty; `classical' denotes the subset of separable states that can be written as a convex combination of pairwise orthogonal pure product states.}  
\label{fig:inductive_bias}
\end{figure}

Figure~\ref{fig:inductive_bias} visualises the hierarchy of classes of states described above, including the open questions concerning which intersections are (non-)empty. 
It seems interesting that \LSshort-ambiguous states lie---in terms of the cardinality of the sets and the degree of entanglement---in between those based on \MPshort\ and \JPshort, both of which do not yield non-trivial models. 

With a view to quantum representation learning, it is natural to see the size of sets of ambiguous states and the correlations between $X$ and $Y$ they can describe, in particular the degree of entanglement they cover, as kinds of inductive biases one may choose accordingly, given a problem that quantum representation learning is to be applied to. 
Of course, in practice this would require a parameterisation of states for that respective class, while the parameterisation of the aforementioned classes is non-trivial---we leave this as an interesting direction for future work. 
It will also be interesting to establish in the future what the classes of ambiguous states are with respect to potential other kinds of \sot\ maps not considered here.

\section{The LS-construction: models with inference and generation} \label{sec:main-section-on-LS-models}

The previous section studied classes of ambiguous states. 
One insight was that---unlike what one might have suspected---the property of a quantum state to be ambiguous, that is, to factorise in a way that generalises how any bipartite classical distribution factorises according to the chain rule, does \emph{not} imply that the correlation between $X$ and $Y$ would have to be classical; ambiguous states are not generally separable. 

Now, the no-go fact in Thm.~\ref{thm:no-go-linear} rules out non-trivial models using \MPshort, \rMPshort\ and \JPshort, however for all we know non-linearity of the extended inference (generation) maps is a necessary but not a sufficient condition---it does not mean that any \sot\ map with non-linear extended inference actually yields models. 
The \LSshort-based construction meets the non-linearity requirement; 
in this section then we gather all the pieces, establish the existence of LS-based models and list all their properties.

\begin{lemma} \label{lem:LS-positive-extended-map}	
	Let $\rho \in \ASset{\LSshort}{X}{Y}$ 
	then for any ambiguity-realising channel $\mathcal{C} \in \ARCsetpos{\LSshort}{X}{Y}{\rho}$, the induced map $\mathcal{C} \LSodot (\argline) : \operators{X} \rightarrow \operators{X \otimes Y}$ is positive; moreover, $\forall \eta \in \states{X}$ the state $\mathcal{C} \LSodot \eta$ is again PPT.\footnote{In light of $\mathcal{C} \LSodot \eta$ being PPT again, as an aside question one may wonder how many distinct PPT states $\rho$ are needed such that the union of $\mathcal{C} \LSodot (\states{X})$ (varying over $\rho$ and $\mathcal{C}\in \ARCset{\LSshort}{X}{Y}{\rho}$) covers all PPT states.} 
\end{lemma}
\begin{proof}
	See \hyperref[proof:lem:LS-positive-extended-map]{Proof in App.~B}. 
\end{proof}

Note Secs.~\ref{sec:sot-maps}-\ref{sec:ambiguous-states} and above Lem.~\ref{lem:LS-positive-extended-map} are agnostic as to which roles $X$ and $Y$ play. 
So, due to the symmetry of the PPT condition, we can hence finally conclude the following main result. 

\begin{theorem} \label{thm:main-LS-model-theorem}
	For a model class $\{\rho_{\theta}\}_{\theta \in I}  \subseteq \states{X \otimes Y}$ with $I \subseteq \mathbb{R}^n$ the following are equivalent: 
	\begin{enumerate}[label=(\arabic*), leftmargin=2cm]
		\item It is a \emph{model class with inference and generation} with respect to $\LSodot$, i.e. $\forall \theta \in I$ there is an \infchan\ $\infchansymbol_{\theta} \in \channels{X}{Y}$ and \genchan\ $\genchansymbol_{\theta} \in \channels{Y}{X}$ such that 
		\begin{align}
			\infchansymbol_{\theta} \LSodot \rho_{\theta,X} & \quad = \ \rho_{\theta} \ = \quad \genchansymbol_{\theta} \LSodot \rho_{\theta,Y}  \label{eq:LS-theorem-equation}	
		\end{align}
		\item $\rho_{\theta}$ is PPT $\forall \theta \in I$. 
	\end{enumerate}	 
\end{theorem}
\begin{proof}
	Theorem~\ref{thm:LSA-PPT-theorem} established that $\rho_{\theta} \in \ASset{\LSshort}{X}{Y}$ iff it is PPT. 
	By the symmetry of the PPT condition it follows that this is also equivalent to $\rho_{\theta} \in \ASset{\LSshort}{Y}{X}$. 
	Hence any class of PPT states is a class of states that can be factorised with respect to $\LSodot$ in both directions, $(X,Y)$  and $(Y,X)$, that is, there exist channels $\infchansymbol_{\theta} \in \channels{X}{Y}$, $\genchansymbol_{\theta} \in \channels{Y}{X}$ such that Eq.~\eqref{eq:LS-theorem-equation} holds. 
	Finally, it follows from Lem.~\ref{lem:LS-positive-extended-map} that $\infchansymbol_{\theta} \LSodot (\argline)$ and $\genchansymbol_{\theta} \LSodot (\argline)$ are both positive maps and from (3) of Prop.~\ref{prop:basic-properties-star-products} that they are trace-preserving. 
	Thus both map states to states---they indeed define extended inference and extended generation maps, respectively. 
\end{proof}

With the \LSshort-based \sot\ map we thus have identified model classes with inference and generation in the sense of Sec.~\ref{sec:the-framework}. 
The framework we devised from first principles is not vacuous---a priori it was not obvious what the stipulations jointly imply and whether ideal models in its sense exist. 
Naturally, the next steps are the further study of \LSshort-based models and the exploration of the entire landscape of \sot\ maps that can yield non-trivial models. 
See Sec.~\ref{sec:discussion} for further details.

\section{Data extension without inference and generation} \label{sec:data-extension-not-from-inf}

Classically, there is no reason to---and no one ever would---distinguish the concepts of, and methods for, data extension and extended inference. 
However, as Sec.~\ref{subsec:def-data-extension} argued, it is instructive to a priori separate them for a quantum framework. 
Of course, when training a model with inference and generation, ideally one would prefer quantum data extension to be given by some \infmap, be it the model's current one (i.e. an exact inference setup) or some other \infmap\ (hence an instance of approximate inference), for it is a stronger notion, satisfying desirable properties and also a closer analogue to the classical case. 

Still, for what we generally referred to as an approximate inference setup it is in principle conceivable that it uses a data extension map that is not given by any \infmap.  
See the overview of setups in Sec.~\ref{subsec:general-training-setups} and also recall that the conditions on a data extension map only reflect what training via the minimisation of some divergence such as the quantum relative entropy requires, namely states in $\states{X \otimes Y}$ with the right marginals on $X$ (recovering the data state); 
it does not ask for marginals on $Y$ that encode the output of inference or to yield the right `classical limit' (cf. Def.~\ref{def:model-with-inf-or-gen}). 
Under such weaker constraints more possibilities may open up.

Yet, even in considering these different approaches to training so far, it always concerned some model with inference and generation---as is the main focus of this work. 
As a final direction let us now consider a different and weaker goal, where one is not, a priori, interested in inference and generation as such at all, but only in successfully learning some quantum model (on $\hilb{X} \otimes \hilb{Y}$) of data (on $\hilb{X}$); 
\edited{
see, for instance, \cite{kieferova2017tomography, torlai2018latent, song2019geometry, Huijgen_2024, ezzell2023quantum, Guo2024QuantumStateTOmography, Wilde2025} for works that fall into this category.}

It is reasonable to assume that the purpose of such a model $\rho_{\theta} \in \states{X \otimes Y}$ then ultimately is to compute expectation values for certain observables. 
NB using a model to do inference or generation, on the other hand, means a stronger thing---deterministically outputting a quantum state again.   
In light of this premise, and further motivated by the concrete example below, it is instructive to weaken the notions of data extension and being a model-of-data as follows. 

\begin{definition}[Generalisation of data extension and model-of-data condition] \label{def:generalised-data-extension}
	A map $\Phi \ : \ \operators{X} \ \rightarrow \ \operators{X \otimes Y}$ 
	is a \emph{\gendataext} iff $\Phi \ : \ \states{X} \hookrightarrow \states{X \otimes Y}$ and $\forall O \in \Herm{X}$ there exists a `correction' function $\corr_{O} : \mathbb{R} \rightarrow \mathbb{R}$ such that 
	\begin{align}
		\forall \eta \in \states{X} \quad : \quad  \tr[O \eta ] \ = \ \corr_{O} \big(\tr\big[ O \trace_{Y}[\Phi(\eta)]	 \big] \big) 
	\end{align}	

	A model $\rho_{\Theta} \in \states{X \otimes Y}$ 
	is a \emph{model of data} $\sigma$ iff 
	there exists a \gendataext\ $\Phi$ such that  $\rho_{\Theta} = \Phi(\sigma)$.
\end{definition}

Clearly, this weakening is innocent from a practical point of view---of course, assuming that the functions $\corr_O$ are efficiently computable. That a model's marginal $\rho_{\theta,X}$ is not equal to some data $\sigma$, be it on the nose or approximately, but deviates in a---maybe even significant---way, yet one that can be compensated for perfectly through some (classical) postprocessing of measurement statistics would not change how practically suitable the notion of model and data extension are.\footnote{Note that the two parts of Def.~\ref{def:generalised-data-extension} are of course consistent with each other and the general idea of training a model via the minimisation of some divergence that comes with a DPI---even though $\Phi(\sigma)$ and $\sigma$ are then generally not related to each other via a channel that also relates $\rho_{\Theta}$ and its marginal. If $\rho_{\Theta}$ is close to some state $\Phi(\sigma)$ then so are their marginals and therefore $\rho_{\Theta,X}$ is precisely such that $\rho_{\Theta}$ is close to being a model of $\sigma$ in the generalised sense.}

Let us then give an example of the kind of map that instantiates the above slight generalisation. 
First define the maps $\mathcal{B}^+$ and $\mathcal{B}^-$ as follows
\begin{align}
	\mathcal{B}^\pm \ & : \ \operators{X} \ \rightarrow \ \operators{X \otimes X} \\
					& \quad \ A \ \ \mapsto \ \ \frac{2}{d_X \pm 1} \left(\frac{ \idop_{XX} \pm S}{2}\right) (A \otimes \idop_X) \left(\frac{\idop_{XX} \pm S}{2}\right) \label{eq:def-channels-B-pm}
\end{align}
where $S \in \operators{X \otimes X}$ is the SWAP operator.\footnote{Defined by $S \ket{\phi} \ket{\psi} = \ket{\psi} \ket{\phi} \ \forall \ket{\psi}, \ket{\phi} \in \hilb{X}$.} 
Note that both maps $\mathcal{B}^{\pm}$ are channels~\cite{Parzygnat_2023}. 
For distinction let $X'$ refer to the second factor of the joint space, i.e. write $\hilb{X \otimes X'} $ where $\hilb{X'}=\hilb{X}$. 
It is straightforward to calculate the marginal on $X$ to be given by (for $A \in \operators{X}$)
\begin{align}
	\tr_{X'} \big[ \mathcal{B}^\pm (A) \big] \ = \ \frac{d_X \pm 2}{2(d_X \pm 1)} A \ + \ \frac{\tr[A]}{2(d_X \pm 1)} \idop_{X} \ . \label{eq:marginal-B-pm}
\end{align}
Hence, for $O \in \Herm{X}$ and a state $\eta \in \states{X}$ (hence $\tr[\eta]=1$):
\begin{align}
	\tr [O \eta] \ = \ \corr_O^\pm \Big(\tr \big[ O \tr_{X'} [ \mathcal{B}^\pm (\eta) ] \big] \Big) \ , \label{eq:marginal-statistics-B-pm}
\end{align}
where the respective functions $\corr_O^\pm : \mathbb{R} \rightarrow \mathbb{R}$ (assuming $d_X > 2$ in case of $\mathcal{B}^-$) are defined as
\begin{align}
	\corr_O^\pm(a) \ := \ \frac{2(d_X \pm 1)}{d_X \pm 2} \ a \ - \  \frac{1}{d_X \pm 2} \ \tr[O] \ . \label{eq:marginal-statistics-B-pm-corr}
\end{align}
Now letting $\mathcal{E}_{\phi} \in \channels{X'}{Y}$ be any family of channels parametrised by $\phi \in I \subseteq \mathbb{R}^m$, then 
\begin{align}
	\Phi_{\phi} \ := \  (\idchan_X \otimes \mathcal{E}_{\phi}) \circ \mathcal{B}^+ \ : \ \operators{X} \ \rightarrow \ \operators{X \otimes Y} 	\label{eq:B-plus-DEM}
\end{align}
defines a family of generalised data extension maps in the sense of Def.~\ref{def:generalised-data-extension} (similarly if using $\mathcal{B}^-$ instead). 
This is straightforward to see since $\tr_Y \circ \Phi_{\phi} = \tr_{X'} \circ \mathcal{B}^+$ and given Eqs.~\eqref{eq:marginal-statistics-B-pm}, \eqref{eq:marginal-statistics-B-pm-corr}. 

It is worth noting that, beyond simply yielding examples of generalised data extension maps, the channels $\mathcal{B}^\pm$ are particularly interesting for they can be used, through the right kind of linear combination of both, to implement the virtual broadcasting map~\cite{Parzygnat_2023}; and the latter in turn is closely related to `\JPshort-based \infmap{s}' (see App.~\ref{app:algorithm} for the details).   

We note that the map in Eq.~\eqref{eq:B-plus-DEM} is indeed an example of `data extension \emph{not} from inference' seeing as the map
\begin{align}
	 \cpmaps{X}{Y} \times \operators{X} & \ \rightarrow \ \operators{X \otimes Y} \\ 
	 (\mathcal{E}, A) \quad \quad  & \ \mapsto \  \ \big( (\idchan_X \otimes \mathcal{E}) \circ \mathcal{B}^+ \big) (A)
\end{align}
is not an \sot\ map. 
Even if one considered a weakening of \sot\ maps in keeping with Def.~\ref{def:generalised-data-extension} to allow for the fact that $\tr_Y\big[ \big( (\idchan_X \otimes \mathcal{E}) \circ \mathcal{B}^+ \big) (A) \big] \neq A$, the fact that $\tr_X\big[ \big( (\idchan_X \otimes \mathcal{E}) \circ \mathcal{B}^+ \big) (A) \big] = \mathcal{E} \big( \tr_X[\mathcal{B}^+ (A) ]\big) \neq \mathcal{E}(A)$\footnote{Noting Eq.~\eqref{eq:marginal-B-pm}, which for $\mathcal{B}^+$ is symmetric in which of the two factors one traces.} means there is no reason to consider $\mathcal{E}$ outputting any representation of $A$. 
This kind of argument also makes manifest that the above $\Phi_{\phi}$ is incompatible with using it for approximate inference training. In fact, this observation is not specific to the above example, but a general consequence of Def.~\ref{def:generalised-data-extension}---generalised data extension is not compatible with training some \sot-based model with inference or generation. 

The further exploration in the future of such weaker approaches based on Def.~\ref{def:generalised-data-extension} will likely depend on any success of and the costs of practically training models \emph{with} inference and generation in the sense this work focused on otherwise. 

\section{Discussion} \label{sec:discussion}
	
We have introduced a framework for unsupervised representation learning from coherently available quantum data in which a model is a joint state over visible and latent systems, while inference and generation are defined through channel-state factorisations induced by a state-over-time map. 
We pointed out how, unlike classically, extended inference---that is, a map induced by an inference channel but that goes from the visible to the joint system---is distinguished from the logically weaker notion of data extension;  
and we sketched how the notions we stipulated indeed constitute the building blocks for various training setups for a quantum generalisation of unsupervised representation learning. 
On the basis of this framework we then showed two main results. 
First, an \emph{extended} inference or generation map that is linear, positive and exactly preserves its input marginal must be trivial (Thm.~\ref{thm:no-go-linear}). Thus any non-trivial model satisfying our consistency requirements necessarily involves a non-linear extension map to joint states, even though the operational inference and generation maps themselves remain ordinary quantum channels. 
Second, for the LS construction, a model class has inference and generation if and only if every model state is PPT (Thm.~\ref{thm:main-LS-model-theorem}).

It is worth emphasising that the latter result identifies PPT not as an externally imposed ansatz, but as the precise boundary condition within all joint states allowing for valid bidirectional LS factorisations.  
This boundary is restrictive but genuinely quantum: it includes separable states as well as PPT-entangled states, while excluding any non-PPT states such as maximally entangled ones. 
The hierarchy of ambiguous-state classes found for the operator-product, LS and Jordan-product constructions may be thought of as potential choices of inductive bias, however, our results showed that factorising a fixed joint state is a necessary but not sufficient condition for obtaining a globally valid learning model. 
In particular, the Jordan-product construction admits a larger class of ambiguous states than LS, but its linearity prevents it from yielding non-trivial models under our positivity requirement. 

Our structural results by themselves do not imply efficient trainability---whether the PPT restriction provides a useful inductive bias for learning, or instead limits expressivity too strongly, is an empirical question we leave open. 
Also note there is no claim of quantum advantage in the usual or any rigorous sense of the word; however the problem may be seen, by definition, as intrinsically quantum and could not be done classically seeing as part of the requirements is to take in and process a single quantum state and output a quantum state, be it in the inferential or the generative direction. If any such map factored through a classical system it would clash with the `single-shot requirement'---a key difference between training on data from a quantum data source and using the model for inference or generation.

The work also provides an application-driven perspective on the wider theory of states over time. Axiomatic approaches have studied SOT constructions using requirements such as (bi-)linearity, covariance and conditionability, leading to recent uniqueness results in both bipartite and multipartite settings \cite{Lie_2024, lie2025multipartite}. 
Our setting imposes a different additional criterion: for learning, the channel-state construction should define a genuine joint state throughout the domain on which data may be processed. Our no-go theorem makes explicit the resulting tension between linearity, positivity and non-triviality, while the LS construction shows how non-linearity is indeed a viable way out of the tension. These observations complement, rather than contradict, existing uniqueness results, whose assumptions and intended operational interpretations differ from ours.

\edited{
The framework has a more qualified relevance to supervised representation learning. A purely discriminative quantum model needs only an encoder followed by a readout and therefore does not require a compatible joint state as a model with also generation. 
The structure we assume would become relevant though when supervised learning is combined with unsupervised pretraining, when the representation is required to remain compatible with a coherent generative model, or when labels and prediction targets are themselves quantum systems. 
A natural extension might append a task register and optimise the preservation of task-relevant information in the latent state, thereby connecting the framework to quantum information-bottleneck approaches 
\cite{banchi2021generalization, hayashi2023efficient, Catli_2024}.} 
For completeness we also mention that the notion of inference treated in this work should not be equated with quantum causal inference. 
The LS factorisation is causally neutral and bidirectional, and no interventions or causal Markov assumptions are introduced here; 
though at a technical level as well as in historical respects there are of course links to quantum causal model frameworks \cite{allen2017quantum, barrett2019quantum} and to works on quantum causal inference \cite{song2025bipartite}.

This work naturally leads to several directions for future work. 
A first priority is to study the general \emph{landscape of \sot\ maps} and in particular characterise those maps  that are non-linear in their second argument and induce some \emph{non-separable} ambiguous states $\rho \in \states{X \otimes Y}$, which moreover are such that $\rho$ is a model with inference and generation; in short characterise the \sot\ maps that support non-trivial models with inference and generation.  
Within that set the goal then is to identify the potentially interesting candidates to compare against \LSshort\, and to establish which further requirements might single out an optimal choice as unique---be it LS or another one. 
A further aspect is to generally explore a `channels-first' approach, that is a complementary characterisation of ambiguity: instead of characterising ambiguous states, determine the class of `admissible ambiguity realising channels', that is given some $\odot$ what are the channels $\mathcal{E} \in \channels{X}{Y}$ such that $\mathcal{E} \odot \eta$ is a state for all states $\eta \in \states{X}$?  
This also is a natural approach if wanting to define a parametrised model class in terms of decoder maps and prior states (on the latent system) in the first place, as is common classically. 

A second priority is an approximate and robust version of the framework. 
Realistic models will only approximately reproduce their visible marginals and may only approximately factorise with respect to a chosen \sot\ map. 
It will therefore be important to define quantitative notions of proximity on the basis of which all central concepts in the framework can be generalised suitably and in a consistent manner. 
Also generalisation beyond training data could be worth investigating. 
At least for the \sot\ maps this work considered ambiguity-realising channels are unique only on the support of a given (data) state $\sigma$. 
A question hence is whether there are canonical and stable off-support extensions of respective unique CP maps to channels. 

Third, on the algorithmic side and with a view to LS-based models, it will be useful to identify tractable parametrised subclasses of PPT states that are expressive enough to contain non-separable states (e.g. using constructions based \edited{on unextendible product bases, low-rank extremal PPT families, or strong-PPT states defined through Cholesky block factorisations \cite{BennetEtAL_1999_Unextendible, LeinaasEtAL_2010_LowRankPPT, Bylicka_2013}}). For each family one would then need efficient parametrisations of the respective ambiguity realising channels, i.e. the inference and generation channels, as well as estimable training objectives and gradients (see App.~\ref{app:gradient_qre}), together with bounds on sample and circuit complexity. 

A fourth direction concerns practical and more heuristically guided work. 
All the mentioned future developments will hopefully enable proof-of-principle \emph{experiments}---in simulation and on actual quantum hardware---to demonstrate when, and how efficiently, models can be trained successfully, as well as eventually used for downstream tasks.

Finally, one might consider a more radical extension of the current framework, namely by relaxing the requirement that a quantum model be given by a joint quantum state, i.e. a positive operator which induces well-defined statistics (for observables). 
This possibility has close precedents in the pseudo-density-matrix and virtual-broadcasting formalisms, where nonpositive Hermitian objects can nevertheless encode operationally accessible temporal or quasiprobabilistic statistics \cite{Fitzsimons2015, fullwood2025quantum, Parzygnat_2024}. 
Holding onto the \sot\ construction, but allowing non-PSD operators as output of extended inference maps would open many more possibilities---in particular avoid our no-go result and bring back linearity---however require a way to make sense of `non-physical' operators as quantum models in a learning framework.  
Note that by definition of \sot\ maps, when \emph{using} such supposed models for inference or generation one would again be dealing with ordinary channels and states; it remains open though whether one can define any suitable loss function and in principle implementable training procedures for such models.


\phantomsection
\addcontentsline{toc}{section}{References}
\bibliography{bibliography}
\bibliographystyle{alpha}

\appendix

\section{Complementary material}

\subsection{Training and the gradients of quantum relative entropy} \label{app:gradient_qre}
While this work does not provide the technical details to implement the training of quantum models, this section will sketch some basic aspects thereof. 
As noted in Sec.~\ref{subsec:general-training-setups}, given quantum data $\sigma$ training a model class $\{\rho_{\theta}\}_{\theta}$ in the framework means the minimisation of  
\begin{align}	
	D\big( \Phi_{\phi}( \sigma) \ \| \ \rho_{\theta} \big) \label{eq:basic-objective-function_app}
\end{align}
wrt parameters $\theta$ and $\phi$ for some family of data extension maps $\Phi_{\phi}$ and some objective function $D$, which for the sake of concreteness we will here henceforth assume is given by the quantum relative entropy.\footnote{As always, assuming the corresponding inclusion relation on the support of the operators so as to guarantee a well-defined loss; in fact we will here assume that that support is independent from $\theta$ to avoid complication below when considering differentiation.} 
For practical reasons one would normally minimise wrt either $\phi$ or $\theta$ at a time, not both at the same time.  

Let us start by considering the gradient wrt $\theta_k$, i.e. effectively the second argument of $D(\argline \| \argline)$.\footnote{In an exact inference setup, with the same parameters appearing in both arguments, in order to obtain more tractable expressions one may for instance hold one argument fixed and proceed in an iterative manner akin to a classical Expectation-Maximization algorithm (see Sec.~\ref{subsec:intro-classical-case}). \label{footnote-label}}  
Recall
\begin{align}
	D\big( \Phi_{\phi}( \sigma) \ \| \ \rho_{\theta} \big) = \tr \big[ \Phi_{\phi}( \sigma) (\log \Phi_{\phi}( \sigma) - \log\rho_{\theta}) \big] 	\label{eq:objective-as-rel-ent}
\end{align}	
and note that from~\cite[Lemma 3.4]{Sutter2017} and~\cite[Remark 24]{Wilde2025} it follows that 
\begin{align}
	\frac{\partial}{\partial \theta_k} \log \rho_{\theta} = \int_{-\infty}^{\infty} \beta_0(t) \rho_{\theta}^{-\frac{1+it}{2}} \frac{\partial}{\partial \theta_k}(\rho_{\theta}) \rho_{\theta}^{-\frac{1-it}{2}} dt \label{eq:log-der}
\end{align} 
where $\beta_0(t) := \frac{\pi}{4} \sech^2\left(\frac{t\pi}{2}\right)$ is a probability density. Thus one finds  
\begin{align}
	\Big( \grad_\theta  D\big( \Phi_{\phi}( \sigma) \ \| \ \rho_{\theta} \big) \Big)_k \ & = \ 
		- \tr \big[ \Phi_{\phi}( \sigma) \ \frac{\partial}{\partial \theta_k} \log \rho_{\theta} \big] \\ 
		& = \  - \tr \big[ \Phi_{\phi}( \sigma) \   
		\int_{-\infty}^{\infty} \beta_0(t) \rho_{\theta}^{-\frac{1+it}{2}} \frac{\partial}{\partial \theta_k}(\rho_{\theta}) \rho_{\theta}^{-\frac{1-it}{2}} dt
		   \big] \\
		& = \  - \int_{-\infty}^{\infty} \beta_0(t) \ \tr \Big[ \Phi_{\phi}( \sigma) \   
		 \rho_{\theta}^{-\frac{1+it}{2}} \frac{\partial}{\partial \theta_k}(\rho_{\theta}) \rho_{\theta}^{-\frac{1-it}{2}} \Big] dt \label{eq:app-grad-expression_before_cycling} \\
		 & = \  -  \tr \Big[  \Big( \int_{-\infty}^{\infty} \beta_0(t) \rho_{\theta}^{-\frac{1-it}{2}} \Phi_{\phi}( \sigma) \  \rho_{\theta}^{-\frac{1+it}{2}} dt \Big) \ \frac{\partial}{\partial \theta_k}(\rho_{\theta}) \Big]   \label{eq:app-grad-expression}
\end{align}
where we used the cyclicity and the linearity of the trace. 

How might one actually estimate such gradients? 
While the integral can be seen to define a linear (in fact CP) map in $\Phi_{\phi}( \sigma)$, there are various sources of non-linearity. 
First, we saw that $\Phi_{\phi}( \argline)$ should be expected to not be CPTP (Prop.~\ref{prop:no-braodcasting-kind-of-fact}), in fact non-linear in general (Lem.~\ref{lem:no-braodcasting-kind-of-fact-for-any-positive-map}). 
Second, just like classically one might want to suppose a model class $\{\rho_{\theta}\}_{\theta}$ that comes with generation wrt some \sot\ map, i.e. there are generation channels $\genchansymbol_{\theta}$ such that $\rho_{\theta} = \genchansymbol_{\theta} \odot \rho_{\theta,Y}$ and the model class is in fact given by a parametrisation of the prior latent states $\rho_{\theta,Y}$ and the decoder channels $\genchansymbol_{\theta}$ (also see Sec.~\ref{subsec:general-training-setups}). 
For the same reason as before (Lem.~\ref{lem:no-braodcasting-kind-of-fact-for-any-positive-map}) the map $\genchansymbol_{\theta} \odot (\argline)$ should be expected to be non-linear, yet $\frac{\partial}{\partial \theta_k}(\genchansymbol_{\theta} \odot \rho_{\theta,Y})$ would be required. 
Third, if considering the gradient of Eq.~\eqref{eq:objective-as-rel-ent} with respect to $\phi$, or in any case the first argument (see comment in footnote~\ref{footnote-label}), then this evidently leads to expressions manifestly non-linear even in $\Phi_{\phi}( \sigma)$ as such. 
The bottom line is that dealing with gradients that are non-linear in the given states is inevitable. 

In order to demonstrate why this is not a fundamental issue in principle, we just offer the following generic observation. 
The desired gradients are non-linear functionals in states. 
Although it is not a necessity for the estimation techniques mentioned below, suppose the functional of interest can be written in the form 
\begin{align}
	(\grad_\theta  D)_k = \tr [O_k f(\sigma)] 
	\label{eq:grad-as-exp-value}
\end{align}
for some observables $O_k \in \Herm{X \otimes Y}$ and some generally non-linear function $f : \states{X} \rightarrow \states{X \otimes Y}$.\footnote{If one allowed just any $f : \operators{X} \rightarrow \operators{X\otimes Y}$ then a rewriting as in Eq.~\eqref{eq:grad-as-exp-value} is of course always possible.} 
As is straightforward to see, this is the case for the above: setting  $O_{k,t} := \rho_{\theta}^{-\frac{1+it}{2}} \frac{\partial}{\partial \theta_k}(\rho_{\theta}) \rho_{\theta}^{-\frac{1-it}{2}}$, which is of the form $A B A^{\dagger}$ with $B \in \Herm{X \otimes Y}$ and thus Hermitian, one may rewrite the RHS of Eq.~\eqref{eq:app-grad-expression_before_cycling} as $\int_{-\infty}^{\infty}  \beta_0(t) \tr [O_{k,t} \Phi_{\phi}( \sigma)] dt$, so that  
a Monte Carlo approximation of the integral then has the gradient essentially become an expectation value of the claimed form. 

Now, there are various well-known results and standard tricks to estimate non-linear functionals of states. 
\edited{They typically involve acting jointly on several independent copies of the state. 
In particular, every degree-n polynomial functional of a state $\sigma$ can be represented as the expectation value of an observable on $\sigma^{\otimes{n}}$; 
the controlled-SWAP and cyclic-permutation tests being basic examples of this multi-copy strategy \cite{ekert2002direct, Brun2004}. 
Thus, for polynomial functionals, or for nonpolynomial functionals after a suitable polynomial approximation, one may seek a decomposition of the form
 \begin{align}
	\tr \big[O f(\sigma)\big] \ \approx \ \sum_{l=1}^L c_l \tr \big[O_l \ \mathcal{C}_l(\sigma^{\otimes{n_l}}) \big] 
	\label{eq:generic-decomposition-exp-values}
\end{align}
where $L \in \mathbb{N}$ and for $\forall l=1,...,L$: $c_l \in \mathbb{R}$, \rleditb{$O_l \in \Herm{Z}$ and $\mathcal{C}_l \in \channels{X^{\otimes{n_l}}}{Z}$ for some ancilla system $Z$}. 
} 

Finally, suppose a dataset $(\sampleindex{\sigma}{j})_{j=1}^N$ is given, drawn i.i.d. from some quantum data source $\{ (p_i, \decompindex{\sigma}{i}) \}_{i=1}^k$. 
If partitioning the index set $\{1,..,N\}$ into blocks $B_{l,r}$ each of size $n_l$  (with $r=1,...,R$ for some suitable $R \in \mathbb{N}$ such that $\sum_{l=1}^L Rn_l = N$) then for large enough $N$ one may approximate the RHS of Eq.~\eqref{eq:generic-decomposition-exp-values} by
\begin{align}
	\sum_{l=1}^L c_l \sum_{r=1}^R \frac{1}{R} \tr \Big[O_l \ \mathcal{C}_l \big(\bigotimes_{j \in B_{l,r}} \sampleindex{\sigma}{j} \big) \Big] 
\end{align}

The assessment of whether a practical implementation of such a procedure is tractable has to be with respect to a specific choice of a parametrised model class (be it \LSshort, or some other suitable \sot\ map; see Secs.~\ref{sec:main-section-on-LS-models} and \ref{sec:discussion}); we leave this to future work.

\subsection{Example of a non-PPT state that is \DXPPT} \label{app:example-DPPT-but-not-PPT} 

The below concrete example establishes that bipartite states exist that are \DXPPT, but not PPT. 
Let $d_X = 3, d_Y = 2$ and define $\rho := \rho_\mathrm{real} + i \rho_\mathrm{imag}$ where 
\begin{equation}
\label{eq:ppt_dppt_example_32}
\begin{split}
\rho_\mathrm{real} &=
{\small
\left(
\begin{array}{cccccc}
0.341392 & -0.018181 & -0.115488 & -0.067222 & -0.002548 & -0.027566 \\
-0.018181 & 0.443973 & 0.012921 & -0.025352 & 0.005881 & 0.015663 \\
-0.115488 & 0.012921 & 0.083515 & -0.001546 & -0.003357 & 0.024048 \\
-0.067222 & -0.025352 & -0.001546 & 0.098782 & 0.000088 & -0.009840 \\
-0.002548 & 0.005881 & -0.003357 & 0.000088 & 0.019474 & 0.003833 \\
-0.027566 & 0.015663 & 0.024048 & -0.009840 & 0.003833 & 0.012863 \\
\end{array}
\right)
} \\
\rho_\mathrm{imag} &=
{\small
\left(
\begin{array}{cccccc}
0.000000 & -0.046591 & 0.037808 & -0.010960 & -0.059082 & -0.009599 \\
0.046591 & 0.000000 & -0.026300 & 0.010275 & 0.003416 & 0.001479 \\
-0.037808 & 0.026300 & 0.000000 & 0.013893 & 0.029772 & 0.015276 \\
0.010960 & -0.010275 & -0.013893 & 0.000000 & 0.026384 & -0.011634 \\
0.059082 & -0.003416 & -0.029772 & -0.026384 & 0.000000 & -0.006456 \\
0.009599 & -0.001479 & -0.015276 & 0.011634 & 0.006456 & 0.000000 \\
\end{array}
\right)
}
\end{split}
\end{equation}

\rledit{Note $\rho_\mathrm{imag}^T = -\rho_\mathrm{imag}$ and hence $\rho$ is Hermitian; also $\tr[\rho]=1$. 
Now as one may straightforwardly check the eigenvalues of $\rho$, $\mathcal{T}_X(\rho)$ and $\mathcal{T}_X({\rleditb{\dephchannel}}_{\rho_X}(\rho))$, respectively, are as follows
\begin{align}
\mathrm{EV(\rho)} = 
{\small
\left(
\begin{array}{c}
4.2202828 \cdot 10^{-9} \\
1.3058831 \cdot 10^{-3} \\
3.3435285 \cdot 10^{-2} \\
9.7071648 \cdot 10^{-2} \\
3.8111183 \cdot 10^{-1} \\
4.8707500 \cdot 10^{-1} \\
\end{array}
\right)
}
\,, \qquad
\mathrm{EV(\mathcal{T}_X({\rleditb{\dephchannel}}_{\rho_X}(\rho)))} = 
{\small
\left(
\begin{array}{c}
0.00121465 \\
0.00873855 \\
0.04097736 \\
0.10474302 \\
0.37509558 \\
0.46923079 \\
\end{array}
\right)
} \,,
\label{eq:eigenvalues_line_1}
\end{align}
\begin{align}
\mathrm{EV(\mathcal{T}_X(\rho))} = 
{\small
\left(
\begin{array}{c}
-0.01634675 \\
0.00291748 \\
0.04072443 \\
0.10628327 \\
0.386001 \\
0.48042065 \\
\end{array}
\right)
} \ .
\label{eq:eigenvalues_line_2}
\end{align}
While the positivity of all eigenvalues in Eq.~\eqref{eq:eigenvalues_line_1} establishes that $\rho \succeq 0$ and hence that $\rho$ is a state, as well as that it is \DXPPT, the negative eigenvalue $-0.01634675$ in Eq.~\eqref{eq:eigenvalues_line_2} establishes that $\rho$ is indeed not PPT.}

Note that it is one of the open questions in quantum information theory whether there exist non-PPT states that have bound, rather than distillable, entanglement. It is known however that in case $d_X=2$ and $d_Y \geq 2$ (or the other way round) non-PPTness is both necessary and sufficient for featuring distillable entanglement~\cite{Dur2000}. Hence, the above example shows that \DXPPT\ \rledit{indeed extends PPT states beyond bound entanglement.} 

\rledit{For completeness, we note that the example was found via a brute-force random search over a particular parametrisation of density matrices. 
For $d = d_X d_Y$ we generated $d(d+1)/2 + d(d-1)/2$ independent random real numbers from a normal distribution with mean zero and variance one to define a lower triangular matrix $A$; the first $d(d+1) / 2$ numbers were used to fill the lower triangular part (including the diagonal) of $A$, with zeros on its upper strict triangular part (excluding the diagonal), while the remaining $d(d-1)/2$ numbers were used to define the imaginary part of the strict lower triangular part of $A$. The corresponding random density matrix was then set to
\begin{equation}
	\rho := \frac{A^\dagger A}{\tr[A^\dagger A]} \,.
\end{equation}
For each generated $\rho$ we then computed $\mathcal{T}_X(\rho)$ and $\mathcal{T}_X({\rleditb{\dephchannel}}_{\rho_X}(\rho))$ and the corresponding eigenvalues for each of the three matrices. 
In addition to the setting with $d_X = 3, d_Y = 2$ we also did a search for $d_X =4, d_Y = 2$; in both settings we generated $1000$ random $\rho$. 
For $d_X = 3$ we found $\approx 2.2\%$ examples of non-PPT states that are \DXPPT, for $d_X = 4$ we found $\approx 0.1\%$ examples of such states.}

\subsection{On restricted \JPshort-based inference and virtual broadcasting} \label{app:algorithm} 

As Prop.~\ref{prop:basic-properties-star-products} noted $\star^{JP}$ is bilinear and hermiticity-preserving, but $\mathcal{E} \JPodot (\argline) : \operators{X} \rightarrow \operators{X \otimes Y}$ is generally not positive for CP maps $\mathcal{E}$. 
Thm.~\ref{thm:no-go-linear} then established that for $\rho \in \ASset{\JPshort}{X}{Y}$ and $\mathcal{E} \in \ARCset{\JPshort}{X}{Y}{\rho}$ the map $\mathcal{E} \odot (\argline)$ is in fact positive iff $\rho$ and $\mathcal{E}$ are trivial in the respective senses of Sec.~\ref{sec:general-results-sot-ambiguity}. 

It is interesting to note though that $\mathcal{E} \JPodot (\argline)$ is the linear approximation of $\mathcal{E} \LSodot (\argline)$ when the input state is close to the maximally mixed state~\cite{Parzygnat_2023}. 
In a regime of highly mixed datapoints, the map  $\mathcal{E} \JPodot (\argline)$ may thus still yield physical statistics. 
More generally, given $\rho \in \ASset{\JPshort}{X}{Y}$ (a \DXPPT\ state) and some choice of $\mathcal{E} \in \ARCset{\JPshort}{X}{Y}{\rho}$, one may ask what the size and nature of the subset of input states is that are `amenable to inference', i.e. 
\begin{align}
	IS_{\rho}^{\JPshort} := \left\{ \eta \in \states{X} \mid \mathcal{E} \JPodot \eta \succeq 0 \right\}
\end{align}
		
Now, for any $\mathcal{E} \in \ARCset{\JPshort}{X}{Y}{\rho}$ this map $\mathcal{E}$ is a \emph{channel} and hence $\mathcal{E}\big(\states{X}\big) \subseteq \states{Y}$. 
Yet, for $\eta \in \states{X}$ we regard it only sound to interpret $\mathcal{E}(\eta)$ as inference if also $\mathcal{E} \JPodot \eta \in \states{X \otimes Y}$. 
The reason of course is consistency within our framework, that is, otherwise one would apply `inference' to states, which, had they been included in the training data, the training of the model could not even have processed them due to ill-defined statistics. 
Hence, inference on states outside of $IS_{\rho}^{\JPshort}$ should also be considered ill-defined. 
 
Yet the linearity is of course attractive from the perspective of an efficient implementation. 
So suppose it were to be the case that for some non-separable states $\rho$ the size and nature of $IS_{\rho}^{\JPshort}$ was interesting enough and moreover that it was plausible that a quantum data source $\{ (p_i, \decompindex{\sigma}{i}) \}_{i=1}^k$ came with a promise that all $\decompindex{\sigma}{i} \in IS_{\rho}^{\JPshort}$. 
Then the idea of \emph{inference on restricted inputs only} could in principle be worth considering as an alternative to circumvent the no-go result from Sec.~\ref{sec:general-results-sot-ambiguity}. 
We leave the characterisation of the sets $IS_{\rho}^{\JPshort}$ for future work. 

While the comment in this subsection is mostly intended to be a general one, i.e. noting a way \sot\ maps with linear \infmap{s} might possibly still be relevant,\footnote{For completeness we note that with regards to \MPshort-based models, `restricted inference' is not an interesting route to pursue. 
First, we have that $\ASset{\MPshort}{X}{Y} \subset \ASset{\JPshort}{X}{Y}$. 
Second, for $\rho \in \ASset{\MPshort}{X}{Y}$ and $\mathcal{C} \in \ARCset{\MPshort}{X}{Y}{\rho}$ 
the corresponding subsets of input states $IS_{\rho}^{\MPshort} := \left\{ \eta \in \states{X} \mid \mathcal{C} \MPodot \eta \succeq 0 \right\}$ are precisely those that commute with $D^J(\mathcal{C})$ and hence, for such $\rho$ we also have that  $IS_{\rho}^{\MPshort} \subset IS_{\rho}^{\JPshort}$.} 
it is worth expanding a little on the \JPshort\ case due to links with other parts of this work and due to the key role \JPshort\ plays in the rest of the state-over-time literature \cite{Fullwood_2022,Parzygnat_2023, Lie_2024}. 

Note that for $\rho \in \ASset{\JPshort}{X}{Y}$ and $\mathcal{E} \in \ARCset{\JPshort}{X}{Y}{\rho}$ the map of interest, $\mathcal{E}_{\rho} \JPodot (\argline) \ : \ \operators{X} \rightarrow \operators{X\otimes Y}$, may be rewritten as follows
\begin{align}
	A \ \mapsto \ \mathcal{E}_{\rho} \JPodot A = \big( \idchan_X \otimes \mathcal{E}_{\rho}\big) \Big( \frac{1}{2} \left\{ S, A \otimes \idop_X \Big\} \right) = \big( \idchan_X \otimes \mathcal{E}_{\rho}\big) \circ \mathcal{B} (A) , \label{eq:def-JP-map-using-vir-broadcasting}
\end{align}
where $S$ is the SWAP operator (see Sec.~\ref{sec:data-extension-not-from-inf}) and the map $\mathcal{B} : \operators{X} \rightarrow \operators{X\otimes X}$ has been introduced, which is also called the \emph{virtual broadcasting} map~\cite[Theorem 1]{Parzygnat_2024}. 
The latter is linear, hermiticity- and trace-preserving and hence can be decomposed as the weighted difference of two channels. Indeed,  
\begin{align}
	\mathcal{B} (\argline) \ = \ \frac{d_X+1}{2} \; \mathcal{B}^+ (\argline) \ -  \ \frac{d_X-1}{2} \; \mathcal{B}^- (\argline)
\end{align}
where $\mathcal{B}^\pm$ are the channels as given in Eq.~\eqref{eq:def-channels-B-pm}. 
These two channels, $\mathcal{B}^+$ and $\mathcal{B}^-$, map into the symmetric and antisymmetric subspaces of $\hilb{X} \otimes \hilb{X}$, respectively, and can be realised by the well-known circuit for the swap test (see~\cite{Zheng2025} for an experimental demonstration).
So although $\mathcal{B}$ is not a positive (let alone CP) map, the above linear decomposition allows for the computation (of estimates) of expectation values, that is, provided the ability to execute the `inference channel' $\mathcal{E}_{\rho}$, we can for all $O \in \Herm{X\otimes Y}$ and any $\eta \in IS_{\rho}$ indeed estimate
\begin{align}
	\tr\big[\big(\mathcal{E}_\rho \JPodot \eta \big) \; O\big] \ = \ \frac{d_X + 1}{2} \tr \big[ \big( \mathcal{E}_\rho \circ \mathcal{B}^+ (\eta) \big) \; O \big] \ - \ \frac{d_X - 1}{2} \tr \big[ \big( \mathcal{E}_\rho \circ \mathcal{B}^-  (\eta) \big) \; O \big] . \label{eq:JPmap-via-virtual-broadcasting}  
\end{align}
We close with a few comments. 

First, 
Eq.~\ref{eq:JPmap-via-virtual-broadcasting} looks similar to Eq.~\ref{eq:generic-decomposition-exp-values} in App.~\ref{app:gradient_qre}; in the latter case the map $f$ is \emph{not assumed} to be linear, but \emph{does map} all states to states; in the former case the map \emph{is} linear---making the decomposition much simpler---but the map \emph{is not} generally positive. 

Second, the virtual broadcasting map $\mathcal{B}$ is a \emph{broadcasting} map, i.e. it satisfies $\tr_X \circ \mathcal{B} = \idchan_X$ (for either of the two codomain factors being the one that is traced out). 
It is the kind of linear map that by Prop.~\ref{prop:no-braodcasting-kind-of-fact} is forced to not be CP; and more strongly, by Lem.~\ref{lem:no-braodcasting-kind-of-fact-for-any-positive-map} to not even be positive. It is for this reason that it is commonly referred to as a \emph{virtual} broadcasting map---it could not be a physical map according to quantum theory.  

Third, dropping the $\mathcal{B}^-$ summand in the above equations (and correctly normalising again) leads precisely to the example of a generalised data extension map from Sec.~\ref{sec:data-extension-not-from-inf}, namely  considering the map $\big( \idchan_X \otimes \mathcal{E}_{\rho}\big) \circ \mathcal{B}^+$. 
This move thus trades the existence of an \infmap\ for that now any state, not just from $IS_{\rho}$, can be input and still yields an output state. 
In this context it is interesting to note that $\mathcal{B}^+$ is the channel, i.e. the physical map, that optimally approximates the virtual broadcasting map $\mathcal{B}$ in the diamond distance \cite[Theorem 2]{Parzygnat_2024}. 
We leave it for future work to answer whether it is also the case that the channel $\big( \idchan_X \otimes \mathcal{E}_{\rho}\big) \circ \mathcal{B}^+$ is the optimal approximation to the corresponding \JPshort-based \infmap, $\big( \idchan_X \otimes \mathcal{E}_{\rho}\big) \circ \mathcal{B}$.  

Fourth, if one considered keeping \emph{both} summands, i.e. working with the map in Eq.~\eqref{eq:def-JP-map-using-vir-broadcasting}, but \emph{yet} considered any state $\eta \in \states{X}$ as allowed input, this would constitute a basic deviation from our current framework, where now quantum models are not anymore given by quantum states (i.e. positive operators). 
Such generalisation of the framework---if it is possible to develop such at all---would leave the current scope and is a direction Sec.~\ref{sec:discussion} mentions for future exploration.

\section{Proofs} \label{app:proofs}

\subsection{Proofs Secs.~\ref{subsec:def-data-extension} and \ref{sec:sot-maps}} \label{subsec:earlier-proofs}

\begin{proof}[Proof of Prop.~\protect\ref{prop:no-braodcasting-kind-of-fact}] \label{proof:prop:no-braodcasting-kind-of-fact} 
	Suppose $\mathcal{E} \in \channels{X}{X \otimes Y}$ and such that $\tr_Y \circ \mathcal{E}  = \idchan_X$. 
	Let $\big\{K_i : \hilb{X} \rightarrow \hilb{X \otimes Y} \big\}_{i=1}^k$ (for some $k \in \mathbb{N}$) be a Kraus-representation of $\mathcal{E}$ and $\alpha \in \states{X}$, as well as $\{\ket{l}_X\}_{l=1}^{d_X}$ and $\{\ket{m}_Y\}_{m=1}^{d_Y}$ be ONBs of $\hilb{X}$ and $\hilb{Y}$, respectively. 	
	Then
	\begin{align}
		\tr_Y \circ \mathcal{E}(\alpha) \ = \ \sum_{i=1}^k \tr_Y \big[ K_i \alpha K_i^{\dagger} \big] 
		\ = \ \sum_{i=1}^k \sum_{j=1}^{d_Y}  \ \prescript{}{Y}{\bra{j}} K_i \alpha K_i^{\dagger} \ket{j}_Y \
	\end{align}
	and $\big\{L_{ij} :=  \prescript{}{Y}{\bra{j}} K_i \big\}_{i,j}$ is a Kraus representation of the marginal channel. 
	Since  $\tr_Y \circ \mathcal{E}  = \idchan_X$ and all Kraus representations of a CPTP map are isometrically related to each other, there exist $c_{ij} \in \mathbb{C}$ such that $\sum_{i,j} |c_{ij}|^2 =1$ and $L_{ij} = c_{ij} \idop_X$. 
	One then finds 
	\begin{align}
		 \mathcal{E}(\alpha) & 
		 	\ = \ \sum_{i=1}^k \ K_i \alpha K_i^{\dagger}  
		 	\ = \ \sum_{i=1}^k \sum_{l,l'=1}^{d_X} \sum_{m,m'=1}^{d_Y} \ \dyad{l,m} K_i \alpha K_i^{\dagger} \dyad{l',m'} \\
		 	& \ = \ \sum_{i=1}^k \sum_{l,l'=1}^{d_X} \sum_{m,m'=1}^{d_Y} \ \dyad{l,m}{l} L_{im} \alpha L_{im'}^{\dagger} \dyad{l'}{l',m'} \\
		 	& \ = \ \sum_{l,l'=1}^{d_X} \ \dyad{l}{l} \alpha  \dyad{l'}{l'} \otimes \Big( \sum_{i=1}^k  \sum_{m,m'=1}^{d_Y} c_{im} c^*_{im'} \dyad{m}{m'} \Big)  \ = \ \alpha \otimes \beta
	\end{align}
	for $\beta := \sum_{i=1}^k  \sum_{m,m'=1}^{d_Y} c_{im} c^*_{im'} \dyad{m}{m'}$. Since $\mathcal{E}$ is a channel, $\beta \in \states{Y}$, which concludes the proof.
\end{proof}

\begin{proof}[Proof of Prop.~\ref{prop:basic-properties-star-products}] \label{proof:prop:basic-properties-star-products}
	Concerning (1), the bilinearity of $\star^{\MPshort}$ is evident. 
	For $A,B \in \Herm{H}$ the product $AB$ (equivalently $BA$) is again Hermitian iff $[A,B]=0$; clearly, $\star^{\MPshort}$ is thus not hermiticity-preserving. 
	The claimed properties of $\mathcal{E} \MPodot (\argline)$ for some given $\mathcal{E} \in \channels{X}{Y}$ thus follow immediately, with trace-preservation following straightforwardly from that of  $\mathcal{E}$, since for all $A \in \operators{X}$ we have that 
	\begin{align}
		\tr[ \mathcal{E} \MPodot A ] \ = \  \tr \big[ D^J(\mathcal{E}) A \big]  
 		\ = \ \tr_X \big[ \tr_Y[ D^J(\mathcal{E}) ] A \big] \ = \ \tr(A) \ . \nonumber
	\end{align}
	Analogously for $\star^{\rMPshort}$. 
	
	As for (2), the bilinearity of $\star^{\JPshort}$ is evident. 
	Let $A,B \in \Herm{H}$. While $AB$ or $BA$ are generally not Hermitian, one does find that $(A \star^{JP} B)^{\dagger} = \frac{1}{2} \big((AB)^{\dagger} + (BA)^{\dagger} \big) =  \frac{1}{2} \big(B^{\dagger} A^{\dagger}  + A^{\dagger}B^{\dagger} \big) = \frac{1}{2} \big(BA + AB \big) = A \star^{JP} B$.  
	Suppose $A,B\succeq 0$. If also $[A,B]=0$ then of course $A\star^{\JPshort}B \succeq 0$, but also if $[A,B]\neq 0$ it may be the case that $A\star^{\JPshort}B \succeq 0$ due to the right kind of cancellations in a sum of non-PSD operators. Although maybe obvious enough, in order to see that this cancellation does not generally happen, consider e.g. 
	$A:= \frac{1}{2}\begin{pmatrix} 1 & 1 \\ 1 & 1 \end{pmatrix}$ and 
	$B:= \begin{pmatrix} 1 & 0 \\ 0 & 0 \end{pmatrix}$, 
	both understood as elements of $L(\mathbb{C}^2)$. 
	It is easy to verify that $A$ and $B$ are indeed Hermitian and positive-semidefinite, but $A \star^{JP} B$ has negative eigenvalue $\frac{1 - \sqrt{2}}{4}$. 
	Given $\mathcal{E} \in \channels{X}{Y}$ the claimed properties of $\mathcal{E} \JPodot (\argline)$ are then immediate with the trace-preservation property straightforward by the same reasons as for (1): for all $A \in \operators{X}$ we have that 
	\begin{align}
		\tr[ \mathcal{E} \JPodot A ] \ = \ \frac{1}{2} \tr \big[ D^J(\mathcal{E}) A + A D^J(\mathcal{E}) \big]  
		\ = \ \tr \big[ D^J(\mathcal{E}) A\big] \ = \ \tr(A) \ . \nonumber
	\end{align}

	Finally, concerning (3), the linearity of $\star^{\LSshort}$ in the first argument and non-linearity in the second is obvious. 
	That $B^{\frac{1}{2}}AB^{\frac{1}{2}} \in \Herm{H}$  whenever $A,B \in \Herm{H}$ is straightforward and 
	that $B^{\frac{1}{2}}AB^{\frac{1}{2}} \succeq 0$ whenever $A,B \succeq 0$ holds by Lem.~\ref{lem:psd-corrected}. 
	Given $\mathcal{E} \in \channels{X}{Y}$ the claimed properties of $\mathcal{E} \LSodot (\argline)$ are then implied---noting $D^J(\mathcal{E})$ is always Hermitian and again noting (both directions of) Lem.~\ref{lem:psd-corrected}--- with the trace-preservation property following straightforwardly just as for (1) and (2): for all $A \in \operators{X}$
	\begin{align}
		\tr[ \mathcal{E} \LSodot A ] \ = \ \tr[A^{1/2} D^J(\mathcal{E}) A^{1/2}] \ = \ \tr[D^J(\mathcal{E}) A] 
\ = \ \tr(A) \ . \nonumber
	\end{align}
\end{proof}

\subsection{Proofs Sec.~\ref{sec:general-results-sot-ambiguity}} \label{subsec:proofs-general-results-sot-ambiguity}

\begin{proof}[Proof of Lem.~\protect\ref{lem:general-fact-product-states}] \label{proof:lem:general-fact-product-states} 
	Suppose $\odot : \cpmaps{X}{Y} \times \operators{X} \rightarrow \operators{X \otimes Y}$ is an \sot\ map. 
	Let $\alpha \in \states{X}$ and $\beta \in \states{Y}$ then 
	\begin{align}
		\alpha \otimes \beta & \ = \ (\idop_X \otimes \beta) (\alpha \otimes \idop_Y) \ = \ D^J(\beta \circ \tr_X) (\alpha \otimes \idop_Y) 
		\ = \ (\beta \circ \tr_X) \odot \alpha \ , \label{eq:triviality-lemma-eq-2}
	\end{align}		
	where the last step holds by (C) of Def.~\ref{def:sot-map} seeing as there evidently is a product ONB wrt which $D^J(\beta \circ \tr_X)$ and $(\alpha \otimes \idop_Y)$ are simultaneously diagonal. 
	This establishes that $\alpha \otimes \beta \in \ASset{\sotmap}{X}{Y}$, and hence the last claim of Lem.~\protect\ref{lem:general-fact-product-states}. 
	
	Note that the same steps also establish (1) $\Rightarrow$ (2). 
	
	That (2) $\Rightarrow$ (3) is straightforward, again by essentially the same steps as in Eq.~\ref{eq:triviality-lemma-eq-2} just read backwards. In more detail, suppose $(2)$, i.e. $\rho_Y \circ \tr_X \in \ARCset{\sotmap}{X}{Y}{\rho}$, then for $\alpha \in \states{X}$ we have $(\rho_Y \circ \tr_X) \odot \alpha = D^J(\rho_Y \circ \tr_X) (\alpha \otimes \idop_Y) = (\idop_X \otimes \rho_Y) (\alpha \otimes \idop_Y) = \alpha \otimes \rho_Y$ and therefore, since it holds for all states $\alpha$, we have that $(\rho_Y \circ \tr_X) \odot (\argline) = \idchan_X \otimes \rho_Y$. 
	
	Finally, (3) $\Rightarrow$ (1) is evident, since by assumption $\rho = (\rho_Y \circ \tr_X) \odot (\rho_X) = \rho_X \otimes \rho_Y$. 
\end{proof}

\subsection{Proofs Sec.~\ref{subsec:useful-facts}} \label{subsec:proofs-useful-facts}

\begin{proof}[Proof of Lem.~\protect\ref{lem:psd-corrected}] \label{proof:lem:psd-corrected} 
Let $A,B \in \operators{X}$ with $A \succeq 0$. 
First suppose also $B \succeq 0$. Let $\ket{v} \in \hilb{X}$ and set $\ket{v'} = A^{1/2}\ket{v}$. 
By assumption then, and seeing as $A^{1/2}$ is Hermitian, $\bra{v'} B \ket{v'} = \bra{v}A^{1/2} B A^{1/2}\ket{v} \succeq 0$. 
Conversely, suppose $A^{1/2} B A^{1/2}  \succeq 0$. 
Then there exists $C \in \operators{X}$ such that $A^{1/2} B A^{1/2} = C^\dag C$; 
and multiplying both sides by $A^{-1/2}$ we obtain $\pi_AB\pi_A = A^{-1/2} C^\dag C A^{-1/2} = (C A^{-1/2})^\dag(C A^{-1/2}) $ and thus $\pi_AB\pi_A \succeq 0$.
\end{proof}

\begin{proof}[Proof of Lem.~\ref{lem:transpose-commutativity}] \label{proof:lem:transpose-commutativity}
	Let $A \in \operators{X}$ be Hermitian and $\{\ket{i} \}_{i=1}^{d_X}$ an eigenbasis of $A$, writing $A= \sum_i \alpha_i \dyad{i}{i}$; let $\mathcal{T}_X$ be wrt that eigenbasis and let $\{\ket{k} \}_{k=1}^{d_Y}$ be some ONB of $\hilb{Y}$. 
	Then for any $ B \in \operators{X \otimes Y}$ 
	\begin{align}
		\mathcal{T}_X \circ \Phi_{f(A)} (B) \ 
		& = \ \mathcal{T}_X \Big( \left(f(A) \otimes \idop_Y\right) B \left(f(A)^\dag \otimes \idop_Y\right) \Big) \\
		&= \ \mathcal{T}_X \Big( \Big(\sum_i f(\alpha_i) \dyad{i}{i} \otimes  \sum_k  \dyad{k}{k} \Big) B \Big(\sum_j f(\alpha_j)^* \dyad{j}{j} \otimes \sum_l  \dyad{l}{l} \Big) \Big) \\
		&= \ \sum_{i,j,k,l} f(\alpha_i) f(\alpha_j)^* \; \mathcal{T}_X \Big( \bra{i,k} B \ket{j,l} \ \dyad{i,k}{j,l}  \Big) \\
		&= \ \sum_{i,j,k,l} f(\alpha_i) f(\alpha_j)^*  \bra{j,k} B \ket{i,l} \ \dyad{i,k}{j,l} \\ 
		&= \ \sum_{i,j,k,l} f(\alpha_i) f(\alpha_j)^*  \bra{i,k} \mathcal{T}_X \big( B \big) \ket{j,l} \ \dyad{i,k}{j,l} \\
		&= \ \Big(\sum_i f(\alpha_i) \dyad{i}{i} \otimes  \sum_k  \dyad{k}{k} \Big) \mathcal{T}_X \big( B \big) \Big(\sum_j f(\alpha_j)^* \dyad{j}{j} \otimes \sum_l  \dyad{l}{l} \Big) \\ 
		& = \ \Phi_{f(A)} \circ \mathcal{T}_X (B)
	\end{align}
\end{proof}

\begin{proof}[Proof of Lem.~\ref{lem:support-inclusion-property}] \label{proof:lem:support-inclusion-property}
	We first prove the claim for positive rank-1 operators.  
	Let $\ket{\Psi} \in \hilb{X} \otimes \hilb{Y}$ and let $\ket{\Psi} = \sum_{i \in M} c_i \ket{i,i}$ denote a Schmidt-decomposition of it, where $M$ is finite, $c_i \neq 0$ for all $i \in M$ and $\ket{i,i} \in \hilb{X} \otimes \hilb{Y}$ are orthonormal vectors. 
	Then 
	\begin{align}
		\dyad{\Psi}{\Psi} \ & = \  \sum_{i,j \in M} c_i c_j^* \dyad{i,i}{j,j} \\
		 \tr_{Y} \left(\dyad{\Psi}{\Psi}\right) \ & = \  \sum_{i \in M} |c_i|^2 \dyad{i}{i} \\
		 \tr_{Y} \left(\dyad{\Psi}{\Psi}\right) \otimes \idop_Y \ & = \ \sum_{i \in M} \sum_{j \in \overline{M}} |c_i|^2 \dyad{i,j}{i,j} 
	\end{align}	 
	where $\overline{M}$ extends $M$ such that $\{\ket{i}\}_{i\in \overline{M}}$ is an ONB of $\hilb{Y}$. 
	Clearly, $supp(\dyad{\Psi}{\Psi}) = span(\{\ket{\Psi}\}) \subseteq span(\{ \ket{i,j} | i\in M, j\in \overline{M}\}) = supp(\tr_{Y} \left(\dyad{\Psi}{\Psi}\right) \otimes \idop_Y)$. 
	
	Next note that for $A,B \in \operators{Z}$ for any finite-dimensional $\hilb{Z}$ and such that $A,B \succeq 0$ it holds that 
	\begin{align}
		supp(A + B) & = ker(A+B)^{\perp} \stackrel{(a)}{=} \big( ker(A) \cap ker(B) \big)^{\perp} 
		= span \big( ker(A)^{\perp} \cup ker(B)^{\perp} \big) \label{eq:supp-equation} \\ 
		& = span \big( supp(A) \cup supp(B) \big) \nonumber 
	\end{align}
	where (a) is easy to see given the positive semi-definiteness of $A$ and $B$.\footnote{Clearly, $ker(A+B) \supseteq ker(A) \cap ker(B)$; for the converse note $A+B \succeq 0$ and for $\ket{\psi} \in \hilb{Z}$ it holds $0 = \bra{\psi} (A+B) \ket{\psi} = \bra{\psi} A  \ket{\psi} + \bra{\psi} B \ket{\psi}$ iff both summands on the RHS are zero and thus $ker(A+B) \subseteq ker(A) \cap ker(B)$.} 

	Now let $\rho \in \operators{X \otimes Y}$ such that $\rho \succeq 0$ and let $\rho = \sum_k p_k \dyad{\phi_k}{\phi_k}$ be an eigendecomposition, where $\{\ket{\phi_k}\}_k$ denotes an ONB of $\hilb{X} \otimes \hilb{Y}$ and $p_k \geq 0$ for all $k$. 
	Using the orthonormality of the ONB vectors, the result on the above rank-1 case, and Eq.~\eqref{eq:supp-equation} one finds
	\begin{align}
		supp(\rho) & = span\Big(\cup_k supp(p_k \dyad{\phi_k}{\phi_k})\Big) 
			\subseteq span\Big(\cup_k supp\big(p_k \tr_Y\left(\dyad{\phi_k}{\phi_k}\right) \otimes \idop_Y \big) \Big) \\ 
			& = supp\Big( \sum_k p_k \tr_Y\left(\dyad{\phi_k}{\phi_k} \right) \otimes \idop_Y \Big) 
			= supp\big( \rho_X \otimes \idop_Y \big)
	\end{align} 
\end{proof}

\subsection{Proofs Sec.~\ref{subsec:MP-ambiguity}} \label{subsec:proofs-MP-ambiguity}

\begin{lemma} \label{lem:MPodot-lemma} 
Given an \MPshort-ambiguous state, $\rho \in \ASset{\MPshort}{X}{Y}$, then for all $\mathcal{E} \in \ARCset{\MPshort}{X}{Y}{\rho}$ it holds that 
\begin{align}
	[D^J(\mathcal{E}), \rho_X] = 0 \quad \text{and} \quad D^J(\mathcal{E} \circ \Pi_X) \succeq 0 \ ,
\end{align}
where $\pi_X$ is with respect to $\rho_X$.
\end{lemma}
\begin{proof} 
Let $\rho \in \ASset{\MPshort}{X}{Y}$ and $\mathcal{E} \in \ARCset{\MPshort}{X}{Y}{\rho}$, i.e. $\rho = \mathcal{E} \MPodot \rho_X = D^J(\mathcal{E}) \rho_X$. 
Recall that the Jamio{\l}kowski-representation of any CP map, in particular $D^J(\mathcal{E})$, is Hermitian. 
Thus the hermiticity of $\rho$, and therefore that of $D^J(\mathcal{E}) \rho_X$, implies that $[D^J(\mathcal{E}), \rho_X]=0$. 
Further, the positivity of $\rho$ together with $[D^J(\mathcal{E}), \rho_X]=0$ imply that 
$\pi_XD^J(\mathcal{E}) \pi_X \succeq 0$, which is equivalent to $D^J(\mathcal{E} \circ \Pi_X) \succeq 0$. 
\end{proof}

\begin{proof}[Proof of Thm. ~\ref{thm:OPA-equivalence}] \label{proof:thm:OPA-equivalence}	\

	``(1) $\Leftrightarrow$ (2)": Given the commutation condition observed in Lem.~\ref{lem:MPodot-lemma} it is immediate that $\ASset{\MPshort}{X}{Y} \subseteq \ASset{\rMPshort}{X}{Y}$. Analogous arguments give the opposite inclusion.  
	
	``(1) $\Rightarrow$ (3)": 
	Suppose $\rho \in \ASset{\MPshort}{X}{Y}$ and let $\mathcal{E} \in \ARCset{\MPshort}{X}{Y}{\rho}$, i.e. $\rho = D^J(\mathcal{E}) \rho_X$. 
	Lem.~\ref{lem:MPodot-lemma} tells us that $[D^J(\mathcal{E}), \rho_X]=0$ and $D^J(\mathcal{E} \circ \Pi_X) \succeq 0$. 
	First, it is immediate that then also $[\rho,\rho_X ]=0$ and that $\rho = D^J(\mathcal{E} \circ \Pi_X) \rho_X$, as well as $[D^J(\mathcal{E} \circ \Pi_X), \rho_X] = 0$. 
	Second, letting $\mathcal{T}_X$ be with respect to an eigenbasis of $\rho_X$ it now follows from Prop.~\ref{prop:positive-DJ-equiv-to-PPT} 	that $\exists \eta \in \states{X \otimes Y}$ such that 
	$D^J(\mathcal{E} \circ \Pi_X) = \ N \mathcal{T}_X \big( \eta \big) $ for a suitable $N > 0$. 	 
	Hence, 
	\begin{align}
		\mathcal{T}_X\big( \rho \big) 
			\ = \ N \mathcal{T}_X \Big( \mathcal{T}_X \big( \eta \big) \rho_X \Big) 
			\ = \ N \rho_X^{1/2} \mathcal{T}_X\Big( \mathcal{T}_X \big( \eta \big) \Big) \rho_X^{1/2} 
			\ = \ N	\rho_X^{1/2} \eta \rho_X^{1/2} \ \succeq \ 0
	\end{align} 
	where the second equality uses $[\mathcal{T}_X \big( \eta \big), \rho_X^{1/2}]=0$ and Lem.~\ref{lem:transpose-commutativity}, and the last step Lem.~\ref{lem:psd-corrected}. 
	This establishes that $\rho$ is \CXPPT.
	
	``(1) $\Leftarrow$ (3)": Suppose $\rho \in \states{X \otimes Y}$ is \CXPPT, i.e. PPT and it holds that $[\rho,\rho_X ]=0$. 
	Let $\mathcal{T}_X$ be with respect to an eigenbasis of $\rho_X$.  
	Set $r:= rank(\rho_X)$ and $\eta := \frac{1}{r} \rho_X^{-1/2} \mathcal{T}_X(\rho) \rho_X^{-1/2}$. 
	By assumption and Lem.~\ref{lem:psd-corrected} we have that $\eta \succeq 0$ and it is easy to check (using Prop.~\ref{prop:partial-transpose-trace-properties}) that $\tr(\eta)=1$, hence  $\eta  \in \states{X \otimes Y}$. 
	It is also clear that $\eta$ is PPT (Lem.~\ref{lem:transpose-commutativity} and Lem.~\ref{lem:psd-corrected}) with $\eta = (1/r) \pi_X$. 
	Then 
	\begin{align}
		\rho 
		\ & \stackrel{\text{\tiny Lem.~\ref{lem:support-inclusion-property}}}{=} \ \rho  \pi_X 
		\ = \ \rho \rho_X^{-1} \rho_X 
		\ = \ \rho_X^{-1/2} \rho \rho_X^{-1/2} \rho_X \\
		 & = \ r \frac{1}{r} \rho_X^{-1/2} \mathcal{T}_X\big(\mathcal{T}_X(\rho) \big) \rho_X^{-1/2} \rho_X
		\  \stackrel{\text{\tiny Lem.~\ref{lem:transpose-commutativity}}}{=} \ r \mathcal{T}_X \Big( \frac{1}{r} \rho_X^{-1/2} \mathcal{T}_X(\rho) \rho_X^{-1/2} \Big) \rho_X \\
		& = \ r \mathcal{T}_X \Big( \eta \Big) \rho_X
	\end{align}
	Noting Prop.~\ref{prop:positive-DJ-equiv-to-PPT}, $r \mathcal{T}_X \Big( \eta \Big) = D^J(\mathcal{E})$ for some $\mathcal{E} \in \cpmaps{X}{Y}$. By analogous arguments as in the proof of Lem.~\ref{lem:MPodot-lemma}, we have $\mathcal{E} = \mathcal{E} \circ \Pi_X$. Clearly, for any $\mathcal{E'} \in \channels{X}{Y}$, $\widetilde{\mathcal{E}} := \mathcal{E} + \mathcal{E}' \circ \Pi_X^{\perp} \in \channels{X}{Y}$ and it holds $\rho = D^J(\widetilde{\mathcal{E}}) \rho_X$, which completes the proof.
\end{proof}

\begin{proof}[Proof of Prop.~\ref{prop:unique-CP-map-for-OP}] \label{proof:prop:unique-CP-map-for-OP}
	Let $\rho \in \ASset{\MPshort}{X}{Y}$ and $\mathcal{E}, \mathcal{E}' \in \ARCset{\MPshort}{X}{Y}{\rho}$. 
	By assumption, after multiplication with $\rho_X^{-1}$ (from the right), we have 
	\begin{align}
		D^J(\mathcal{E})\pi_X \ = \ \rho \rho_X^{-1} \ = \ D^J(\mathcal{E}')\pi_X \ . \label{eq:OP-uniqueness-eq1}
	\end{align}
	Note that (see proof of Lem.~\ref{lem:MPodot-lemma}) $[D^J(\mathcal{E}), \rho_X] = 0$ and $[D^J(\mathcal{E}'), \rho_X] = 0$ and hence also $[D^J(\mathcal{E}), \pi_X] = 0$ and $[D^J(\mathcal{E}'), \pi_X] = 0$. 
	Thus by multiplying Eq.~\eqref{eq:OP-uniqueness-eq1} with $\pi_X$ from the left we find 
	\[
		D^J(\mathcal{E} \circ \Pi_X) = \pi_X D^J(\mathcal{E})\pi_X = \pi_X D^J(\mathcal{E}')\pi_X = D^J(\mathcal{E}' \circ \Pi_X)	
	\]
	and hence $\mathcal{E} \circ \Pi_X = \mathcal{E}' \circ \Pi_X$.  
	Defining $\MPCPmap{\rho} := \mathcal{E} \circ \Pi_X $, the above establishes (1).  
	(In particular, clearly for any $\mathcal{E}' \in \channels{X}{Y}$ we have $\MPCPmap{\rho} + \mathcal{E}' \circ \Pi_X^{\perp} \in \ARCset{\MPshort}{X}{Y}{\rho}$.) 
	The above also gives $D^J(\MPCPmap{\rho}) = \rho \rho_X^{-1}$, which is claim (2); 
	claim (3) follows then from $[\rho, \rho_X] = 0$ (Thm.~\ref{thm:OPA-equivalence}). 	
\end{proof}

\subsection{Proofs Sec.~\ref{subsec:LSA-ambiguity}} \label{subsec:proofs-LSA-ambiguity}

For the sake of an analogous presentation compared with \MPshort\ and \JPshort, Sec.~\ref{subsec:LSA-ambiguity} presented first Thm.~\ref{thm:LSA-PPT-theorem} and then Prop.\ref{prop:basic-LSA-propeties}, while the proofs here are given in the reverse order, simply as it makes them more straightforward.

\begin{proof}[Proof of Prop.\ref{prop:basic-LSA-propeties}] \label{proof:prop:basic-LSA-propeties}
	Let $\rho \in \ASset{\LSshort}{X}{Y}$ and $\mathcal{E}, \mathcal{E}' \in \ARCset{\LSshort}{X}{Y}{\rho}$. 
	Then  
	\begin{align}
		0 \ & = \ D^J(\mathcal{E}) \star^{LS} \rho_X \ - \ D^J(\mathcal{E}') \star^{LS} \rho_X \\
			& = \ \rho_X^{1/2} \left( D^J(\mathcal{E}) - D^J(\mathcal{E}') \right) \rho_X^{1/2} \ . 
	\end{align}
	Multiplication with $\rho_X^{-1/2}$ from both sides yields $0 = \pi_X \left( D^J(\mathcal{E}) - D^J(\mathcal{E}') \right)  \pi_X$ and hence, $\pi_X D^J(\mathcal{E}) \pi_X = \pi_X D^J(\mathcal{E}') \pi_X$, 	
	which can easily be seen to be equivalent to $\mathcal{E} \circ \Pi_X =\mathcal{E}' \circ \Pi_X$. 
	Define $\LSCPmap{\rho} := \mathcal{E} \circ \Pi_X$. 
	Since for any $\mathcal{E}' \in \channels{X}{Y}$ we have 
	\[ 
		\rho_X^{1/2} \left( D^J(\LSCPmap{\rho}) + \pi_X^{\perp} D^J(\mathcal{E}') \pi_X^{\perp} \right) \rho_X^{1/2} = \rho_X^{1/2} D^J(\LSCPmap{\rho}) \rho_X^{1/2} = \rho
	\] 
	we find that $\LSCPmap{\rho} + \mathcal{E}' \circ \Pi_X^{\perp} \in \ARCset{\LSshort}{X}{Y}{\rho}$.	
	This establishes (1). 
	Given (1) and noting that in particular $\LSCPmap{\rho} = \LSCPmap{\rho} \circ \Pi_X$, (2) is immediate from multiplying $\rho = \LSCPmap{\rho} \LSodot \rho_X$ with $\rho_X^{-1/2}$ on both sides. 
	Finally, (3) is immediate from (2) together with Lem.~\ref{lem:psd-corrected}. 
\end{proof}

\begin{proof}[Proof of Thm.~\ref{thm:LSA-PPT-theorem}] \label{proof:thm:LSA-PPT-theorem}
Let $\rho \in \states{X \otimes Y}$.\\ 
	\noindent ``$\Rightarrow$": Suppose $\rho \in \ASset{\LSshort}{X}{Y}$, then with Prop.\ref{prop:basic-LSA-propeties} we have that $\rho = \LSCPmap{\rho} \LSodot \rho_X = \rho_X^{1/2} \ D^J(\LSCPmap{\rho}) \ \rho_X^{1/2}$.  
	Let $\mathcal{T}_X$ denote the partial transpose with respect to an eigenbasis of $\rho_X$ and let $D^C(\LSCPmap{\rho})$ be the Choi state with respect to that same basis. Then 
	\begin{align}
		\mathcal{T}_X \left(\rho \right) & \ = \ \mathcal{T}_X \left(  \rho_X^{1/2} \ D^J(\mathcal{E}) \ \rho_X^{1/2}  \right) 
		\ \stackrel{\text{\tiny Lem.~\ref{lem:transpose-commutativity}}}{=} \ \rho_X^{1/2}  \ \mathcal{T}_X \left( D^J(\mathcal{E}) \right) \ \rho_X^{1/2} \\ 
		& \ = \ \rho_X^{1/2}  \ D^C(\mathcal{E}) \ \rho_X^{1/2} 
		\ \stackrel{\text{{\tiny Lem.~\ref{lem:psd-corrected}}}}{\succeq} \ 0
	\end{align}
\noindent ``$\Leftarrow$": Let $\mathcal{T}_X$ again denote the partial transpose with respect to an eigenbasis of $\rho_X$ and suppose $\mathcal{T}_X \left(\rho \right) \succeq 0$. 
	Define the following operator
	\begin{align}
		\gamma := \rho_X^{-1/2} \ \mathcal{T}_X \left( \rho  \right) \ \rho_X^{-1/2}  \label{eq:def-gamma}
	\end{align}
	{Due to Lem.~\ref{lem:psd-corrected}} we have $\gamma \succeq 0$ and it thus defines a $\mathcal{C} \in CP(X,Y)$ via $D^C(\mathcal{C}) := \gamma$, which is trace-non-increasing:
	\begin{align}
		\tr_Y \left( D^C(\mathcal{C}) \right) & 
		\ = \ \tr_Y \left( \rho_X^{-1/2}  \ \mathcal{T}_X(\rho) \ \rho_X^{-1/2} \right) 
		\ \stackrel{\text{\tiny Lem.~\ref{lem:transpose-commutativity}}}{=} \ \tr_Y \left( \mathcal{T}_X \left( \rho_X^{-1/2} \ \rho \ \rho_X^{-1/2}  \right) \right) \\
		& \ \stackrel{\text{\tiny Prop.~\ref{prop:partial-transpose-trace-properties}}}{=} \ \mathcal{T}_X \left( \tr_Y \left( \rho_X^{-1/2}  \ \rho \ \rho_X^{-1/2}  \right) \right)  
		\ = \  \mathcal{T}_X \left( \rho_X^{-1/2} \rho_X \ \rho_X^{-1/2} \right) 
		\ = \ \mathcal{T}_X \left( \pi_X \right) \\
		 &  \ = \ \pi_X
	\end{align}
	Note that $\gamma  = \pi_X \gamma \pi_X $ and thus $\mathcal{C} = \mathcal{C} \circ \Pi_X$. 
	Let $\widetilde{\mathcal{C}} \in \channels{X}{Y}$ be such that 
	$\mathcal{C} = \widetilde{\mathcal{C}} \circ \Pi_X$, i.e. 
	$\gamma = \pi_X \ D^C(\widetilde{\mathcal{C}}) \ \pi_X$  
	(e.g. $ \widetilde{\mathcal{C}} := \mathcal{C} + \mathcal{C}' \circ \Pi_X^{\perp}$ for any $\mathcal{C}' \in \channels{X}{Y}$).
	Now we find
	\begin{align}
		 \widetilde{\mathcal{C}} \LSodot \rho_X 
			& \ = \ \rho_X^{1/2}  \ D^J(\widetilde{\mathcal{C}}) \ \rho_X^{1/2}  
			\ = \ \rho_X^{1/2}  \ \pi_X \ D^J(\widetilde{\mathcal{C}}) \ \pi_X \ \rho_X^{1/2} \\
			& \stackrel{ }{=} \ \rho_X^{1/2} \ \pi_X \ \mathcal{T}_X \left( D^C(\widetilde{\mathcal{C}}) \right) \ \pi_X \ \rho_X^{1/2}  \\
			& \stackrel{\text{\tiny Lem.~\ref{lem:transpose-commutativity}}}{=} \ \rho_X^{1/2}  \  \mathcal{T}_X \left( \pi_X \ D^C(\widetilde{\mathcal{C}}) \ \pi_X \right) \ \rho_X^{1/2} 
			\ =  \ \rho_X^{1/2}  \  \mathcal{T}_X \left( \gamma \right) \ \rho_X^{1/2}  \\
			& \stackrel{\text{\tiny Eq.~\eqref{eq:def-gamma}}}{=} \ \rho_X^{1/2}  \ \mathcal{T}_X \left( \rho_X^{-1/2} \ \mathcal{T}_X \left( \rho  \right) \ \rho_X^{-1/2}  \right) \ \rho_X^{1/2}  \\
			& \stackrel{\text{\tiny Lem.~\ref{lem:transpose-commutativity}}}{=} \ \rho_X^{1/2}  \ \rho_X^{-1/2}  \ \mathcal{T}_X \left(  \mathcal{T}_X \left( \rho  \right) \right) \ \rho_X^{-1/2}  \ \rho_X^{1/2} 
			\ = \pi_X  \ \rho \ \pi_X \ = \ \rho
	\end{align}
	Thus $\rho \in \ASset{\LSshort}{X}{Y}$. 
\end{proof}

\subsection{Proofs Sec.~\ref{subsec:Jordan-product-based-ambiguity}} \label{subsec:proofs-Jordan-product-based-ambiguity}

\begin{proposition}	 \label{prop:F-rho}
	Let $\rho \in \states{X \otimes Y}$, then 
	\begin{align}
	F_{\rho} \ := \ \int_{0}^{\infty} d\tau \ e^{-1/2 \rho_X \tau} \ \rho \ e^{-1/2 \rho_X \tau}
\end{align}
	satisfies the following properties: 
	(1) $\rho \ = \ F_{\rho} \ \star^{JP} \ \rho_X $; 
	(2) $F_{\rho} \succeq 0$; 
	(3) $\tr_Y\left( F_{\rho} \right) = \pi_X$. 
\end{proposition}
\begin{proof} 
	First, $F_{\rho}$ is the known solution to condition (1) seen as a Lyapunov equation as is straightforward to verify: 
	\begin{align}
		F_{\rho} \ \star^{JP} \ \rho_X \ & = \ \frac{1}{2} \left(F_{\rho} \rho_X \ + \ \rho_X F_{\rho} \right) \\
			& = \ \frac{1}{2} \int_{0}^{\infty} d\tau \ \left(  e^{-1/2 \rho_X \tau} \ \rho \ e^{-1/2 \rho_X \tau}  \rho_X  \ + \ \rho_X  e^{-1/2 \rho_X \tau} \ \rho \ e^{-1/2 \rho_X \tau} \right) \\ 
			& = \ - \int_{0}^{\infty} d\tau \ \frac{d}{d\tau} \left( e^{-1/2 \rho_X \tau} \ \rho \ e^{-1/2 \rho_X \tau} \right) \\
			& = \ - \left. e^{-1/2 \rho_X \tau} \ \rho \ e^{-1/2 \rho_X \tau} \right|_{0}^{\infty} 
			\ = \ \pi_X \ \rho \ \pi_X \ \stackrel{\text{\tiny Lem.~\ref{lem:support-inclusion-property}}}{=} \ \rho 
	\end{align} 
	(2) is immediate from Lem.~\ref{lem:psd-corrected}. 
	(3) is a straightforward calculation:
	\begin{align}
		\tr_Y \left( F_{\rho} \right) \ & = \ \tr_Y \left( \int_{0}^{\infty} d\tau \ e^{-1/2 \rho_X \tau} \ \rho \ e^{-1/2 \rho_X \tau} \right) \\
		& = \ \int_{0}^{\infty} d\tau \  \tr_Y \left( e^{-1/2 \rho_X \tau} \ \rho \ e^{-1/2 \rho_X \tau}  \right) \\
		& = \ \int_{0}^{\infty} d\tau \ e^{-1/2 \rho_X \tau} \ \rho_X \ e^{-1/2 \rho_X \tau} 
		\ = \ \int_{0}^{\infty} d\tau \ e^{-\rho_X \tau} \ \rho_X  \\
		& = \ \left. - e^{-\rho_X \tau} \right|_{0}^{\infty} \ = \ e^{-\rho_X 0} \ = \ \pi_X
	\end{align}
\end{proof}

\begin{proof}[Proof of Prop.~\ref{prop:channel-D}] \label{proof:prop:channel-D}
		In order to establish that ${\rleditb{\dephchannel}}_{\eta}$ is a CP map, the following will be useful. For each $i=1,...,r$ set $z_i := \ln(p_i)$ and observe that 
	\begin{align}
		\frac{2\sqrt{p_k p_l}}{p_k + p_l} = \frac{2e^{\frac{z_k + z_l}{2}}}{e^{z_k} + e^{z_l}} = \frac{2}{e^{\frac{z_k - z_l}{2}} + e^{-\frac{z_k - z_l}{2}}} = \sech( \frac{z_k - z_l}{2} ) 
	\end{align}	 
	and that the hyperbolic secant is a fixed point of the Fourier transform, in particular,
	\begin{align}
		\sech( \frac{\omega}{2} ) = \int_{-\infty}^{\infty} \sech(\pi t) e^{i t \omega} dt .
	\end{align}
	Let $A \in \operators{X}$. 
	\begin{align}
		{\rleditb{\dephchannel}}_{\eta} (A) \ & = \ \sum_{k,l=1}^{r} \ \frac{2 \sqrt{p_k p_l}}{p_k + p_l} \  \dyad{k}{k} A \dyad{l}{l}\\
		& = \ \sum_{k,l=1}^{r} \ \left( \int_{-\infty}^{\infty} dt \sech(\pi t) e^{i t (z_k - z_l)} \right)  \ \dyad{k}{k} A \dyad{l}{l} \\
		& = \ \sum_{k,l=1}^{d_X} \ \int_{-\infty}^{\infty} dt \sqrt{\sech(\pi t)} \ p_k^{it} \dyad{k}{k} A \dyad{l}{l} p_l^{-it} \sqrt{\sech(\pi t)} \\
		& = \ \sum_{k,l=1}^{d_X} \ \int_{-\infty}^{\infty} dt \sqrt{\sech(\pi t)} \ \eta^{it} \dyad{k}{k} A \dyad{l}{l} \eta^{-it} \sqrt{\sech(\pi t)} \\
		& = \ \int_{-\infty}^{\infty} dt \ \sqrt{\sech(\pi t)} \ \eta^{it} \ A \ \eta^{-it} \sqrt{\sech(\pi t)} \\
		& = \ \int_{-\infty}^{\infty} dt \ K_t \ A \ K_t^{\dagger}
	\end{align}
	for the Kraus-operators $K_t := \sqrt{\sech(\pi t)} \ \eta^{it}$. 
	Also note
	\begin{align}
		 \int_{-\infty}^{\infty} dt \  \ K_i^{\dagger} K_t \ & = \ \int_{-\infty}^{\infty} dt \ \sech(\pi t)  \ \eta^{-it} \eta^{it} \ = \ \left( \int_{-\infty}^{\infty} dt \ \sech(\pi t) \right) \pi_X \ = \  \pi_X
	\end{align}
	where $\pi_X$ is wrt $\eta$. 
	Hence, ${\rleditb{\dephchannel}}_{\eta}$ is indeed a (trace-non-increasing) CP map. 
\end{proof}

\begin{lemma} \label{lem:F-rho-via-D}
	Let $\rho \in \states{X \otimes Y}$ then 
	\begin{align}
		F_{\rho} \ = \  \rho_X^{-1/2} \ {\rleditb{\dephchannel}}_{\rho_X} (\rho) \  \rho_X^{-1/2} \label{eq:F-rho-via-D}
	\end{align}
\end{lemma}
\begin{proof} 
	Let $\rho \in \states{X \otimes Y}$ and write $\rho_X = \sum_i p_i \dyad{i}{i}$ for an eigendecomposition of $\rho_X$ (with $r=rank(\rho_X)$ and wlog assuming that the support of $\rho_X$ is spanned by the first $r$ basis vectors). 
\begin{align}
	F_{\rho} \ 
		& = \ \int_{0}^{\infty} d\tau \ e^{-1/2 \rho_X \tau} \rho \ e^{-1/2 \rho_X \tau} \\
		& = \ \sum_{i,j=1}^{d_X} \ \int_{0}^{\infty} d\tau \ e^{-1/2 \rho_X \tau} \ \big( \dyad{i}{i} \otimes \idop_Y \big) \rho \big( \dyad{j}{j} \otimes \idop_Y \big) \ e^{-1/2 \rho_X \tau} \label{eq:Frho-calc-a-2} \\
		& = \ \sum_{i,j=1}^r \ \Big( \int_{0}^{\infty} d\tau \ e^{-1/2 (p_i + p_j) \tau} \Big) \big( \dyad{i}{i} \otimes \idop_Y \big) \rho \ \big( \dyad{j}{j} \otimes \idop_Y \big) \label{eq:Frho-calc-a-3} \\ 
		& = \ \sum_{i,j=1}^r \ \frac{2}{p_i + p_j} \ \big( \dyad{i}{i} \otimes \idop_Y \big) \rho \big( \dyad{j}{j} \otimes \idop_Y \big) \\
		& = \ \sum_{i,j=1}^r \ \frac{2\sqrt{p_i p_j}}{p_i + p_j} \  p_i^{-1/2}  \big( \dyad{i}{i} \otimes \idop_Y \big) \rho \big( \dyad{j}{j} \otimes \idop_Y \big) p_j^{-1/2} \\
		& = \ \rho_X^{-1/2} \sum_{i,j=1}^r \ \frac{2\sqrt{p_i p_j}}{p_i + p_j} \  \big( \dyad{i}{i} \otimes \idop_Y \big) \rho \big( \dyad{j}{j} \otimes \idop_Y \big) \rho_X^{-1/2} \\
		& = \ \rho_X^{-1/2} \ {\rleditb{\dephchannel}}_{\rho_X} (\rho) \  \rho_X^{-1/2} 
\end{align}  
\end{proof}

\begin{proof}[Proof of Prop.~\ref{prop:JPA-channel-uniqueness}] \label{proof:prop:JPA-channel-uniqueness} 
Let $\rho \in \ASset{\JPshort}{X}{Y}$ and $\mathcal{E} \in \ARCset{\JPshort}{X}{Y}{\rho}$. 
Using $\rho = \pi_X \rho \pi_X $ (see Lem.~\ref{lem:support-inclusion-property}) we find
\begin{align}
	\rho & \ = \ D^J(\mathcal{E}) \star^{JP} \rho_X \ = \ \pi_X \rho \pi_X 
	 \ = \ \frac{1}{2} \left( \pi_X D^J(\mathcal{E}) \rho_X \ + \ \rho_X D^J(\mathcal{E}) \pi_X  \right) \\
	& = \ \frac{1}{2} \left( \pi_X D^J(\mathcal{E}) \pi_X \rho_X \ + \ \rho_X \pi_X D^J(\mathcal{E}) \pi_X  \right) 
	\ = \ \pi_X D^J(\mathcal{E}) \pi_X \ \star^{JP} \rho_X 
\end{align}	
Note that $\pi_X D^J(\mathcal{E}) \pi_X =  D^J(\mathcal{E} \circ \Pi_X)$ and define 
$\JPCPmap{\rho} := \mathcal{E} \circ \Pi_X \in \cpmaps{X}{Y}$. 
Now the uniqueness of the solution of the Lyapunov equation---the Hurwitz-stability criterion amounts to precisely restricting to the support of $\rho_X$ here---implies that $D^J(\JPCPmap{\rho})$ is that unique solution and hence $\forall \mathcal{C} \in \ARCset{\JPshort}{X}{Y}{\rho}$ we have that $\mathcal{C} \circ \Pi_X = \JPCPmap{\rho}$. 
Moreover since for any $\mathcal{E}' \in \channels{X}{Y}$ we have $D^J(\mathcal{E}' \circ \Pi_X^{\perp}) = \pi_X^{\perp}D^J(\mathcal{E}') \pi_X^{\perp}$, by the above steps it is hence also obvious that $\JPCPmap{\rho} + \mathcal{E}' \circ \Pi_X^{\perp} \in \ARCset{\JPshort}{X}{Y}{\rho}$. 
This establishes (1). 
This uniqueness of $D^J(\JPCPmap{\rho})$ as a solution of $\rho \ = \ F_{\rho} \ \star^{JP} \ \rho_X $, together with Lem.~\ref{lem:F-rho-via-D}, thus yields (2). 
(3) is immediate from (2) due to Lem.~\ref{lem:psd-corrected}. 
\end{proof}

\begin{lemma} \label{lem:JPA-via-CP-F}
	For $\rho \in \states{X \otimes Y}$ the following are equivalent:
	\begin{enumerate}[label=(\arabic*), leftmargin=2cm]
		\item $\rho \in \ASset{\JPshort}{X}{Y}$.
		\item $\exists \mathcal{E} \in \cpmaps{X}{Y} \ \ \text{such that} \ \ D^J(\mathcal{E}) = F_{\rho}$. 
		\item $\mathcal{T}_X \big( F_{\rho} \big) \succeq 0$ with $\mathcal{T}_X$ wrt any ONB of $\hilb{X}$. 
	\end{enumerate}
\end{lemma}
\begin{proof} 
	Let $\rho \in \states{X \otimes Y}$ \\	
	``(1) $\Leftarrow$ (2)": This follows essentially from Prop.~\ref{prop:F-rho} together with the fact that the CP map can be extended suitably to a CPTP map. 
	In more detail, let $\mathcal{E} \in \cpmaps{X}{Y}$ such that $D^J(\mathcal{E}) = F_{\rho}$. 
	The analogous steps as in the proof of (1) of Prop.~\ref{prop:JPA-channel-uniqueness} yield (due to the uniqueness argument) $\pi_X F_{\rho} \pi_X = F_{\rho}$ and thus $D^J(\mathcal{E}) = D^J(\mathcal{E} \circ \Pi_X)$, i.e  $\mathcal{E} = \mathcal{E} \circ \Pi_X$.  
	Hence, for any $\mathcal{E}' \in \channels{X}{Y}$ setting $\widetilde{\mathcal{E}} := \mathcal{E}  + \mathcal{E}' \circ \Pi_X^{\perp}$ we have that $\widetilde{\mathcal{E}} \in \channels{X}{Y}$. 
	It is easy to check that $\pi_X D^J(\widetilde{\mathcal{E}}) = D^J( \mathcal{E}) = D^J(\widetilde{\mathcal{E}}) \pi_X $ and hence we find (using $\rho_X = \rho_X \pi_X =  \pi_X \rho_X$):
	\begin{align}
		D^J(\widetilde{\mathcal{E}}) \ \star^{JP} \ \rho_X \ 
			& = \ \frac{1}{2} \left( D^J(\widetilde{\mathcal{E}}) \rho_X \ + \ \rho_X D^J(\widetilde{\mathcal{E}}) \right) 
			\ = \ \frac{1}{2} \left(D^J(\widetilde{\mathcal{E}}) \pi_X  \rho_X \ + \ \rho_X \pi_X D^J(\widetilde{\mathcal{E}}) \right)	\\
			& = \ \frac{1}{2} \left(D^J(\mathcal{E}) \rho_X \ + \ \rho_X D^J(\mathcal{E}) \right) 
			\ = \ D^J(\mathcal{E}) \ \star^{JP} \ \rho_X \ = \ \rho
	\end{align}
	So (1) indeed holds for $\rho$. \\
	``(1) $\Rightarrow$ (2)": Immediate from Prop.~\ref{prop:JPA-channel-uniqueness} and Prop.~\ref{prop:F-rho}.   \\	
	``(2) $\Rightarrow$ (3)": Suppose (2) holds, then for any choice of ONB of $\hilb{X}$ we have $\mathcal{T}_X \big( F_{\rho} \big) = \mathcal{T}_X \big( D^J(\mathcal{E}) \big) = D^C(\mathcal{E}) \succeq 0$ due to Choi's theorem. \\
	``(2) $\Leftarrow$ (3)": Suppose $\mathcal{T}_X \big( F_{\rho} \big) \succeq 0$ then $\exists \mathcal{E} \in \cpmaps{X}{Y}\ \ \text{such that} \ \ D^C(\mathcal{E}) = \mathcal{T}_X \big( F_{\rho} \big)$ and thus  $D^J(\mathcal{E}) = F_{\rho}$. 
\end{proof}

\begin{lemma}\label{lem:LSA-implies-JPA}
	For $\rho \in \states{X \otimes Y}$ if $\rho \in \ASset{\LSshort}{X}{Y}$ then $\rho \in \ASset{\JPshort}{X}{Y}$, i.e. \LSshort-ambiguity implies \JPshort-ambiguity; 
	and if indeed $\rho \in \ASset{\LSshort}{X}{Y}$ then the respective unique CP maps satisfy $\JPCPmap{\rho} = \LSCPmap{\rho} \circ {\rleditb{\dephchannel}}_{\rho_X}$.~\footnote{Fun fact: one can show that in this case $D^J(\mathcal{E}^{LS}) = {\rleditb{\dephchannel}}_{\rho_X}(\rho) \starprod^{JP} \rho_X^{-1}$, while Prop.~\ref{prop:JPA-channel-uniqueness}(2) gives $D^J(\JPCPmap{\rho}) \ = {\rleditb{\dephchannel}}_{\rho_X}(\rho) \starprod^{LS}  \rho_X^{-1}$. That is, the Jamio{\l}kowski representation for one type of ambiguity can be defined in terms of the star product of the other type of ambiguity.}
\end{lemma}
\begin{proof} 
	Suppose $\rho \in \ASset{\LSshort}{X}{Y}$, hence $\rho = \rho_X^{1/2} \ D^J(\LSCPmap{\rho}) \ \rho_X^{1/2}$. Inserting the latter into the RHS of Eq.~\eqref{eq:F-rho-via-D} it is a straightforward calculation that finds $F_{\rho} = D^J(\LSCPmap{\rho} \circ {\rleditb{\dephchannel}}_{\rho_X})$. 
With $\LSCPmap{\rho}$ and ${\rleditb{\dephchannel}}_{\rho_X}$ both CP it follows that $\rho$ is also \JPshort-ambiguous  (Lem.~\ref{lem:JPA-via-CP-F}). 
Then it also holds that $\JPCPmap{\rho} = \LSCPmap{\rho} \circ {\rleditb{\dephchannel}}_{\rho_X}$ (see Lem.~\ref{lem:JPA-via-CP-F} and Prop.~\ref{prop:JPA-channel-uniqueness}). 
\end{proof}

\begin{proof}[Proof of Thm.~\ref{thm:JPA-equivalence-with-DPPT}] \label{proof:thm:JPA-equivalence-with-DPPT}
	Let $\rho \in \states{X \otimes Y}$. 
	From Lem.~\ref{lem:F-rho-via-D} and Lem.~\ref{lem:transpose-commutativity} it follows that 
	\begin{align}
		\mathcal{T}_X \big( F_{\rho} \big) = \rho_X^{-1/2} \ \mathcal{T}_X \circ {\rleditb{\dephchannel}}_{\rho_X} (\rho) \ \rho_X^{-1/2}
	\end{align}
	The fact that $\mathcal{T}_X \big( F_{\rho} \big) \succeq 0$ is equivalent to $\rho$ being \JPshort-ambiguous  (Lem.~\ref{lem:JPA-via-CP-F}) then concludes the proof observing Lem.~\ref{lem:psd-corrected} (with one direction literally given by Lem.~\ref{lem:psd-corrected} and the converse holds due to $\pi_X \big( \mathcal{T}_X \circ {\rleditb{\dephchannel}}_{\rho_X} (\rho) \big) \pi_X = \mathcal{T}_X \circ {\rleditb{\dephchannel}}_{\rho_X} (\rho) $). 
\end{proof}

\subsection{Proofs Sec.~\ref{sec:main-section-on-LS-models}} \label{subsec:proofs-main-section-on-LS-models}

\begin{proof}[Proof of Lem.~\protect\ref{lem:LS-positive-extended-map}] \label{proof:lem:LS-positive-extended-map} 
		Let $\rho \in \ASset{\LSshort}{X}{Y}$ and $\mathcal{C} \in \ARCsetpos{\LSshort}{X}{Y}{\rho}$. 
		The positivity of $\mathcal{C} \LSodot (\argline)$ follows from the definition of $\ARCsetpos{\LSshort}{X}{Y}{\rho}$, i.e. that $D^J(\mathcal{C}) \succeq 0$ (essentially Prop.~\ref{prop:basic-LSA-propeties}), together with Prop.~\ref{prop:basic-properties-star-products}, or explicitly,
	\[
		\mathcal{C} \LSodot \eta \ = \  \eta^{1/2} D^J(\mathcal{C}) \eta^{1/2} \ \stackrel{{\tiny {Lem.~\ref{lem:psd-corrected}}}}{=} \  \succeq 0 \ .
	\]
	The PPTness of the output state is also straightforward, since $\forall \eta \in \states{X}$
	\[
		\mathcal{T}_X \big( \mathcal{C} \LSodot \eta \big) \ = \  \mathcal{T}_X \big( \eta^{1/2} D^J(\mathcal{C}) \eta^{1/2} \big) \ \stackrel{{\tiny Lem.~\ref{lem:transpose-commutativity}}}{=} \ \eta^{1/2}  \mathcal{T}_X \big( D^J(\mathcal{C}) \big) \eta^{1/2}  \ = \ \eta^{1/2} D^C(\mathcal{C})  \eta^{1/2} \ \stackrel{{\tiny {Lem.~\ref{lem:psd-corrected}}}}{\succeq } 0 \ .
	\]
\end{proof}

\end{document}